\documentclass[a4paper,11pt]{article}
\usepackage[utf8]{inputenc}
\usepackage[margin = 1in]{geometry}
\pdfoutput=1

\providecommand{\keywords}[1]
{\small\textbf{Keywords ---} #1}
\usepackage[nobysame]{amsrefs}

\usepackage[english]{babel}
\usepackage[T1]{fontenc}
\usepackage{graphicx}
\usepackage{color}
\usepackage{booktabs}
\usepackage{caption}
\usepackage{subcaption}
\usepackage{enumitem}
\usepackage{multirow}
\usepackage{float}
\usepackage[toc,page]{appendix}
\usepackage{xcolor}
\usepackage[automake,nonumberlist,nopostdot]{glossaries}
\usepackage{blindtext}
\usepackage{authblk}
\usepackage[ruled,vlined]{algorithm2e}
 \usepackage{placeins}
\usepackage{amsmath}
\usepackage{amsfonts}
\usepackage{amssymb}
\usepackage{amsthm}
\usepackage{bm}
\usepackage{bbm}
\usepackage{verbatim}
\usepackage{algorithmic}
\usepackage{todonotes}
\usepackage{hyperref}

\theoremstyle{plain}
 
\newtheorem{lemma}{Lemma}[section]

\usepackage{mathtools}
\def\sqrtexplained#1{%
	\begingroup
	\sbox0{$#1$}
	\def\underbrace##1_##2{##1}
	\sbox2{$#1$}
	\dimen0=\wd0 \advance\dimen0-\wd2
	\mathrlap{\sqrt{\phantom{\displaystyle#1}\kern\dimen0 }}
	\hphantom{\sqrt{\vphantom{\displaystyle#1}}}
	\endgroup
	#1}
\newcommand{\order}[1]{\mathcal{O}\left(#1\right)}
\DeclareMathOperator*{\expectation}{\mathbb{E}}
\newcommand{\expec}[1]{\expectation\left[#1\right]}
\newcommand{\mse}[1]{\text{MSE}\left(#1\right)}
\newcommand{\cost}[1]{\text{Cost}\left(#1\right)}

\newcommand{\norm}[1]{\left\|#1\right\|}

\DeclareMathOperator*{\esssup}{ess\,sup}

\DeclareMathOperator*{\argmin}{arg\,min}

\newcommand{\qoi}{\ensuremath{Q}}
\newcommand{\linf}{\mathrm{L}^{\infty}}
\newcommand{\ltwo}{\mathrm{L}^{2}}
\newcommand{\slinf}{l^{\infty}}
\newcommand{\sltwo}{l^{2}}

\newcommand{\measure}{\ensuremath{\mathbb{P}}}

\newcommand{\setR}{\ensuremath{\mathbb{R}}}
\newcommand{\setN}{\ensuremath{\mathbb{N}}}
\makeatletter
\newcommand{\abbrie}{i.e\@ifnextchar.{}{.\@}}
\newcommand{\abbreg}{e.g\@ifnextchar.{}{.\@}}
\newcommand{\abbrae}{a.e\@ifnextchar.{}{.\@}}
\newcommand{\abbrst}{s.t\@ifnextchar.{}{.\@}}
\makeatother
\newcommand{\interpm}[1]{\mathcal{S}_n^{(m)} \left(#1\right)}

 \newcommand{\interp}[1]{\mathcal{S}_n \left(#1\right)}

\newcommand{\estmlmc}{\hat{\mu}}

\newcommand{\thetab}{\bm\theta}
\newcommand{\Phiem}{\hat{\Phi}_L^{(m)}}
\newcommand{\Phieone}{\hat{\Phi}_L^{(1)}}
\newcommand{\Phietwo}{\hat{\Phi}_L^{(2)}}

\newcommand{\Phie}{\hat{\Phi}_L}

\newcommand{\Phim}{\Phi^{(m)}}
\newcommand{\Phione}{\Phi^{(1)}}

\newcommand{\quant}{q_{\tau}}
\newcommand{\quantest}{\hat{q}_{\tau}}
\newcommand{\cvar}{c_{\tau}}
\newcommand{\cvarest}{\hat{c}_{\tau}}
\newcommand{\stat}{s_{\tau}}
\newcommand{\statest}{\hat{s}_{\tau}}
\newcommand{\costsim}{\mathfrak{c}_l}
\newcommand{\costfunc}{\mathfrak{c}_{\phi}}
\newcommand{\costint}{\mathfrak{c}_{\text{int}}}
\newcommand{\err}{\text{e}}
\newcommand{\indicator}{\mathbbm{1}}
\newcommand{\errest}{\hat{\err}}
\newcommand{\errbound}{\bar{\err}}
\DeclareMathOperator\erf{\text{erf}}

\newtheorem{proposition}{Proposition}[section]

\SetKw{And}{and}
\SetKw{Or}{or}
\SetKw{Then}{then}
\SetKw{Elif}{elif}
\SetKw{Not}{not}
\SetKw{Update}{update}
\SetKw{Compute}{compute}
\SetKw{Estimate}{estimate}
\SetKwComment{Mycomment}{}{}
\SetCommentSty{textit}
\SetKwProg{Fn}{}{}{}
\SetKwSwitch{Switch}{Case}{Other}{switch}{do}{case}{otherwise}{endcase}{endsw}

\newacronym{qoi}{QoI}{Quantity of Interest}
\newacronym{pde}{PDE}{Partial Differential Equation}
\newacronym{kde}{KDE}{Kernel Density Estimation}
\newacronym{mlmc}{MLMC}{Multi-Level Monte Carlo}
\newacronym{mc}{MC}{Monte Carlo}
\newacronym{cmlmc}{CMLMC}{Continuation Multi-Level Monte Carlo}
\newacronym{mimc}{MIMC}{Multi-Index Monte Carlo}
\newacronym{mfmc}{MFMC}{Multi-Fidelity Monte Carlo}
\newacronym{mse}{MSE}{Mean Squared Error}
\newacronym{cdf}{CDF}{Cumulative Distribution Function}
\newacronym{pdf}{PDF}{Probability Density Function}
\newacronym{cvar}{CVaR}{Conditional-Value-at-Risk}
\newacronym{var}{VaR}{Value-at-Risk}
\newacronym{ouu}{OUU}{Optimization Under Uncertainty}
\newacronym{uq}{UQ}{Uncertainty Quantification}
\newacronym{sde}{SDE}{Stochastic Differential Equation}
\makeglossaries

\title{Quantifying uncertain system outputs via the multi-level Monte Carlo method --- distribution and robustness measures}

\author[1]{Quentin Ayoul-Guilmard\thanks{quentin.ayoul-guilmard@epfl.ch}}
\author[1]{Sundar Ganesh\thanks{sundar.ganesh@epfl.ch}}
\author[2]{Sebastian Krumscheid\thanks{krumscheid@uq.rwth-aachen.de}}
\author[1]{Fabio Nobile\thanks{fabio.nobile@epfl.ch, corresponding author}}
\affil[1]{\small{Institute of Mathematics, {\'E}cole Polytechnique F{\'e}d{\'e}rale de Lausanne, Switzerland}}
\affil[2]{\small{Department of Mathematics, Rheinisch-Westf\"alische Technische Hochschule, Aachen, Germany}}

\begin{document}
\maketitle
\vspace{-0.25in}

\begin{abstract}
In this work, we consider the problem of estimating the probability distribution, the quantile or the conditional expectation above the quantile, the so called \gls{cvar}, of output quantities of complex random differential models by the \gls{mlmc} method.
We follow the approach of \cite{krumscheid2017multilevel}, which recasts the estimation of the above quantities to the computation of suitable parametric expectations. 
In this work, we present novel computable error estimators for the estimation of such quantities, which are then used to optimally tune the \gls{mlmc} hierarchy in a continuation type adaptive algorithm. 
We demonstrate the efficiency and robustness of our adaptive continuation-\gls{mlmc} in an array of numerical test cases.
\end{abstract}

\noindent \keywords{Multilevel Monte Carlo Methods, Value-at-Risk, Conditional-Value-at-Risk, Uncertainty Quantification, Kernel Smoothing, Bootstrap Sampling.}

\section{Introduction}\label{sec:introduction}
Complex differential models are used in many disciplines across science and engineering as predictive or design tools.
More often than not, however, some input parameters of these models are uncertain, either due to missing information, lack of proper characterization or intrinsic variability.
It is hence of utmost interest to study and quantify the effects of these uncertainties on an output \gls{qoi} of the model, or several \glsplural{qoi}, which are in turn used for prediction or design.
When uncertainty is modeled as randomness in a probabilistic framework, each \gls{qoi} becomes a random variable and its distribution is often inaccessible in closed form.
We assume here however that the \gls{qoi} can be simulated in an approximate way, typically by sampling the random input parameters and computing the solution of a suitable discretisation of the underlying differential model.
It is therefore of great practical interest to estimate by simulation, and with controlled accuracy, the distribution of the \gls{qoi} or some summary statistics such as central moments of different orders, a quantile of a given significance, alternatively known as the \gls{var}, or super quantiles such as the so-called \gls{cvar}, which is often used as a risk-measure in stochastic optimisation problems applied to finance \cite{rockafellar2002conditional, uryasev2001conditional}.

Solving the underlying differential model at a desired accuracy typically has a high computational cost, even for a single realisation of the random input.
An accurate estimation of the summary statistics of a \gls{qoi} by a direct Monte Carlo approach is often prohibitively expensive.
As an alternative to Monte Carlo sampling, one may consider stochastic collocation or polynomial chaos techniques, which, for certain problems, feature a much faster convergence rate than Monte Carlo approaches \cite{ghanem2003stochastic, le2010spectral, babuvska2007stochastic, xiu2002wiener}.
For applications of polynomial chaos methods to the approximation of density/distribution functions or quantiles, we mention \cite{sagiv2021spectral}.
However, the effectiveness of these techniques is restricted to problems with a relatively small intrinsic dimensionality of the input uncertainty and high smoothness of the input-output map, which are features we do not necessarily assume in this work.


%

In this work, we focus instead on \gls{mlmc} methods since they perform well for high-dimensional input uncertainties in comparison to the methods described above.
\gls{mlmc} methods as introduced in the works \cite{giles2008multilevel} and \cite{heinrich2001multilevel} are a well established technology to compute the expected value of a \gls{qoi} $\qoi$ that is an output quantity of a stochastic differential model.
They are typically useful when one cannot sample directly from $\qoi$ but rather, when inexact sampling is possible from a set of approximations $\{\qoi_{l} \}_{l=0}^L$ to $\qoi$ on a sequence of $L+1$ discretisations with different characteristic discretisation parameters, for example induced by mesh sizes $h_0 > h_1 > ... > h_L$, typically a geometric sequence $h_{l-1}=s h_l$ with $s>1$.

Let $(\Omega,\mathcal{F},\measure)$ denote a complete probability space, $\omega \in \Omega$ an elementary random event and $\qoi:\Omega \to \setR$ a real valued \gls{qoi}.
The \gls{mlmc} estimator to estimate the expected value $\mu = \expec{\qoi}$ of the \gls{qoi} is given by
\begin{align}
\estmlmc &\coloneqq \frac{1}{N_0}\sum_{i=1}^{N_0} \qoi_{0}^{(i,0)} +\sum_{l=1}^L \frac{1}{N_l} \sum_{i=1}^{N_l} \left[\qoi_{l}^{(i,l)} - \qoi_{{l-1}}^{(i,l)} \right], \label{eq:mlmc_expec}
\end{align}
where $\qoi_{l}^{(i,l)}$ and $\qoi_{{l-1}}^{(i,l)}$ are correlated realisations of the \gls{qoi} computed  with the same underlying realisation of the input parameters on meshes with discretisation parameters $h_l$ and $h_{l-1}$ respectively and $\{N_l\}_{l=0}^{L}$ is a decreasing sequence of sample sizes.
Notice that the sum over the discretisation levels $l=0,...,L$ telescopes in expectation:
\begin{align}
\expec{\estmlmc} = \expec{\qoi_{0}} + \sum_{l=1}^L \expec{\qoi_{l} - \qoi_{{l-1}}} = \expec{\qoi_{L}} \approx \expec{\qoi},
\end{align}
and hence, the bias error $|\expec{\estmlmc - \qoi}| = |\expec{\qoi_L - \qoi}|$ of the \gls{mlmc} estimator depends only on the finest discretisation level considered.
It was shown in \cite{giles2008multilevel} that if the parameters $L$ and $\{N_l\}_{l=0}^L$, hereafter called the \gls{mlmc} ``hierarchy'', are chosen appropriately, the \gls{mlmc} estimator can in theory achieve a dramatic performance improvement over a simple Monte Carlo estimator in terms of computational cost versus accuracy.
From the practical perspective, finding good \gls{mlmc} hierarchies that realize the theoretical performance improvement requires sharp and reliable error estimators for the bias error and statistical error contributions of the estimator, which are due to discretisation and finite sampling respectively.
The parameters of the hierarchy can then be selected based on these error estimates to provide the best possible performance.
This selection process is commonly referred to as ``tuning''.
In the case of the simple estimator of the expected value of the \gls{qoi} introduced in Eq.~\eqref{eq:mlmc_expec}, these errors can be estimated by sample average and sample variance estimators using the same samples used to compute the \gls{mlmc} estimator $\estmlmc$ itself.
The reader is referred to \cite{giles2015multilevel, collier2014continuation, pisaroni2016continuation} for detailed descriptions of different \gls{mlmc} algorithms that estimate these errors and adaptively tune the parameters of the \gls{mlmc} hierarchy based on them.
It is important to note that appropriate mesh convergence is necessary to achieve the superior performance of \gls{mlmc} methods over Monte Carlo estimators.

The application of \gls{mlmc} methods to estimate statistics other than the expected value is less developed, as well as the problem of error estimation and adaptive tuning of the \gls{mlmc} hierarchy.
In our companion ``Part I'' paper \cite{Pisaroni_mlmcPart1_pre}, we have presented \gls{mlmc} estimators for the computation of higher order central moments and detailed practical algorithms to optimally tune the \gls{mlmc} hierarchy.
In this paper, we focus instead on the estimation of the \gls{cdf} and the \gls{pdf}, as well as the \gls{var} and the \gls{cvar} of a given significance $\tau$.
Particularly, we follow the approach proposed in \cite{krumscheid2017multilevel}, which consists of introducing suitable parametric expectations, and deriving the sought after statistics as a post-processing step.
Parametric expectations are expectations of the form
\begin{align}
\Phi(\theta) \coloneqq \expec{\phi(\theta,\qoi)}\label{eq:parest_form}.
\end{align}
In this work, we follow \cite{krumscheid2017multilevel} and use the following particular form for the function $\phi$:
\begin{align}
\phi(\theta, \qoi) &\coloneqq \theta + \frac{1}{1-\tau}(\qoi-\theta)^+, \hspace{1em} X^+ \coloneqq \max(0,X), \quad \theta \in \Theta \subset \setR,  \label{eq:phi_form}
\end{align}
where $\tau \in (0,1)$ denotes a significance parameter and $\Theta$ denotes a suitable interval of interest.
This form has the advantage that after estimating the function $\Phi$ and its derivatives
\begin{align}
\Phi^{(m)}(\theta) \coloneqq \frac{\partial^m}{\partial \theta^m}\expec{\phi(\theta,\qoi)},\quad m \in \setN,\label{eq:phi_derivatives}
\end{align}
the \gls{cdf} $F_{\qoi}(\theta) = \expec{\indicator_{\qoi \geq \theta}}$ and the \gls{pdf} $f_{\qoi}(\theta) = F_{\qoi}^{(1)}(\theta)$ over the interval $\Theta$, as well as the \gls{var} $\quant$ and the \gls{cvar} $\cvar$ of any significance $\tau$ for which $\quant \in \Theta$, can be obtained by simple post-processing:
\begin{equation}
\begin{aligned}
F_{\qoi}(\theta) &= \tau + (1-\tau) \Phi^{(1)}, & \qquad\quant &= \argmin_{\theta \in \Theta}\Phi(\theta),\\
 f_{\qoi}(\theta) &= (1-\tau) \Phi^{(2)}, &\qquad  \cvar &= \min_{\theta \in \Theta}\Phi(\theta) = \Phi(\quant).
\end{aligned}\label{eq:stat_defs}
\end{equation}
On the notation, we comment that $\Phi^{(0)} = \Phi$, and that $\indicator_{\qoi \geq \theta}$ denotes the characteristic function which takes on a value of $1$ in the interval denoted by the subscript and $0$ everywhere else.

We remark that the \gls{cdf} could also be estimated by direct \gls{mlmc} estimation of the expectations $F_{\qoi}(\theta) = \expec{\indicator_{\qoi \geq \theta}}$ for different values of $\theta$.
However, using \gls{mlmc} to estimate the expected value of a discontinuous function can lead to samples from the fine and coarse discretisation levels for the same random input being on either side of the discontinuity.
This can result in the problem no longer satisfying the conditions necessary for the superior performance of \gls{mlmc} over Monte Carlo methods \cite{avikainen2009irregular,giles2009analysing}.
The parametric expectation approach overcomes this problem, since the function $\phi$ in Eq.~\eqref{eq:phi_form} is Lipschitz continuous in $\qoi$ for all $\theta \in \Theta$.

We briefly review alternative approaches that have been proposed in literature to use \gls{mlmc} methods for estimating the distribution of a \gls{qoi}.
A \gls{mlmc} estimator for the \gls{cdf} was proposed and analysed in \cite{giles2015multilevelfunc} wherein a smoothened approximation to the characteristic function $\indicator_{\qoi \geq \theta}$ was used.
The \gls{mlmc} method was also used in \cite{giles2019multilevel} for nested conditional expectations from which the \gls{var} and \gls{cvar} could be derived.
An alternative smoothing of the characteristic function based on the \gls{kde} method was proposed in \cite{taverniers2020estimation}, combined with an \gls{mlmc} estimator wherein stratification based sampling was applied at each level.
The authors of \cite{bayer2020multilevel} combined an approach to locate the discontinuity using a root-finding algorithm, followed by numerical pre-integration.
One can also derive the \gls{var} and the \gls{cvar} from surrogate distributions derived from moments.
For example, in the works \cite{bierig2016approximation,gou2016estimating}, a maximum entropy approach was used to estimate the \gls{pdf} and the \gls{var} using moment estimates from \gls{mlmc} estimators.
The use of \gls{mlmc} estimators for parametric expectations is still an ongoing research area.
The work in \cite{krumscheid2017multilevel} built upon the ideas presented in \cite{giles2015multilevelfunc}, but generalises them further to approximate general parametric expectations.
Furthermore, novel \gls{mlmc} estimators for the characteristic function were presented based on the idea of pointwise estimation combined with interpolation.

The current work builds upon the theoretical work in \cite{krumscheid2017multilevel} and aims at deriving practical algorithms for the \gls{mlmc} estimators proposed therein.
This requires the derivation of reliable and possibly sharp error estimators that can be used to adaptively calibrate the hierarchy of the \gls{mlmc} estimators to achieve optimal performance, i.e., a computational complexity aligned with the theoretical predictions.
An error bound was already presented in \cite{krumscheid2017multilevel} based on the use of inverse inequalities.
However, this bound results in conservative error estimates that lead to \gls{mlmc} hierarchies that are impractically expensive to compute.
In this work, we propose novel error estimators which are much sharper than those reported previously and can be used for practical engineering purposes.
These estimators improve on the large leading constants while preserving their optimal theoretical decay rates as the discretisation is refined.
More precisely, we propose a bias error estimator based on a smoothened density as well as a statistical error estimator based on bootstrapping \cite{tibshirani1993introduction}. 
We then use our novel error estimators to design a continuation \gls{mlmc} algorithm that successively improves the hierarchy to meet a target tolerance with optimal performance. 
We show on three numerical tests, including an option pricing problem in finance and a laminar fluid dynamics problem, that our methodology does indeed feature a computational complexity aligned with the theoretical rates presented in \cite{krumscheid2017multilevel}.
Furthermore, we demonstrate that the methodology is robust in the sense that the true \gls{mse}, computed with respect to a reference solution, is always smaller than the prescribed tolerance.
We add that the novel \gls{mlmc} contributions of this work have been implemented in the Python package XMC, available at \cite{ExaQUte_XMC}.

The structure of this work is as follows.
In Section~\ref{sec:mlmc_par_est}, we present the \gls{mlmc} estimator for the parametric expectation $\Phi$ in Eq.~\eqref{eq:parest_form} and introduce a notion of the \gls{mse} for $\Phi$ and its derivatives.
We briefly recall the results of \cite{krumscheid2017multilevel} on error bounds for \gls{mlmc} estimators of parametric expectations and present a simplified complexity result for an optimally tuned \gls{mlmc} estimator.
We also detail the practical aspects of implementing such an error estimator.
In Section~\ref{sec:novel}, we describe novel error estimators that provide tighter bounds on the true error.
We compare the performance of these error estimators with the a priori ones presented in \cite{krumscheid2017multilevel} on a simple case for which theoretical results are known.
Section~\ref{sec:adaptivity} details an adaptive strategy for selecting the parameters of the hierarchy such that a given tolerance can be achieved on the \gls{mse} of the \gls{mlmc} estimator of $\Phi^{(m)},\; m \in \{0,1,2\}$, as well as on the \gls{mse} of \gls{mlmc} estimators of the \gls{var} and the \gls{cvar}.
In particular, Section~\ref{sec:var_cvar_error_bound} recalls a result from \cite{krumscheid2017multilevel} to relate the error on derived quantities such as the \gls{var} and the \gls{cvar} to the error on $\Phi$ and its derivatives.
The novel error estimator, the adaptive strategy and the performance of the \gls{mlmc} algorithm are demonstrated on an array of problems of increasing complexity in Section~\ref{sec:results}.
Finally, Section~\ref{sec:conclusion} offers a conclusion and a discussion on the presented work.

\section{Multi-Level Monte Carlo approximation of parametric expectations}\label{sec:mlmc_par_est}
As presented in Section~\ref{sec:introduction}, we focus in this paper on the problem of approximating parametric expectations of the form in Eqs.~\eqref{eq:phi_form} and~\eqref{eq:phi_derivatives} using the \gls{mlmc} method. 
The approach we follow is motivated by \cite{krumscheid2017multilevel,giles2015multilevelfunc}.
We approximate the parametric expectation $\Phi$ and its derivatives $\Phim$ on an interval $\Theta$ via the \gls{mlmc} method as follows:
We first consider a set of $n \in \setN$ nodes 
\begin{align}
\thetab \coloneqq \{\theta_1, \theta_2, ..., \theta_n\}, \quad \theta_j \in \Theta \subset \setR , \quad 1 \leq j \leq n, \quad \theta_j < \theta_{j+1},
\end{align}
such that $\Theta = [\theta_1, \theta_n]$.
The function $\Phi$ is then approximated pointwise at any point $\theta_j \in \Theta$ as
\begin{align}
\Phi(\theta_j) \approx \expec{\phi(\theta_j, \qoi_{L})} = \expec{\phi(\theta_j, \qoi_{0})} + \sum_{l=1}^L \expec{\phi(\theta_j, \qoi_{l})-\phi(\theta_j, \qoi_{{l-1}})},
\end{align}
where each expected value is estimated using a Monte Carlo estimator. We then define the MLMC estimator $\Phie(\theta_j)$ of $\Phi(\theta_j)$ as
\begin{align}
\Phie(\theta_j) &\coloneqq \frac{1}{N_0}\sum_{i=1}^{N_0} \phi(\theta_j, \qoi_{0}^{(i,0)}) +\sum_{l=1}^L \frac{1}{N_l} \sum_{i=1}^{N_l} \left[\phi(\theta_j, \qoi_{l}^{(i,l)}) - \phi(\theta_j, \qoi_{{l-1}}^{(i,l)}) \right]. \label{eq:mlmc_est_1}
\end{align}
It is important to note that the same set of random events is used to evaluate the estimator for all $\theta_j$.
Finally, we obtain a \gls{mlmc} estimator $\Phie$ of the whole function $\Phi:\Theta \to \setR$ by interpolating over the pointwise estimates as below:
\begin{align}
\Phie = \interp{\Phie(\thetab)},
\end{align}
where $\mathcal{S}_n$ denotes an appropriate interpolation operator and $\Phie(\thetab)$ denotes the set of pointwise \gls{mlmc} estimates in Eq.~\eqref{eq:mlmc_est_1}, that is:
\begin{align}
\Phie(\thetab) = \{\Phie(\theta_1), \Phie(\theta_2), \dots, \Phie(\theta_n)\}.
\end{align}
An estimate of the function derivative of order $m \in \setN$ denoted by $\Phiem$ is then obtained by computing the derivative of the resultant interpolated function:
\begin{align}
\Phiem &\coloneqq \interpm{ \Phie(\thetab) } \coloneqq \frac{\partial^m}{\partial \theta^m} \interp{ \Phie(\thetab) },\label{eq:mlmc_est_2}
\end{align}
provided that it exists.
Throughout this work, cubic spline interpolation with equally spaced interpolation points is used. 
Hence, we restrict ourselves to $m \in \{0,1,2\}$, although other interpolant operators and interpolation points can be used as well \cite{krumscheid2017multilevel}.

We use the following \gls{mse} criterion to quantify the accuracy of the function derivative estimate:
\begin{align}
\mse{\Phiem} &\coloneqq \expec{\norm{\Phim - \Phiem }^2_{\linf(\Theta)}}, \quad m \in \{0,1,2\}, \label{eq:phim_mse}
\end{align}
where the norm $\norm{f}_{\linf(\Theta)}$ of a function $f : \Theta \to \setR$ is defined as 
\begin{align}
\norm{f}_{\linf(\Theta)} \coloneqq \esssup_{\theta \in \Theta} |f(\theta)|.
\end{align}
By the triangle inequality, the \gls{mse} can be separated into three terms:
\begin{align}
\mse{\Phiem} &\leq 3\bigg\{ \underbrace{\norm{\Phim - \interpm{\Phi(\thetab)}}_{\linf(\Theta)}^2}_{\text{Squared interpolation error}} + \underbrace{\norm{\interpm{\Phi(\thetab) - \expec{ \Phie(\thetab)}}}^2_{\linf(\Theta)}}_{\text{Squared bias error}} \nonumber \\ 
&+ \underbrace{\expec{ \norm{\interpm{\Phie(\thetab)  - \expec{\Phie(\thetab)}} }^2_{\linf(\Theta)}}}_{\text{Squared statistical error}} \bigg\} \nonumber\\
&\eqqcolon 3 \left\{(\err^{(m)}_{i})^2 + (\err^{(m)}_{b})^2 + (\err^{(m)}_{s})^2\right\},\label{eq:error_split_three}
\end{align}
where we have used the notation $\err^{(m)}_{i}$, $\err^{(m)}_{b}$ and $\err^{(m)}_{s}$ for the interpolation, bias and statistical errors respectively. 

Both the computational cost and accuracy, and thus the complexity, of the \gls{mlmc} estimator are determined by three different sets of parameters; namely the number of interpolation points $n$, the level-wise sample size $N_l$ at each level $l$ and the number of levels $L$.
These should be chosen in a cost optimal way based on suitable a priori or a posteriori error estimates. 
In the next sections, we first review the a priori error estimates and the corresponding complexity analysis from \cite{krumscheid2017multilevel}, before presenting our new and refined error estimators in Section~\ref{sec:novel}.

\section{A priori error estimates on function derivatives and complexity analysis}\label{sec:apriori}
We review in this section the a priori estimators derived in \cite{krumscheid2017multilevel} for each of the error terms in the \gls{mse} bound presented in Eq.~\eqref{eq:error_split_three}.
We review as well the \gls{mlmc} method described therein to adaptively select the parameters of the hierarchy based on a simplified cost model, for which we also state the corresponding complexity result.
The main idea behind the error bounds introduced in \cite{krumscheid2017multilevel} is to exploit the properties of the particular form of the function $\phi$ given in Eq.~\eqref{eq:phi_form} in order to derive an upper bound for the \gls{mse} in Eq.~\eqref{eq:error_split_three}.
Since the function $\phi(\theta,\qoi)$ is uniformly Lipschitz continuous in $\qoi$ for all $\theta \in \Theta$, we have that
\begin{align}
\left|\phi(\theta, \qoi_{l})-\phi(\theta, \qoi_{{l-1}})\right| \leq C_{lip} \left|\qoi_{l}-\qoi_{{l-1}} \right| \qquad \forall \theta \in \Theta,
\end{align}
with finite Lipschitz constant $C_{lip} = 1/(1-\tau),\; \tau \in (0,1)$.
If one can control the decay rates of the expected value and variance of the difference $\qoi_{l}-\qoi_{{l-1}}$ with level $l$, then the corresponding statistics of the difference in the function $\phi$ evaluated at the two levels $\phi(\cdot, \qoi_l) -\phi(\cdot, \qoi_{l-1})$ decay as well with the same or better rates in the $\linf(\Theta)$-norm.
Consequently, complexity results analogous to those available for \gls{mlmc} estimators of the simple expectation of $\qoi$ can be obtained for $\Phie$.

In \cite{krumscheid2017multilevel}, inverse inequalities were used to relate the \gls{mse} of $\Phiem, m \geq 0$ in Eq.~\eqref{eq:phim_mse} to pointwise errors on $\Phie(\theta_j),\; j \in \{1,...,n\}$.
Particularizing the general a priori bound from \cite{krumscheid2017multilevel} to the case of cubic spline interpolation, we obtain:
\begin{subequations}
\begin{align}
\mse{\Phiem}&\leq 3 \left\{ (\errbound^{(m)}_i)^2 + (\errbound^{(m)}_b)^2 + (\errbound^{(m)}_s)^2 \right\}, \label{eq:mse:bound:m-deriv} \\
\text{where }\errbound^{(m)}_i &\coloneqq {C_1}(m) \norm{\Phi^{(4)}}_{\linf(\Theta)} \left(\frac{|\Theta|}{n}\right)^{(4-m)},  \label{eq:apriori_bound_i}\\
\errbound^{(m)}_b &\coloneqq C_2(m) C_3 (n-1)^{m}{b_L}, \label{eq:apriori_bound_b}\\
\errbound^{(m)}_s &\coloneqq  C_2(m) C_3 (n-1)^{m} \sqrt{ c(n)\sum_{l=0}^L \frac{V_l}{N_l}},\label{eq:apriori_bound_s} 
\end{align}
\end{subequations}
where $|\Theta|$ denotes the size of the domain $\Theta$.
Each of the three terms $\err_i^{(m)},\err_b^{(m)}$ and $\err_s^{(m)}$ in Eq.~\eqref{eq:error_split_three} are bounded respectively by the corresponding term $\errbound_i^{(m)},\errbound_b^{(m)}$ and $\errbound_s^{(m)}$ in Eq.~\eqref{eq:mse:bound:m-deriv} and the constants $C_1 (m)$, $C_2(m)$ and $C_3$ are related to the properties of the cubic spline interpolation operator and are detailed in Appendix~\ref{app:err_est} (together with some relevant properties of cubic splines).
The constant $c(n)$ in Eq.~\eqref{eq:apriori_bound_i} is introduced in \cite{linde1975mappings}, further detailed in \cite{giles2017adaptive} and reads:
\begin{align}
c(n)=2\pi\left(\ln(n+1) + \sqrt{8/\pi}\sum_{k=2}^{n+1}k^{-2}{\ln(k)}^{-1/2}\right).
\end{align}
We have also introduced the notation $b_l$ and $V_l$ for the level-wise biases and variances respectively, which are defined as
\begin{equation}
  b_l := \norm{ \Phi - \expec{\phi(\cdot,\qoi_{l})}}_{\slinf(\thetab)}\;,\quad\text{and}\quad
  V_l:= \expec{\norm{\phi(\cdot,\qoi_{l}) - \phi(\cdot,\qoi_{{l-1}})}_{\slinf(\thetab)}^2}\;, \label{eq:bl_vl}
\end{equation}
where the norm $\norm{\cdot}_{\slinf(\thetab)}$ is defined for a function $f\colon \Theta \to \setR$ evaluated at a set of points $\thetab \equiv \{\theta_1,...,\theta_n\}$ as follows:
\begin{align}
\norm{f}_{\slinf(\thetab)} \coloneqq \max_{1 \leq i \leq n} |f(\theta_i)|.
\end{align}
Note that $b_l$, $V_l$ and $\norm{\Phi^{(4)}}_{\linf(\Theta)}$ are usually not directly computable in practice.
However, it is possible to estimate them reliably from the \gls{mlmc} samples themselves.
This will be discussed later in this section. 

Using the a priori bounds derived in Eqs.~\eqref{eq:mse:bound:m-deriv}-\eqref{eq:apriori_bound_s}, we now describe how to select the optimal values $n^*$, $N_l^*$ and $L^*$ such that the \gls{mse} on the function derivative satisfies a tolerance $\epsilon^2$ split with positive weights $w_i, w_b$ and $w_s$ between the squared interpolation, bias and statistical error terms respectively. 
The weights are such that $w_i + w_b + w_s = 1$. 
We define the tolerances $\epsilon_i^2$, $\epsilon_b^2$ and $\epsilon_s^2$ as follows and require each of the terms in Eq.~\eqref{eq:mse:bound:m-deriv} to satisfy their respective tolerances:
\begin{align}
(\errbound^{(m)}_i)^2 \le \epsilon_{i}^2 \coloneqq \frac{w_i \epsilon^2}{3}, \quad \quad (\errbound^{(m)}_b)^2 \le \epsilon_{b}^2 \coloneqq \frac{w_b \epsilon^2}{3}, \quad \quad (\errbound^{(m)}_s)^2 \le \epsilon_{s}^2 \coloneqq \frac{w_s \epsilon^2}{3}, \label{eq:tol_split}
\end{align}

The interpolation error is controlled solely by the number of interpolation points, which is therefore selected first, namely as
\begin{align}
n^* &= \left\lceil \left[ \frac{C_1(m)\norm{\Phi^{(4)}}_{\linf(\Theta)}}{\epsilon_{i}} \right]^{\frac{1}{(4-m)}} |\Theta| \right\rceil,\label{eq:mlmc:tuning:nodes}
\end{align}
ensuring that the squared interpolation error is bounded by $\epsilon_{i}^2$ once $n$ is chosen as in Eq.~\eqref{eq:mlmc:tuning:nodes}.
Given $n^*$, the optimal number of levels $L^*$ is selected to be the smallest level such that the squared bias error satisfies a tolerance $\epsilon^2_{b}$; namely that $C_2(m) C_3 (n^*-1)^{m}b_{L^*} \leq \epsilon_{b}$, that is 
\begin{align}
L^* = \min\left\{K\in\setN_0 \colon b_K \le \frac{\epsilon_{b}}{C_2(m)C_3(n^*-1)^m}\right\}\;.\label{eq:mlmc:tuning:levels}
\end{align}
Lastly, with $n^*$ and $L^*$ fixed, the level-wise sample size $N_l$ at level $l$ is selected to minimise the cost of computing the \gls{mlmc} estimator
\begin{equation}
\cost{ \Phiem } \le \sum_{l=0}^{L^*} N_l(\costsim +n^* \costfunc)+ n^*\costint,
\end{equation}
subject to the following constraint on the squared statistical error: 
\begin{align}
C_2^2(m) C_3^2 {(n^*-1)}^{2m}c(n^*)\sum_{l=0}^{L^*} \frac{V_l}{N_l} \leq \epsilon_{s}^2.
\end{align}
Here, $\costsim$ is the cost of computing one realisation of the correlated pair of approximations $(\qoi_{l}, \qoi_{{l-1}})$ at level $l$, $\costfunc$ is the constant that bounds the cost of evaluating the function $\phi(\theta, \qoi)$ for any $(\theta, \qoi) \in \Theta \times \setR$ and $\costint$ is the cost per interpolation point of constructing the cubic spline interpolant on a uniform grid.
In \cite{krumscheid2017multilevel}, the level-wise sample sizes were selected under the assumption that $\costint$ and $\costfunc$ were non-zero.
However, for the applications addressed in this work, it was found that $\costint$ and $\costfunc$ are usually negligible in comparison to $\costsim$.
Hence, we select the level-wise sample sizes $\{N_l\}_{l=0}^{L^*}$ based on the simplified cost model
\begin{align}
\cost{ \Phiem } \approx \sum_{l=0}^{L^*} N_l \costsim. \label{eq:simple_cost_model}
\end{align}
Consequently, the level-wise sample sizes are selected similar to \cite{giles2008multilevel} as follows:
\begin{equation}
N_l^* = \left\lceil \frac{C_2^2(m) C_3^2c(n^*){(n^*-1)}^{2m}}{\epsilon_{s}^2}  \sqrt{\frac{V_l}{\costsim}}
\sum_{k=0}^{L^*} \sqrt{V_k\mathfrak{c}_k}\right\rceil,\quad 0\le l\le L^*.\label{eq:mlmc:tuning:samples}
\end{equation}
Below, we present a complexity result based on the simplified cost model in Eq.~\eqref{eq:simple_cost_model} using the a priori bounds in Eqs.~\eqref{eq:mse:bound:m-deriv}-\eqref{eq:apriori_bound_s}.
This result is a simplified version of the one presented in \cite{krumscheid2017multilevel} and is tailored to the use of cubic spline interpolation.
We give here the proof for completeness.

\begin{proposition}
\label{thm:complexity:MLMC:spline}
Suppose that there exist positive constants $\alpha$, $\beta$, and $\gamma$ such that $2\alpha \geq \min(\beta, \gamma)$ and that
\begin{enumerate}[label=(\roman*)]
\item $b_l$ decays exponentially with order $\alpha>0$ in $l$, in the sense that $b_l \le c_\alpha e^{-\alpha l}$ for some constant $c_\alpha>0$,\label{eq:bias_rates_phi}
\item $V_l$ decays exponentially with order $\beta>0$ in $l$, in the sense that $V_l \le c_\beta e^{-\beta l}$ for some constant $c_\beta>0$,\label{eq:var_rates_phi}
\item the cost to compute each i.i.d. realisation of $(\qoi_{l},\qoi_{{l-1}})$ increases exponentially with rate $\gamma>0$ in $l$, in the sense that $\costsim = \cost{\qoi_{l}, \qoi_{{l-1}}}\le c_\gamma e^{\gamma l}$ for some constant $c_\gamma$,\label{eq:cost_rates_phi}
\end{enumerate}
for all $l\in\setN_0$, when $h_{l-1} = s h_l$ for some $s>1$ and $m\in\{0,1,2\}$.
For any $0<\epsilon < e^{-1}$, the $m$-th derivative, of the \gls{mlmc} estimator $\Phie$ of $\Phi\in C^4(\Theta)$ with the number $n$ of (uniform) nodes chosen according to Eq.~\eqref{eq:mlmc:tuning:nodes}, the maximum number of levels $L$ as in Eq.~\eqref{eq:mlmc:tuning:levels}, and level-wise sample sizes $N_l$ given by Eq.~\eqref{eq:mlmc:tuning:samples}, satisfies $\mse{\Phiem} \le \epsilon^2$ at a computational cost that is bounded by
\begin{equation*}
\cost{\Phiem} \lesssim \log(\epsilon^{-1}) \epsilon^{-2-\frac{2m}{4-m	}}
    \begin{cases}
       1 , & \text{if } \beta>\gamma,\\
       \log(\epsilon^{-1})^2, & \text{if } \beta=\gamma, \\
       \epsilon^{\frac{\beta-\gamma}{\alpha}\frac{4}{4-m}}, & \text{if } \beta<\gamma. \\
     \end{cases}
\end{equation*}
\end{proposition}
\begin{proof}
We begin by considering the choice of the number of interpolation points given by Eq.~\eqref{eq:mlmc:tuning:nodes}.
We have that $n = \order{\epsilon^{\frac{-1}{4-m}}}$.
In the light of hypothesis~\ref{eq:bias_rates_phi}, the optimal choice of $L$ is given by 
\begin{align}
L = \left\lceil \frac{1}{\alpha} \log\left[\frac{\sqrt{3} c_{\alpha} C_2(m)C_3 (n-1)^m}{\sqrt{w_b}\epsilon} \right]\right\rceil = \order{\log\left(\epsilon^{-\frac{4}{\alpha(4-m)}}\right)}.\label{eq:bias_tuning_model}
\end{align}
Using the expression for $N_l$ in Eq.~\eqref{eq:mlmc:tuning:samples} in the simplified cost model gives 
\begin{align}
\cost{\Phiem} &= \sum_{l=0}^L N_l \costsim
\leq \sum_{l=0}^L \costsim + \frac{3 C_2^2(m)C^2_3c(n){(n-1)}^{2m}}{w_s \epsilon^2}  \left[ \sum_{l=0}^L \sqrt{V_l \costsim} \right]^2,
\end{align}
where the first term is added to take into account the cost of computing at least one sample per level.
Using the hypothesis on the cost $\costsim$ at level $l$, it follows from Eq.~\eqref{eq:bias_tuning_model} that
\begin{align}
\sum_{l=0}^L \costsim \leq c_{\gamma}\sum_{l=0}^L e^{\gamma l} = c_{\gamma} \left[\frac{e^{\gamma L} - e^{-\gamma} }{1 - e^{-\gamma}}\right] = \order{\epsilon^{-\frac{\gamma}{\alpha}\frac{4}{4-m}}}.
\end{align}
In addition, we use the hypotheses on the variance $V_l$ at level $l$ to write
\begin{align}
\sum_{l=0}^L \sqrt{V_l \costsim} &= \sqrt{c_{\beta} c_{\gamma}}\sum_{l=0}^L e^{\left(\frac{\gamma-\beta}{2}\right)l} \\
&= \sqrt{c_{\beta} c_{\gamma}} \begin{cases}
\left[ \frac{e^{pL} - e^{-p} }{1-e^{-p}} \right], \text{ if } \beta \neq \gamma\\
(L+1), \text{ if } \beta = \gamma
\end{cases}
\end{align}
where $p=(\gamma-\beta)/2$.
In the event that $\beta > \gamma$, we have $p < 0$.
In combination with Eq.~\eqref{eq:bias_tuning_model}, we have that
\begin{align}
\left[ \frac{e^{pL} - e^{-p} }{1-e^{-p}} \right] \leq \frac{e^{-p}}{e^{-p}-1} = \order{1}.
\end{align}
In the event that $\beta < \gamma$, we have that $p > 0$ and hence that 
\begin{align}
\left[ \frac{e^{pL} - e^{-p} }{1-e^{-p}} \right] = \order{\epsilon^{\frac{\beta-\gamma}{2\alpha}\frac{4}{4-m}}}.
\end{align}
In summary, we can write that 
\begin{align}
\left[\sum_{l=0}^L \sqrt{V_l \costsim}\right]^2 = \begin{cases} \order{1}, &\text{ if } \beta > \gamma,\\
\order{\log(\epsilon^{-1})^2}, &\text{ if } \beta = \gamma,\\
\order{\epsilon^{\frac{\beta -\gamma}{\alpha}\frac{4}{4-m}}} &\text{ if } \beta < \gamma.\\
\end{cases}
\end{align}
As a final step, we note that $c(n)=\order{\log(n)} \equiv\order{\log(\epsilon^{-1})}$ and that $(n-1)^{2m} = \order{\epsilon^{\frac{-2m}{4-m}}}$.
Combining all the terms together, we have that 
\begin{align}
\cost{\Phiem} \lesssim \epsilon^{\frac{-\gamma}{\alpha} \frac{4}{4-m}}
    + \log(\epsilon^{-1}) \epsilon^{-2-\frac{2m}{4-m	}}
    \begin{cases}
       1 , & \text{if } \beta>\gamma,\\
       \log(\epsilon^{-1})^2, & \text{if } \beta=\gamma, \\
       \epsilon^{\frac{\beta-\gamma}{\alpha}\frac{4}{4-m}}, & \text{if } \beta<\gamma. \\
     \end{cases}
\end{align}

In addition, we require that $2\alpha \geq \min(\beta, \gamma)$ for the complexity to be dominated by the second term alone and not by the first term that quantifies the cost of a single simulation.
This can be seen by considering each of the following two cases.
In the first case $\beta \geq \gamma$, we have
\begin{align}
\frac{4 \gamma}{(4-m) \alpha} \leq \frac{8}{4-m} \iff 2\alpha \geq \gamma,
\end{align}
since $m\leq 2$. 
For the second case $\beta < \gamma$, we have that
\begin{align}
\frac{4 \gamma}{(4-m) \alpha} \leq \frac{8}{4-m} + \frac{4(\gamma-\beta)}{(4-m)\alpha}\iff 2\alpha \geq \beta.
\end{align}
This completes the proof. 
\end{proof}

\subsection{Practical aspects and tuning of hierarchy parameters}\label{sec:practical}
As pointed out earlier, the error bounds in Eqs.~\eqref{eq:mse:bound:m-deriv}-\eqref{eq:apriori_bound_s} are still not directly computable.
To this end, we present below a possible way to estimate the level-wise terms $b_l$ and $V_l$, as well as the term $\norm{\Phi^{(4)}}_{\linf(\Theta)}$ based on the available samples of the \gls{mlmc} estimator.
To estimate the level-wise bias terms $b_l$, we first note that with the help of Hypothesis~\ref{eq:bias_rates_phi} from Proposition~\ref{thm:complexity:MLMC:spline}, we have that 
\begin{align}
\lim_{l \to \infty} \expec{\phi(\cdot,\qoi_{l})} = \Phi \label{eq:bias_implication}
\end{align}
in $\slinf(\thetab)$ and hence, similar to the procedure in \cite{Giles2015a}, one can obtain the heuristic estimate
\begin{align}
b_l &= \norm{ \Phi - \expec{\phi(\cdot,\qoi_{l})}}_{\slinf(\thetab)}\\
&=\norm{ \sum_{k=l+1}^{\infty}\expec{\phi(\cdot,\qoi_{{k}})} - \expec{\phi(\cdot,\qoi_{{k-1}})}}_{\slinf(\thetab)}\\
&\approx \frac{ \norm{ \expec{\phi(\cdot,\qoi_{{l}}) - \phi(\cdot,\qoi_{{l-1}})}}_{\slinf(\thetab)}}{(e^{\alpha}-1)}. \label{eq:naive_bias}
\end{align}
The expectation in Eq.~\eqref{eq:naive_bias} is then estimated with a Monte Carlo estimator over the $N_l$ independent and identically distributed correlated sample pairs $\{ \qoi_{l}^{(i,l)},\qoi_{{l-1}}^{(i,l)} \}_{i=1}^{N_l}$, denoted by $\hat{b}_l$.
The variance term $V_l$ can also be computed by replacing the expectation in Eq.~\eqref{eq:bl_vl} with a similar sample average estimator $\hat{V}_l$, yielding the following:
\begin{align}
\hat{b}_l &\coloneqq  \frac{1}{N_l} \frac{ \norm{ \sum_{i=1}^{N_l} \phi(\cdot,\qoi_l^{(i,l)}) - \phi(\cdot,\qoi_{l-1}^{(i,l)})}_{\slinf(\thetab)}}{(e^{\alpha}-1)} ,\\
\hat{V}_l &\coloneqq \frac{1}{N_l} \sum_{i=1}^{N_l} \norm{\phi(\cdot,\qoi_l^{(i,l)}) - \phi(\cdot,\qoi_{l-1}^{(i,l)})}_{\slinf(\thetab)}^2 .\label{eq:vl_est}
\end{align}

To start the estimation procedure, one typically computes a small number of \gls{qoi} realisations on a pre-fixed small number of levels.
Such a small initial hierarchy is called a ``screening'' hierarchy.
The screening hierarchy is typically selected such that it is significantly smaller than the expected optimal hierarchy, so that the computational cost of the screening hierarchy is negligible in comparison to the optimal hierarchy.
Using the screening hierarchy, one can then obtain initial estimates of $\hat{b}_l$ and $\hat{V}_l$, as well as their decay rates in the levels $l$, based on which the optimal number of interpolation points $n^*$, number of levels $L^*$ and level-wise sample sizes $N_l^*$ can be selected for a prescribed tolerance $\epsilon^2$ according to Eqs.~\eqref{eq:mlmc:tuning:nodes}, \eqref{eq:mlmc:tuning:levels} and \eqref{eq:mlmc:tuning:samples}.
A \gls{mlmc} estimator can then be constructed with the estimated optimal hierarchy, upon which better estimates of $\hat{b}_l$, $\hat{V}_l$ and a better \gls{mlmc} estimator can be produced in an iterative manner.
Such an approach was pioneered in \cite{giles2008multilevel}.

To compute the optimal hierarchy $L^*$ and $N_l^*$ from Eqs.~\eqref{eq:mlmc:tuning:levels} and~\eqref{eq:mlmc:tuning:samples}, one may need values of $\hat{b}_l$ and $\hat{V}_l$ on levels $L < l \leq L^*$ beyond the current maximum level $L$ used, for which no samples are available.
To this end, we fit the theorized models $c_{\alpha} e^{-\alpha l}$ and $c_{\beta}e^{-\beta l}$ from Proposition~\ref{thm:complexity:MLMC:spline} to $\hat{b}_l$ and $\hat{V}_l$ respectively, for the levels where these estimates are available, using a least squares fit.
We then use the level-wise biases and variances predicted by these models instead of the actual estimates in Eqs.~\eqref{eq:mlmc:tuning:levels} and~\eqref{eq:mlmc:tuning:samples}. 

For computing the interpolation error bound in Eq.~\eqref{eq:apriori_bound_i}, as well as for computing the optimal number of interpolation points in Eq.~\eqref{eq:mlmc:tuning:nodes}, we are required to estimate the norm of the fourth derivative of the function $\Phi$.
The estimate of the fourth derivative cannot be computed directly from the interpolant as $\mathcal{S}^{(4)}_n(\hat{\Phi}_L(\thetab))$ since $\mathcal{S}_n$ is a cubic spline, hence $\interp{\Phie(\thetab)} \in C^2(\Theta)$ and the fourth derivative does not exist.
We propose instead the use of \gls{kde} techniques to solve this issue.
Such a \gls{kde} smoothing procedure is used extensively through this work and is described in detail in Section~\ref{sec:kde_est}.

The procedure to estimate $\norm{\Phi^{(4)}}_{\linf(\Theta)}$ is as follows.
We begin by selecting the level $\lceil L/2 \rceil$ from the hierarchy.
This level is selected since $N_{\lceil L/2 \rceil}$ is sufficiently large to justify the \gls{kde} approach but $\hat{\Phi}_{\lceil L/2 \rceil}$ is also expected to be sufficiently close to $\Phi$.
Although there may exist an optimal choice for this level, this particular choice was found to suffice for the applications in this study. 
A \gls{kde}-smoothened function estimate $\Upsilon_{\lceil L/2 \rceil}(\theta) \coloneqq \mathbb{E}^{kde}_{\lceil L/2 \rceil}\left[\phi(\theta,\cdot)\right]$ of the function $\hat{\Phi}_{\lceil L/2 \rceil}$ is produced according to the procedure described in Section~\ref{sec:kde_est}, where $\mathbb{E}^{kde}_l$ is defined in Eq.~\eqref{eq:kde_def} below.
The fourth derivative $\Upsilon^{(4)}_{\lceil L/2 \rceil}$ is then computed using a second order central difference approximation where $\Upsilon_{\lceil L/2 \rceil}$ is evaluated on a uniform grid on $\Theta$ with $n' \gg n$ points.
The norm is also approximated on the same grid:
\begin{align}
\norm{\Upsilon^{(4)}_{\lceil L/2 \rceil}}_{\linf(\Theta)} \approx \max_{i \in \{1,...,n'\}} \left|\Upsilon^{(4)}_{\lceil L/2 \rceil}(\theta_i)\right|. \label{eq:kde:norm_approx}
\end{align}
We summarize below the fully computable a priori error estimators:
\begin{subequations}
\begin{align}
\errbound^{(m)}_i \approx \errest^{(m)}_i &\coloneqq  C_1(m)  \norm{\Upsilon^{(4)}_{\lceil L/2 \rceil}}_{\linf(\Theta)} \left(\frac{|\Theta|}{n}\right)^{(4-m)},  \label{eq:apriori_i}\\
\errbound^{(m)}_b \approx \errest^{(m)}_b &\coloneqq C_2(m) C_3 (n-1)^{m} \hat{b}_L, \label{eq:apriori_b}\\
\errbound^{(m)}_s \approx \errest^{(m)}_s &\coloneqq  C_2(m) C_3 (n-1)^{m} \sqrt{ c(n)\sum_{l=0}^L \frac{\hat{V}_l}{N_l} }.\label{eq:apriori_s} 
\end{align} 
\end{subequations}
In Section~\ref{sec:est_tests}, we will compare the a priori error estimators described here to the newly developed error estimators introduced in Section~\ref{sec:novel}.
As will be seen in Section~\ref{sec:est_tests}, the a priori error estimators may prove to be too conservative and lead to hierarchies with large values of $L$ and $N_l$ when selecting these parameters to attain practical tolerances on the \gls{mse}.
These hierarchies become impractically expensive to simulate.
The main advantage of the a priori error estimators is that the bias and variance terms $\hat{b}_l$ and $\hat{V}_l$ computed in the manner described in this section decay exponentially in the levels with the same rate as the underlying \gls{qoi} $\qoi$ (see \cite{krumscheid2017multilevel}).
However, the inequalities used to achieve this favourable property produce large leading constants.
We will introduce new error estimators in Section~\ref{sec:novel} that preserve the exponential decay property, and consequently the complexity result in Proposition~\ref{thm:complexity:MLMC:spline}, while reducing or eliminating these leading constants.

\subsection{Function derivative estimation by \gls{kde} based smoothing}\label{sec:kde_est}
The error estimator Eq.~\eqref{eq:apriori_i} requires estimating the fourth derivative of the function $\Phi_l(\theta) = \expec{\phi(\theta,\qoi_l)}$.
For this, we could first estimate the expected value with a Monte Carlo sample average estimator.
As can be seen easily from Eq.~\eqref{eq:parest_form}, this estimate produces a piecewise linear function in $\theta$.
This in turn implies that the first derivative of such a function is piecewise constant, and that second and higher order derivatives do not exist.
Using an empirical direct Monte Carlo approach to estimating by Monte Carlo quantities such as $\norm{\Phi^{(4)}}_{\linf(\Theta)}$, which are important to the error estimation and to the procedure of adaptively selecting the hierarchy parameters, is hence not viable.

We propose the use of \gls{kde} techniques to remedy this issue.
The \gls{kde} procedure for constructing derivatives of $\Phi_l$ is presented here.
An appropriately smoothed probability density function $p^{kde}_l$ of $\qoi_l$ is constructed using a one dimensional Gaussian kernel centred on each of the $N_l$ fine samples $\{\qoi_{l}^{(i,l)}\}_{i=1}^{N_l}$ at level $l$.
The function $\Phi_l$ is then approximated as follows:
\begin{align}
\Phi_l(\theta) &= \int \phi(\theta, q)  p_l(q) d q \\
&\approx \int \phi(\theta, q) p^{kde}_l(q)  d q \eqqcolon \mathbb{E}^{kde}_{l}\left[\phi(\theta,\qoi_l)\right] \label{eq:kde_def}\\
\text{where } p^{kde}_l(q) &= \frac{1}{N_l} \sum_{i=1}^{N_l} K_{\delta_l}\left(q,\qoi_l^{(i,l)}\right) ,
\end{align}
$K_{\delta_l}(\cdot,\mu)$ denotes the Gaussian kernel with mean $\mu$ and bandwidth parameter $\delta_l > 0$.
The bandwidth parameter $\delta_l$ controls the ``width'' of the kernel and is related to the covariance of the underlying data.
It can in principle be a function of the level $l$ and the sample size $N_l$.
Here, it is chosen according to Scott's rule \cite{scott1979optimal}, which ensures that $\delta_l \to 0$ as $N_l \to \infty$.
Since the expressions for $K_{\delta}$ and $\phi$ are known, a closed form expression can be computed for $\mathbb{E}^{kde}_{l}\left[\phi(\theta,\qoi_l)\right]$, which is a $C^{\infty}$ function in $\theta$ due to the smoothness of the Gaussian kernel.
The smoothed expression can then be evaluated on a fine grid in $\Theta$ with $n' \gg n$ points and derivatives can be evaluated exactly, or more conveniently, estimated by finite different formulas.

The \gls{kde} procedure will also be used in the novel bias estimator proposed in Section~\ref{sec:novel_bias}.
The novel bias estimator requires, in particular, estimating the quantities $\norm{ \interpm{ \expec{\phi(\thetab,\qoi_l) - \phi(\thetab, \qoi_{l-1})}}}_{\linf(\Theta)}$.
For estimating such quantities, which require the computation of derivatives of $\Phi$, the \gls{kde} smoothing follows a similar procedure.
However, the density is now bivariate, namely characterising the distribution of the two correlated random variables $\qoi_l$ and $\qoi_{l-1}$:
\begin{align}
\expec{\phi(\theta,\qoi_l)-\phi(\theta,\qoi_{l-1})} &=\int\int\left[ \phi(\theta, q_l)-  \phi(\theta, q_{l-1})\right] p_{l,l-1}(q_l,q_{l-1}) d q_l d q_{l-1} \label{eq:est_diff}\\
&\approx \int\int\left[ \phi(\theta, q_l)-  \phi(\theta, q_{l-1})\right] p_{l,l-1}^{kde}(q_l,q_{l-1}) d q_l d q_{l-1}\\
&\eqqcolon \mathbb{E}^{kde}_{l,l-1}\left[\phi(\theta,\qoi_l)-\phi(\theta,\qoi_{l-1}) \right]\label{eq:novel_bias_def}\\
\text{where } p^{kde}_{l,l-1}(q_l,q_{l-1}) &= \frac{1}{N_l} \sum_{i=1}^{N_l} K_{\delta_l}\left(q_l,\qoi_l^{(i,l)}\right) K_{\delta_{l-1}}\left(q_{l-1},\qoi_{l-1}^{(i,l)}\right).
\end{align}
The bandwidth parameters $\delta_l$ and $\delta_{l-1}$ are chosen according to Scott's rule based on the sample sets $\{\qoi_{l}^{(i,l)}\}_{i=1}^{N_l}$ and $\{\qoi_{l-1}^{(i,l)}\}_{i=1}^{N_l}$, respectively.
One consequence of this method of bandwidth parameter selection is that the parameter $\delta_{l}$ will be larger on fine levels where $N_l$ is typically small.
Although the joint density $p_{l,l-1}$ will tend to concentrate around the diagonal as $l$ increases, the \gls{kde} density may include a significant off-diagonal mass for large values of $\delta_{l}$ and $\delta_{l-1}$.
This may induce a larger variance of the estimator in Eq.~\eqref{eq:novel_bias_def} with respect to naively estimating the expectation in Eq.~\eqref{eq:est_diff} using Monte Carlo.
The advantage of using Eq.~\eqref{eq:novel_bias_def}, over a Monte Carlo estimate of Eq.~\eqref{eq:est_diff}, is that one can differentiate the approximation in Eq.~\eqref{eq:novel_bias_def} with respect to $\theta$, since the resulting expression is smooth with respect to $\theta$.

One can also use anisotropic or more complex choices for the kernel, which may require numerical integration or special quadrature. 
However, for the purposes of this work, the isotropic Gaussian kernel was found to suffice. 
In addition to being able to compute higher order derivatives of function estimates, the \gls{kde} based smoothing approach also provides other important benefits to the error estimation and adaptivity that will be demonstrated in Section~\ref{sec:est_tests}.

\section{Novel error estimators for function derivatives} \label{sec:novel}
As will be demonstrated numerically in Section~\ref{sec:est_tests} below, the a priori interpolation error estimator $\errest_i^{(m)}$ in Eq.~\eqref{eq:apriori_i}, provides a satisfactory error bound in practice.
Moreover, the interpolation error is often much smaller than the bias and statistical error terms, at least for the cases explored in this work. 
As a result, we primarily target the accurate estimation of the bias and statistical error terms.
In fact, we propose here new estimators for these quantities that provide tighter bounds on the corresponding true errors, while also preserving the same decay rates with respect to $l$ as the a priori level-wise bias and variance contributions $b_l$ and $V_l$ and lead eventually to the complexity bound of Proposition~\ref{thm:complexity:MLMC:spline}.

\subsection{Bias term}\label{sec:novel_bias}
We begin with the bias term $\err^{(m)}_b$ in the expression for the \gls{mse} on $\Phiem$ in Eq.~\eqref{eq:error_split_three}.
The a priori bound $\errbound^{(m)}_b$ from Eq.~\eqref{eq:apriori_bound_b}, together with the hypotheses from Proposition~\ref{thm:complexity:MLMC:spline}, implies that $\err_b^{(m)}$ can be upper-bounded as follows,
\begin{align}
\err_b^{(m)} \coloneqq \norm{\interpm{\Phi(\thetab)-\expec{\Phie(\thetab)}}}_{\linf(\Theta)}\lesssim e^{- \alpha L},
\end{align}
with a hidden constant, possibly depending on the derivative order $m$ and the number of interpolation points $n$. 
Combining this with the implication from Eq.~\eqref{eq:bias_implication}, this in turn implies that the level-wise differences $\norm{\interpm{\expec{\hat{\Phi}_l(\thetab)-\hat{\Phi}_{l-1}(\thetab)}}}_{\linf(\Theta)}$ decay with at least a rate $\alpha$ in the levels $l$; then proceeding as in Eq.~\eqref{eq:naive_bias}, we can reasonably estimate the bias error as follows:
\begin{align}
\err_b^{(m)} \approx \frac{\norm{\interpm{\expec{\hat{\Phi}_L(\thetab)-\hat{\Phi}_{L-1}(\thetab)}}}_{\linf(\Theta)}}{(e^{\alpha}-1)}.
\end{align}
The expectation on the right-hand side could in practice be estimated using a sample average Monte Carlo estimator with the $N_L$ samples on the finest level $L$, for example:
\begin{align}
\expec{\hat{\Phi}_L-\hat{\Phi}_{L-1}} &= \expec{\phi(\cdot,\qoi_L)-\phi(\cdot,\qoi_{L-1})} \nonumber\\
&= \int \left[ \phi(\cdot,q_L) - \phi(\cdot,q_{L-1})\right] p_{L,L-1}(q_L, q_{L-1})dq_L dq_{L-1}\\
&\approx \frac{1}{N_L}\sum_{i=0}^{N_L} \phi(\cdot, \qoi_{L}^{(i,L)})-\phi(\cdot, \qoi_{{L-1}}^{(i,L)}),
\end{align}
where we denote the true joint probability density of the bivariate random variable $(\qoi_L ,\qoi_{L-1} )$ as $p_{L,L-1}$ and have replaced it with the empirical measure induced by the Monte Carlo estimator.
As was seen in Section~\ref{sec:kde_est}, this approximation causes issues when computing quantities that depend on derivatives of such a Monte Carlo estimator; namely that the first derivative is piecewise constant and that the second and higher derivatives do not exist for such an estimator.
This results in level-wise bias estimators that no longer satisfy the decay hypotheses of Proposition~\ref{thm:complexity:MLMC:spline}.

This problem was solved in the a priori estimator $\errest_b^{(m)}$ in Eq.~\eqref{eq:apriori_b} by using spline inverse inequalities as follows:
\begin{align}
\norm{\interpm{\expec{\hat{\Phi}_L(\thetab)-\hat{\Phi}_{L-1}(\thetab)}}}_{\linf(\Theta)} \leq C_2(m)C_3(n-1)^m \norm{\expec{\hat{\Phi}_L-\hat{\Phi}_{L-1}}}_{\linf(\Theta)}.
\end{align}
However, this procedure leads to unacceptably large constants. 
Here we propose a new estimator that avoids using inverse inequalities and, instead, directly estimates the term $\interpm{\expec{\hat{\Phi}_l-\hat{\Phi}_{l-1}}}$ using the \gls{kde} technique described in Section~\ref{sec:kde_est} to smooth the empirical measure; that is, approximating the unknown joint density $p_{L,L-1}$ by a bivariate \gls{kde} smoothed density $p^{kde}_{L,L-1}$ as described in Section~\ref{sec:kde_est}.

The resultant novel estimator for $\err_b^{(m)}$ is hence given by 
\begin{align}
\err_{b}^{(m)} \approx \errest_{b,new}^{(m)} \coloneqq \frac{\norm{\interpm{\mathbb{E}^{kde}_{L,L-1}\left[\phi(\thetab, \qoi_L)-\phi(\thetab, \qoi_{L-1})\right]}}_{\linf(\Theta)}}{(e^{\alpha}-1)} \eqqcolon \frac{\hat{b}_{L,new}^{(m)}}{(e^{\alpha}-1)},\label{eq:blm_def}
\end{align}
where we have defined the level-wise bias terms $\hat{b}_{l,new}^{(m)}, l \in \{0,...,L\}$.
In practice, the decay rate $\alpha$ is estimated by fitting the model $c_{\alpha}e^{-l\alpha}$ by least squares fit on the estimates $\hat{b}_{L,new}^{(m)}$ for $l \in \{ 1,..,L\}$.

\subsection{Statistical error term}\label{sec:stat_err_est}
The squared statistical error term in Eq.~\eqref{eq:error_split_three} has the form
\begin{align}
(\err^{(m)}_{s})^2 = \expec{ \norm{\interpm{ \Phie(\thetab)  - \expec{\Phie(\thetab)}}}^2_{\linf(\Theta)} }.\label{eq:stat_err}
\end{align}
The a priori bound described in Section~\ref{sec:apriori} for the statistical error also suffers from a possibly large leading constant that results in conservative statistical error estimates, as will be highlighted below.
As an alternative, we propose the use of a bootstrapping technique \cite{tibshirani1993introduction} to estimate this term as follows.
First, observe that a \gls{mlmc} estimator $\Phiem$ of $\Phim$ is defined through the hierarchy of samples denoted by 
\begin{align}
\mathcal{Q} \equiv \Big\{ \{\qoi_l^{(i,l)},\qoi_{l-1}^{(i,l)}\}_{i=1}^{N_l} \Big\}_{l=0}^{L}.
\end{align}
The idea behind bootstrapping is to create $N_{bs} \in \setN$ new \gls{mlmc} estimators of $\Phi$ denoted $\hat{\Psi}_1, \hat{\Psi}_2,...,\hat{\Psi}_{N_{bs}}$, each defined by a hierarchy of samples of the same size as the original hierarchy $\mathcal{Q}$.
For each $\hat{\Psi}_j$, this is done by randomly selecting $N_l$ sample pairs $(\tilde{\qoi}_l^{(j;i,l)},\tilde{\qoi}_{l-1}^{(j;i,l)}) = (\qoi_l^{(M_{ij},l)},\qoi_{l-1}^{(M_{ij},l)})$, $i=\{1,...,N_l\}$ with $M_{1j},...,M_{N_lj} \overset{\text{i.i.d}}{\sim} \mathcal{U}(\{1,...,N_l\})$ and $j=\{1,...,N_{bs}\}$ at each level $l$ to define a resampled hierarchy $\mathcal{Q}_j$ by:
\begin{align*}
\mathcal{Q}_j \equiv \Big\{ \{\tilde{\qoi}_l^{(j;i,l)},\tilde{\qoi}_{l-1}^{(j;i,l)}\}_{i=1}^{N_l} \Big\}_{l=0}^{L},  \quad j \in \{1,...,N_{bs}\}.
\end{align*}
The bootstrapped \gls{mlmc} estimate $\hat{\Psi}_j$ defined through $\mathcal{Q}_j$ then also provides an estimator of $\Phim$.
Using the sample of $N_{bs}$ bootstrapped \gls{mlmc} estimators, one can approximate the expectations in Eq.~\eqref{eq:stat_err} by sample averages over the bootstrapped \gls{mlmc} estimators.
That is, the statistical error $\err^{(m)}_{s}$ can be estimated by the bootstrapped estimate $\errest^{(m)}_{s,new}$ as
\begin{align}
(\err^{(m)}_{s})^2 \approx (\errest^{(m)}_{s,new})^2 \coloneqq \frac{1}{N_{bs}}\sum_{j=1}^{N_{bs}} \|\interpm{\hat{\Psi}_j(\thetab) - \bar{\Psi}(\thetab)}\|_{\linf(\Theta)}^2,\label{eq:se_bs}
\end{align}
where $\bar{\Psi}$ denotes the sample average of $\{\hat{\Psi}_j\}_{j=1}^{N_{bs}}$.

The choice of $N_{bs}$ is made adaptively.
First, it is set to an initial fixed value.
Then, since Eq.~\eqref{eq:se_bs} is a Monte Carlo estimator, the sample variance of the $\linf(\Theta)$-norms is used to estimate the \gls{mse} of the statistical error estimate.
If this \gls{mse} exceeds a fixed fraction of the statistical error tolerance $\epsilon_{s}^2$, the number of bootstrapped samples $N_{bs}$ is doubled, and the process is repeated until the tolerance is satisfied. 
In addition, since the cost of bootstrapping and interpolating is negligible in comparison to sample generation costs, $N_{bs}$ can be arbitrarily large without a significant additional cost to computing the \gls{mlmc} estimator $\Phiem$ itself. 

\subsection{Summary of novel error estimator}
To summarise the developments above, we have proposed novel bias and statistical error estimators to improve on the properties of the a priori error estimator demonstrated in Section~\ref{sec:apriori}.
The final estimator reads:
\begin{align}
\mse{\Phiem} &\leq  3{C_1^2}(m) \norm{\Upsilon^{(4)}_{\lceil L/2 \rceil}}_{\linf(\Theta)}^2 \left(\frac{|\Theta|}{n}\right)^{2(4-m)} \nonumber \\
&+ 3 \frac{(\hat{b}_{L,new}^{(m)})^2}{(e^\alpha-1)^2} + \frac{3}{N_{bs}}\sum_{j=1}^{N_{bs}} \|\interpm{\hat{\Psi}_j(\thetab) - \bar{\Psi}(\thetab)}\|_{\linf(\Theta)}^2. \label{eq:novel_mse_all_terms}
\end{align}
In fact, we will show in the following section that the new error estimator preserves decay rates of the underlying \gls{qoi} while reducing or eliminating large leading constants and leading to a tighter error bound.

\subsection{Demonstration and comparison of error estimators}\label{sec:est_tests}
To demonstrate the performance of the error estimators introduced in the previous sections, we introduce a simple toy problem.
Specifically, we consider a random Poisson equation in two spatial dimensions,
\begin{equation}
-\Delta u = f\;,\quad\text{in }D= {(0,1)}^2\;,\label{eq:ellip2d:prob_b}
\end{equation}
with homogeneous Dirichlet boundary conditions.
The forcing term $f$ is given by
\begin{equation}
	f(x) = -C \xi({x_1}^2 + {x_2}^2 - x_1 - x_2)\;,\quad 0 \leq x_1, x_2 \leq 1\;,
\end{equation}
with $\xi$ being a random variable distributed according to the Beta$(2,6)$ distribution and $C>0$ a positive constant. 
This problem was also used as a demonstrative example in the companion paper \cite{Pisaroni_mlmcPart1_pre} of this work, in order to demonstrate \gls{mlmc} estimators for higher order central moments.
For this forcing term, the solution to the \gls{pde} can be computed explicitly and reads 
\begin{align}
u(x_1,x_2) = C \xi x_1 x_2 (1 - x_1) (1 - x_2)/2.
\end{align}
The \gls{qoi} we consider is the spatial average of the solution, that is
\begin{equation}
\qoi := \int_D u\,dx = \frac{C}{72}\xi\;.
\end{equation}
For the remainder of the study, we set $C=432$, leading to $\qoi = 6 \xi$.

Since we have the explicit dependence of the \gls{qoi} $\qoi$ on the random input $\xi$, we can easily compute the exact distribution of $\qoi$ given that we know the distribution of $\xi$.
In particular, we can compute $\Phi(\theta) = \expec{\phi(\theta,\qoi)}$ with $\phi$ as in Eq.~\eqref{eq:phi_form} exactly for any $\tau \in (0,1)$.
Indeed, the density $p_{\eta}$ of a random variable $\eta = \kappa \xi$ with $\kappa > 0$ reads:
\begin{align}
p_{\eta}(x) = \frac{42}{\kappa} \left(1-\frac{x}{\kappa}\right)^5 \frac{x}{\kappa},\quad x \in [0,\kappa].
\end{align}
Setting $\kappa = 6$ leads to the following form for $\Phi$ based on the density $p_{\qoi}$ of $\qoi =6 \xi$:
\begin{align}
\Phi(\theta) &= \theta + \frac{1}{1-\tau} \int_{0}^{6} (q - \theta)^+ p_{\qoi}(q) dq\nonumber\\
&= \theta - \frac{(\theta-6)^7(\theta+2)}{373248(1-\tau)}. \label{eq:true_phi}
\end{align}
We plot the true \gls{cdf} $F_{\qoi}(\theta) \coloneqq \mathbb{P}(\qoi \leq \theta)$ in Fig.~\ref{fig:cdf_ellipt}. 
We also list the true values of the \gls{var} $\quant$ and the \gls{cvar} $\cvar$ for different significances $\tau$ in Table~\ref{tab:ellip2d:quantiles}.

\begin{figure}[H]
	\begin{minipage}[b]{0.35\textwidth}
		\centering
		\includegraphics[width=\textwidth]{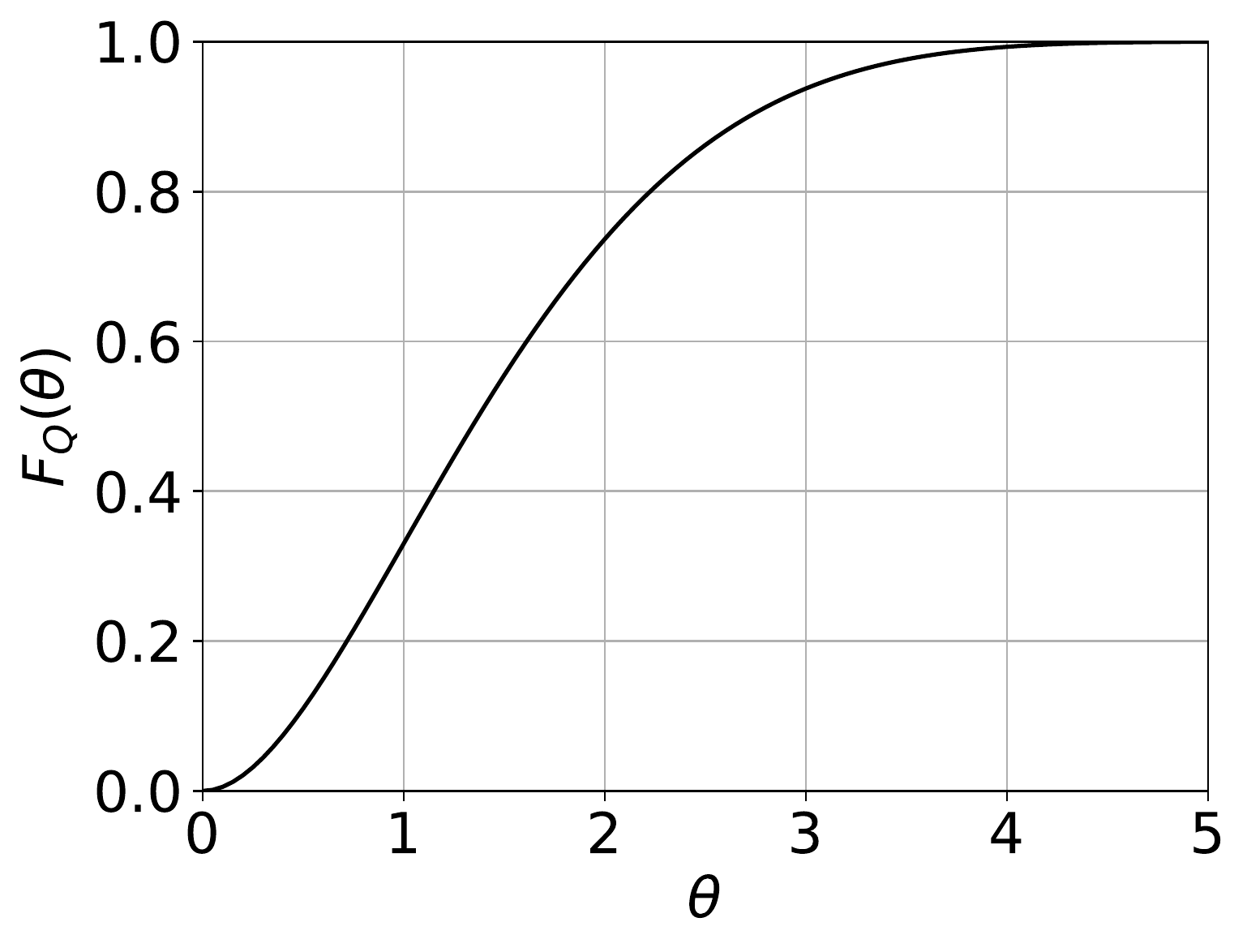}
		\captionof{figure}{\gls{cdf} $F_{\qoi}$ of $\qoi$ for the Poisson problem.}
		\label{fig:cdf_ellipt}
	\end{minipage}
	\hfill
	\begin{minipage}[b]{0.55\textwidth}
		\centering
		\begin{tabular}{ccc}  
			\toprule
			$\tau$ & $\quant = F_\qoi^{-1}(\tau)$ & $\cvar$\\
			\midrule
			$0.6$ & $1.611077$ & $2.369803$\\
			$0.7$ & $1.885696$ & $2.578204$\\
			$0.8$ & $2.225169$ & $2.843327$\\
			$0.9$ & $2.715390$ & $3.236473$\\
			\bottomrule
		\end{tabular}
		\captionof{table}{\gls{var} and \gls{cvar} values for the \gls{qoi} associated with Poisson problem.}
		\label{tab:ellip2d:quantiles}
	\end{minipage}
\end{figure}

For the numerical assessment of the error estimators, the Poisson problem in Eq.~\eqref{eq:ellip2d:prob_b} is discretised using second order central finite differences on a hierarchy of uniform meshes, where the number of degrees of freedom at level $l$ is given by $(5 \times 2^l -2)^2$.
The resultant system is solved directly using sparse LU factorisation.
The approximations $\{\qoi_{l}\}_{l=0}^L$ are also linear in the random variable $\xi$ and hence, the solution for each discretisation can be precomputed for a fixed value $\xi=1$ and simply multiplied afterwards by the random realization of $\xi$ to yield the random \gls{qoi}.
In addition, this implies that the ``true'' function $\Phi_l$ is also known in closed-form for all levels.

\subsubsection{Assessment of the interpolation error estimator}\label{sec:interp_error_comp}
We compute the true interpolation error by considering the true function $\Phi$ presented in Eq.~\eqref{eq:true_phi} for $\tau = 0.7$.
The interval of interest is selected to be $\Theta \equiv [1.5,2.5]$, since we expect the $70\%$-\gls{var} to be within this interval (cf. Table~\ref{tab:ellip2d:quantiles}).
The true interpolation error $e_{i, tru}^{(m)}$ of the $m^{th}$ derivative of $\Phi$ is given by
\begin{align*}
\err^{(m)}_{i,tru} = \norm{\interpm{\Phi(\thetab)}-\Phim}_{\linf(\Theta)}.
\end{align*}
We compare this with the fully a priori error estimate $\hat{e}^{(m)}_{i}$, introduced in Eq.~\eqref{eq:apriori_i}.
Instead of the \gls{kde} method described in Section~\ref{sec:apriori}, the norm of the fourth derivative of $\Phi$ is estimated using the analytical form of $\Phi^{(4)}$ evaluated on a fine grid, in order to focus solely on the quality of the error estimator.

Fig.~\ref{fig:interp_error} compares the estimator $\hat{\err}_i^{(m)}$ with the true error for different values of the number of interpolation points $n$ and for different derivatives $\Phim$ for $m \in \{0,1,2\}$.
As can be seen from plots in that figure, $\hat{\err}_i^{(m)}$ produces a satisfactory bound on the true error for the range of interpolation points tested, and for all value of $m \in \{0,1,2\}$.
The figure also shows that both the true error and the error estimate $\errest^{(m)}_i$ follow the expected decay, which is $\mathcal{O}(n^{4-m})$. 

\begin{figure}[h]
  \begin{subfigure}{0.32\textwidth}
    \includegraphics[width=\textwidth]{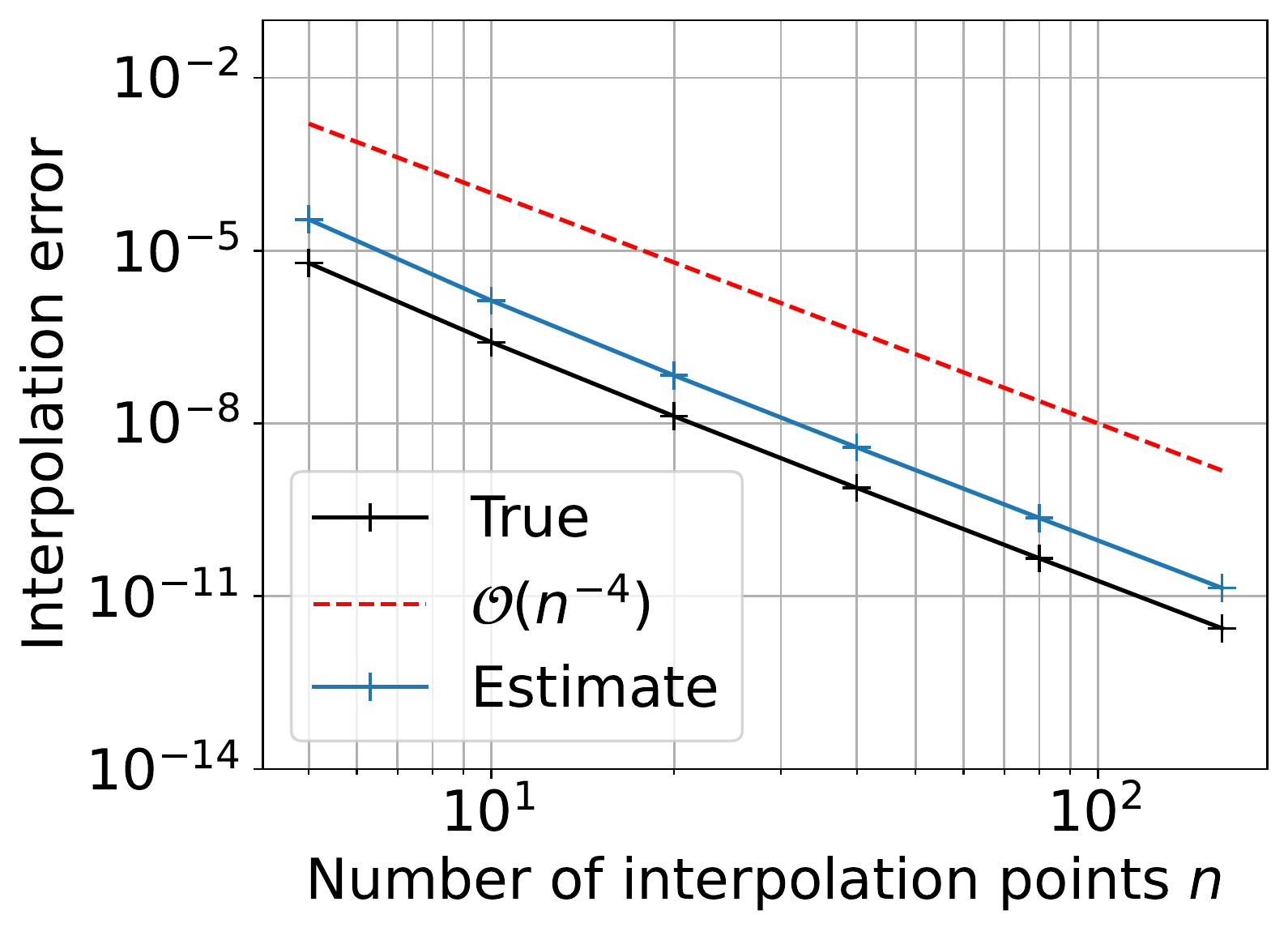}
    \caption{$m=0$}
  \end{subfigure}
    \begin{subfigure}{0.32\textwidth}
    \includegraphics[width=\textwidth]{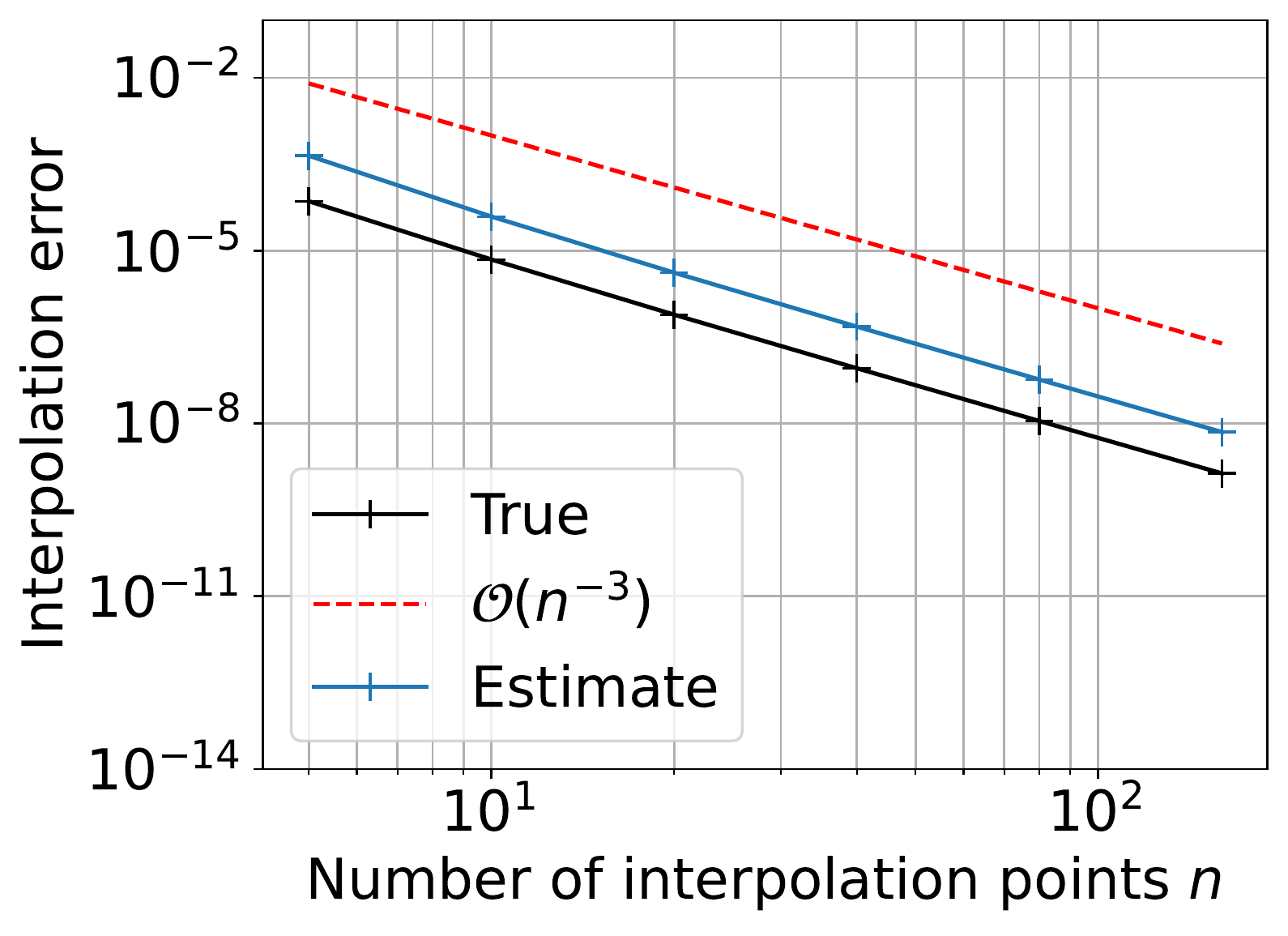}
    \caption{$m=1$}
  \end{subfigure}
    \begin{subfigure}{0.32\textwidth}
    \includegraphics[width=\textwidth]{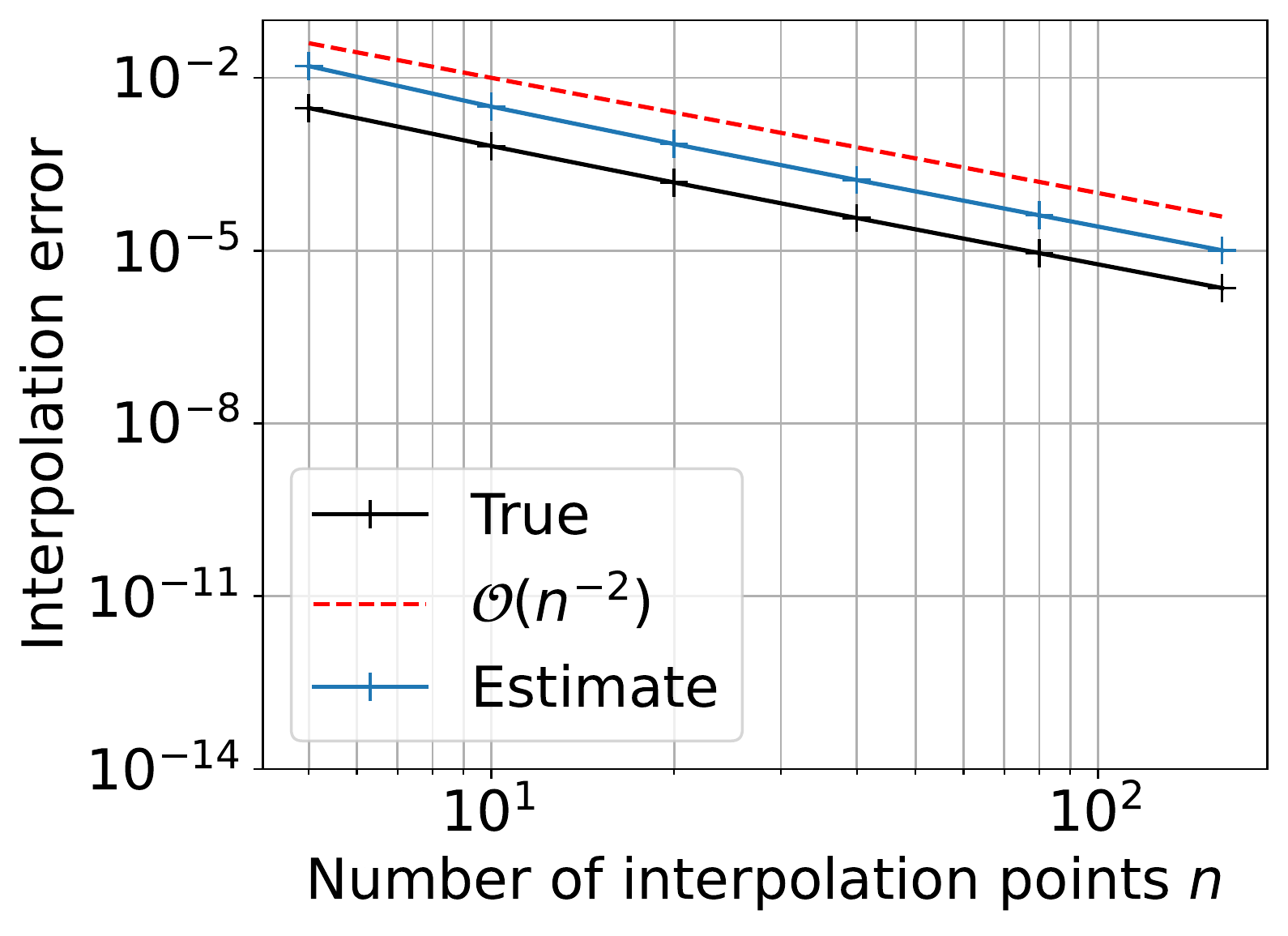}
    \caption{$m=2$}
  \end{subfigure}
  \caption{True interpolation error $\err^{(m)}_{i,tru}$ and interpolation error estimator $\hat{\err}^{(m)}_i$ for different numbers of interpolation points $n$ and different order of derivatives $m$ of $\Phi$.}
  \label{fig:interp_error}
\end{figure}

\subsubsection{Assessment of bias estimators}
To demonstrate the performance of the novel \gls{kde}-based bias estimator $\hat{\err}^{(m)}_{b,new}$ described in Section~\ref{sec:novel_bias}, we compare it with two alternative methods for bias error estimation.
The first method is the fully a priori bias estimator $\errest_b^{(m)}$ given in Eq.~\eqref{eq:apriori_b}.
The second method is to naively estimate the bias using an empirical mean without \gls{kde} smoothing, wherein the pointwise estimates obtained in Eq.~\eqref{eq:est_diff} are directly interpolated.
Note that the number of interpolation points is fixed during this study to $n=10$, which ensures that the interpolation error is much smaller in comparison to the bias as can be inferred from Figs.~\ref{fig:interp_error} and~\ref{fig:bias_error}.
We also fix the level $l=5$ for this study.
The true error is given by 
\begin{align}
\err^{(m)}_{b,tru} = \norm{\Phi_l^{(m)} - \Phi^{(m)}}_{\linf(\Theta)},
\end{align}
which we approximate accurately by evaluating the norm on a very fine grid using the known functions $\Phi_l$ and $\Phi$ and their derivatives.
The a priori estimate, as was described in Section~\ref{sec:apriori}, is given by 
\begin{align}
\errest^{(m)}_{b} \coloneqq \frac{ C_2(m)C_3(n-1)^m}{(e^{\alpha}-1)}\hat{b}_L,
\end{align}
where we used the $N_l$ samples available on level $l$ to estimate the expectation.
The naive non-smoothened estimate is given instead by
\begin{align}
\errest^{(m)}_{b,nai} \coloneqq \frac{1}{(e^{\alpha}-1)} \norm{ \interpm{ \frac{1}{N_l} \sum_{i=1}^{N_l} \phi(\thetab,\qoi^{(i,l)}_l)-\phi(\thetab,\qoi^{(i,l)}_{l-1})}}_{\linf(\Theta)}.\label{eq:bias_naive}
\end{align}
Lastly, the \gls{kde}-smoothened error estimator is given by 
\begin{align}
\errest^{(m)}_{b,new} = \frac{\hat{b}_{L,new}^{(m)}}{(e^{\alpha}-1)},\label{eq:bias_kde}
\end{align}
where $\hat{b}_{L,new}^{(m)}$ is defined in Eq.~\eqref{eq:blm_def} and computed as described in Section~\ref{sec:novel_bias}.
In both Eqs.~\eqref{eq:bias_naive} and~\eqref{eq:bias_kde}, the norm is evaluated on a fine grid with $n'=1000$ points.
For the decay rate $\alpha$, we use the theoretical result that for a hierarchy of meshes whose characteristic mesh size decays as $2^{-l}$ in the levels $l$, the second order central finite difference scheme yields a bias error decay rate of $2^{-2l}$, giving $\alpha = 2 \ln(2) \approx 1.39$.

Fig.~\ref{fig:bias_error} summarises the results on the performance of these bias estimators. 
We plot each of the error estimators as well as the true error for different sample-sizes $N_l$ for fixed $l$ and for different orders of derivatives of $\Phi$.
For each value of $N_l$, we create 20 independent realizations of $N_l$ correlated sample pairs and for each set of $N_l$ correlated sample pairs, we evaluate the different bias error estimators.
We observe that the fully a priori error estimate becomes increasingly conservative for higher order derivatives $m$.
The naive non-smoothened error estimator significantly improves on the a priori estimator but still overestimates the error for higher derivatives and small sample sizes.
This is to be expected since we numerically differentiate a non-smooth function.
The novel \gls{kde}-based approach clearly provides the tightest bound on the true error among the three estimators, consistently for all values of $m$.

\begin{figure}[H]
  \begin{subfigure}{0.32\textwidth}
    \includegraphics[width=\textwidth]{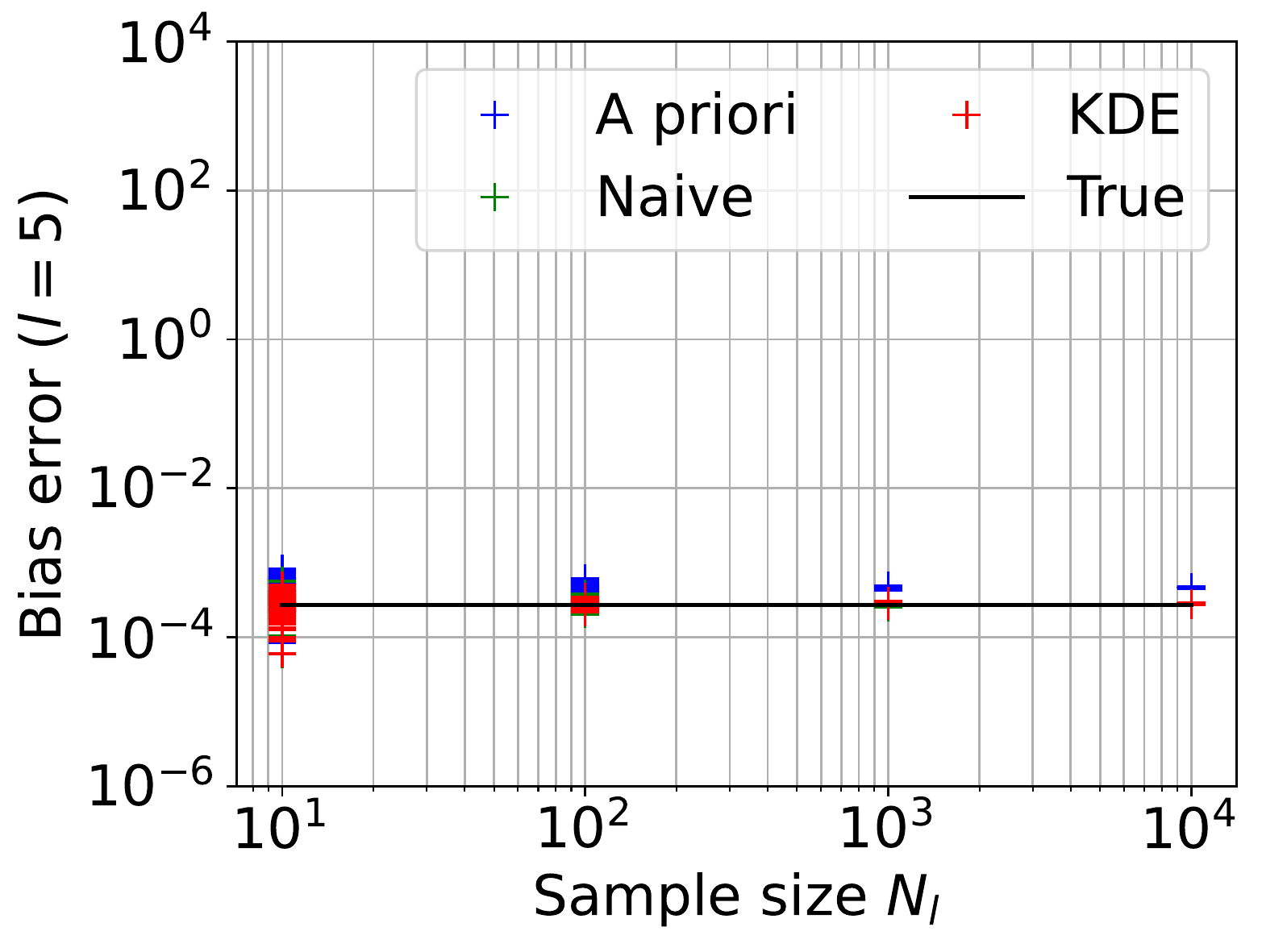}
    \caption{$m=0$}
  \end{subfigure}
    \begin{subfigure}{0.32\textwidth}
    \includegraphics[width=\textwidth]{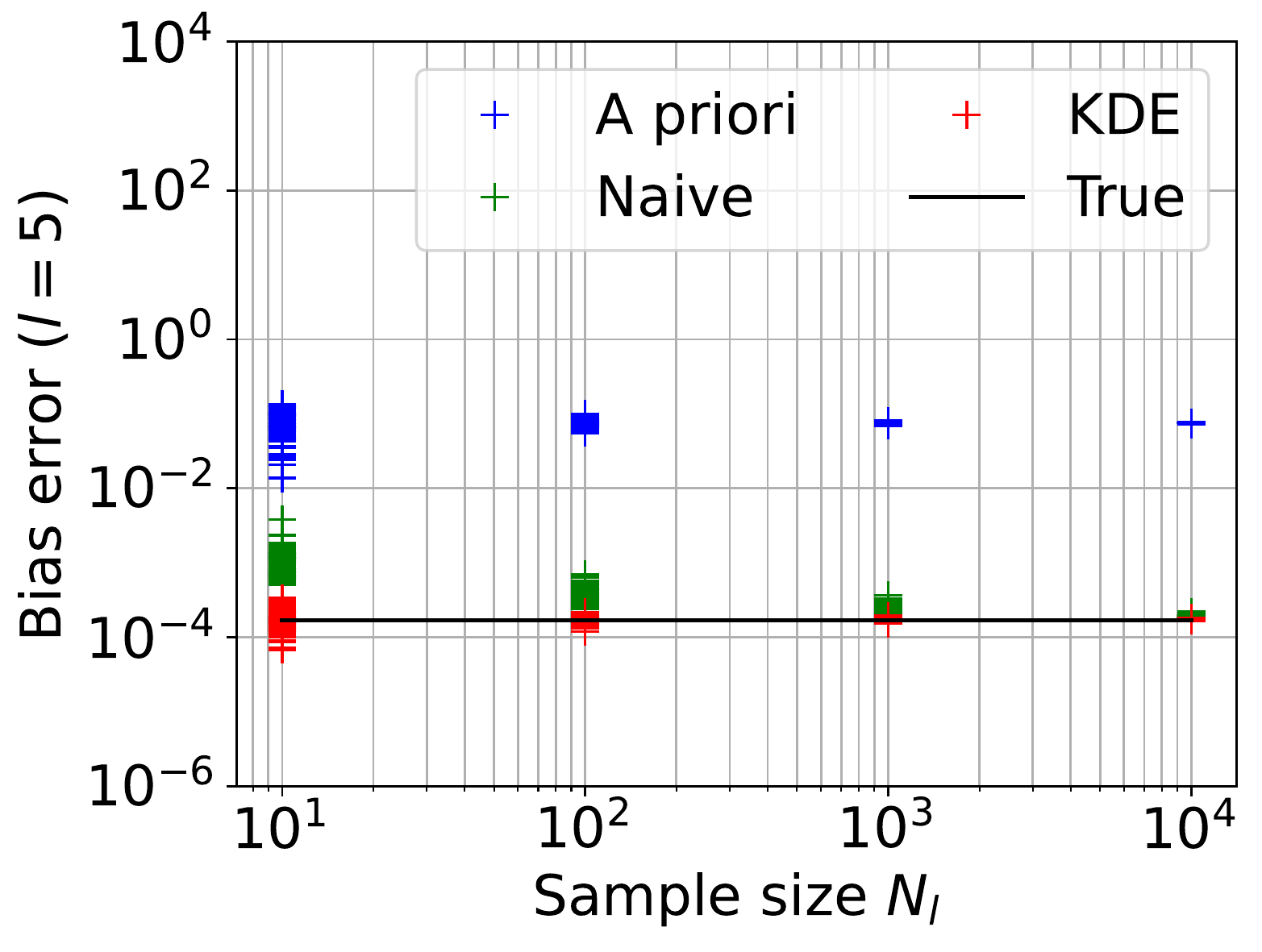}
    \caption{$m=1$}
  \end{subfigure}
    \begin{subfigure}{0.32\textwidth}
    \includegraphics[width=\textwidth]{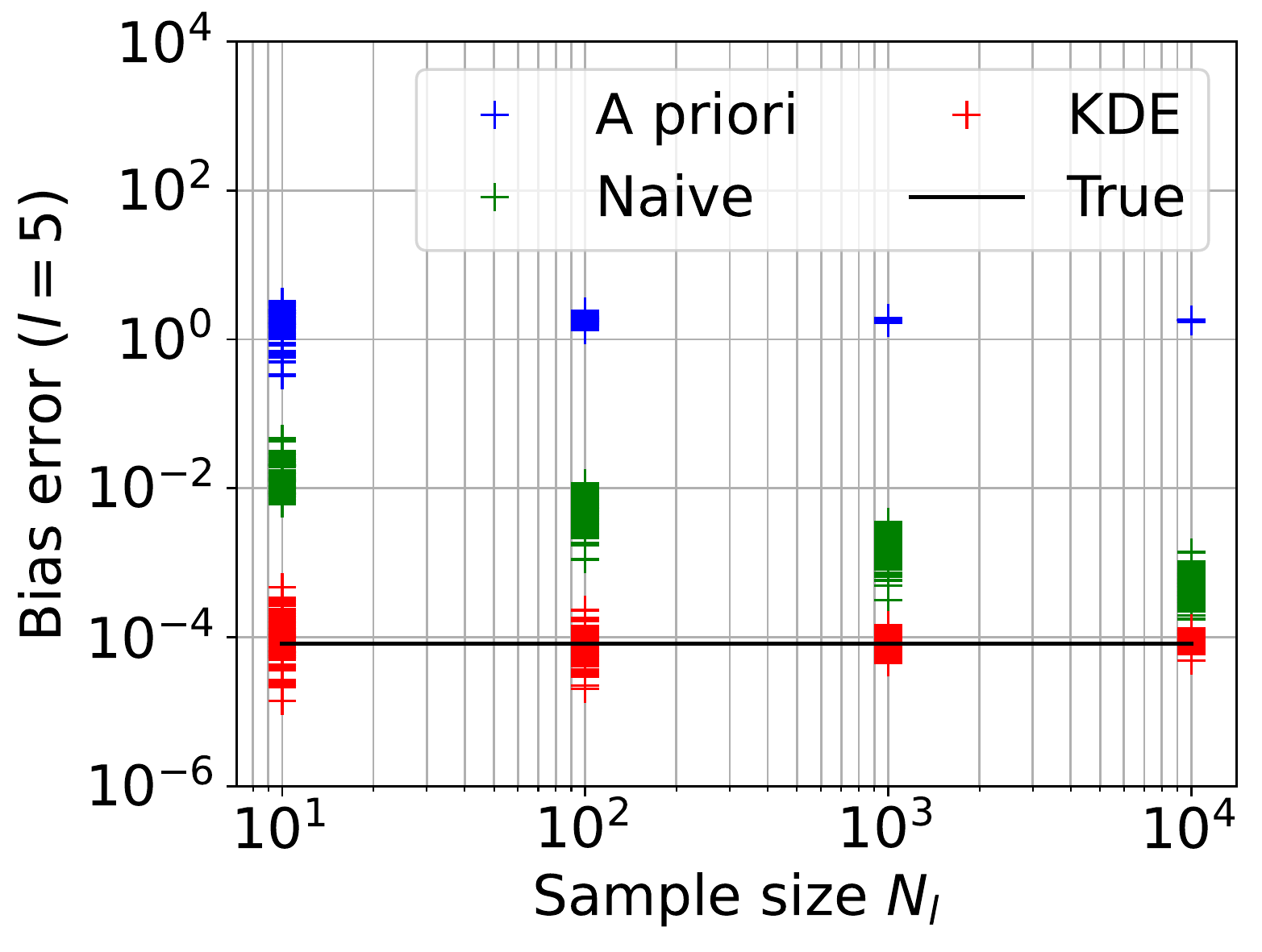}
    \caption{$m=2$}
  \end{subfigure}
  \caption{Comparison of bias error estimators $\errest^{(m)}_{b}$, $\errest^{(m)}_{b,nai}$ and $\errest^{(m)}_{b,new}$ for different sample sizes and for different derivatives of $\Phi$.}
  \label{fig:bias_error}
\end{figure}

Moreover, the \gls{kde}-based approach preserves the underlying decay rate of the \gls{qoi} with respect to the level $l$.
To illustrate this, we compute a hierarchy that uses 100 samples per level for 5 levels.
A hierarchy of this size is sampled independently 20 times and the resulting bias estimates are plotted against levels for each realisation of the hierarchy.
These results are summarised in Fig.~\ref{fig:bias_error_decay}.
The average least-squares-fit decay rate over these 20 simulations is computed and shown in the corresponding figure legend. 
We reiterate that theoretical considerations predict that the bias decays at a rate $\alpha = 2 \log(2) \approx 1.39$, independent of the order $m$ of the derivative.
As can be seen from the figure, the a priori estimates capture the correct decay rate but are conservative on the true error.
In addition, they become drastically more conservative for higher order derivatives.
The naive approach provides a much tighter bound than the a priori approach, but its decay rate deteriorates at least for the second order derivative.
The \gls{kde} approach, on the other hand, provides the tightest bound while also preserving the correct underlying decay rate. 

\begin{figure}[H]
  \begin{subfigure}{0.32\textwidth}
    \includegraphics[width=\textwidth]{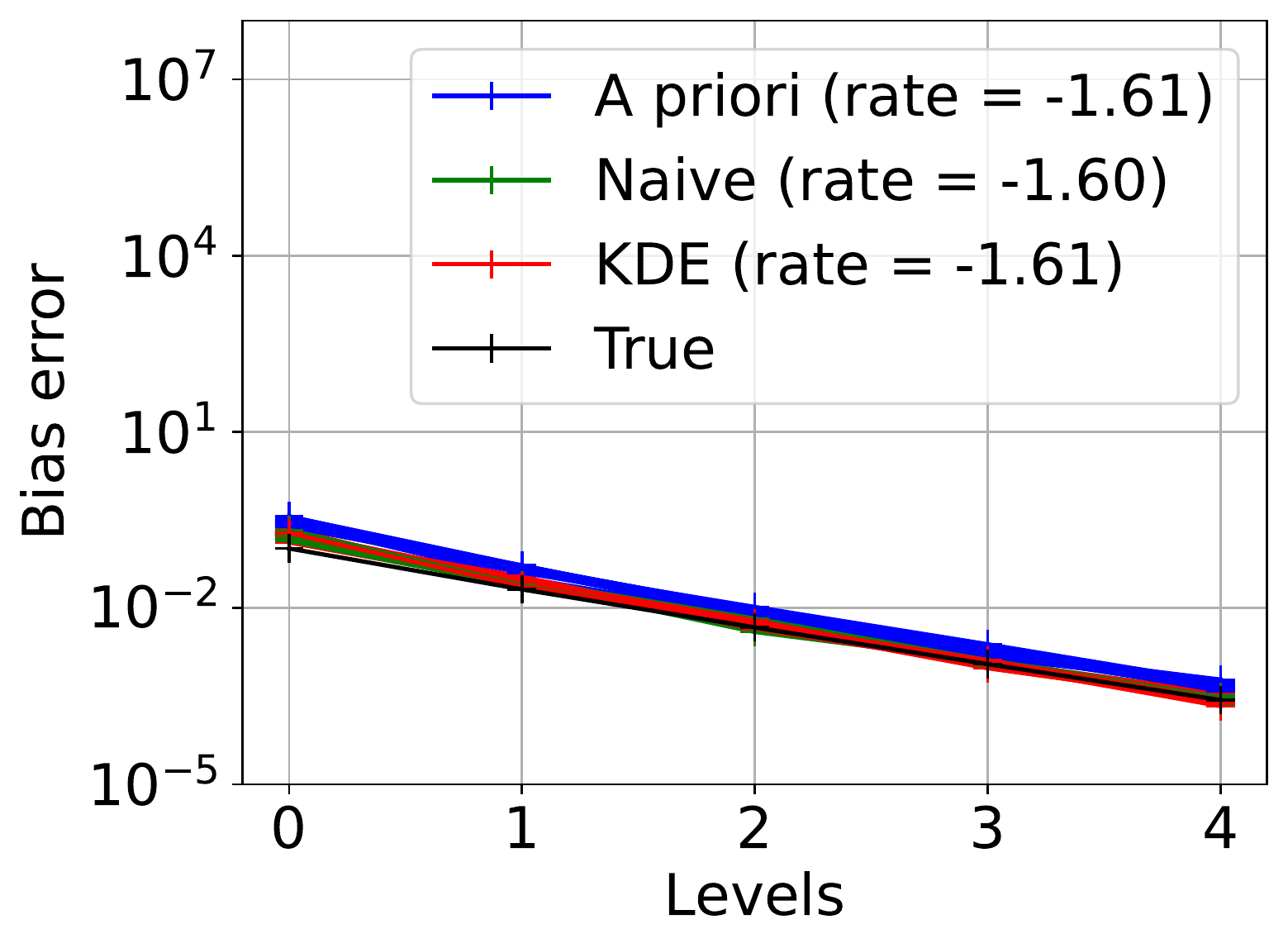}
    \caption{$m=0$}
  \end{subfigure}
    \begin{subfigure}{0.32\textwidth}
    \includegraphics[width=\textwidth]{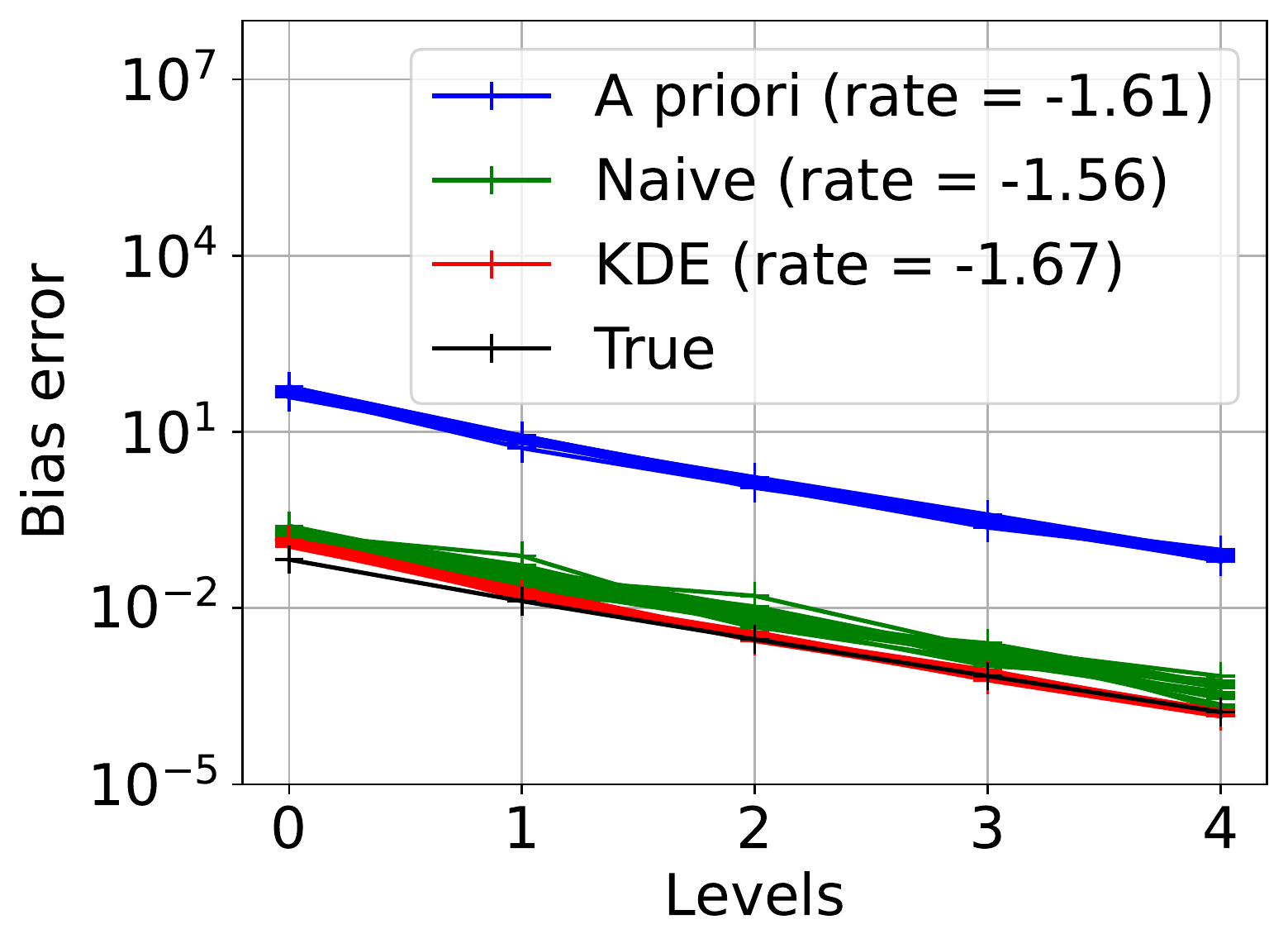}
    \caption{$m=1$}
  \end{subfigure}
    \begin{subfigure}{0.32\textwidth}
    \includegraphics[width=\textwidth]{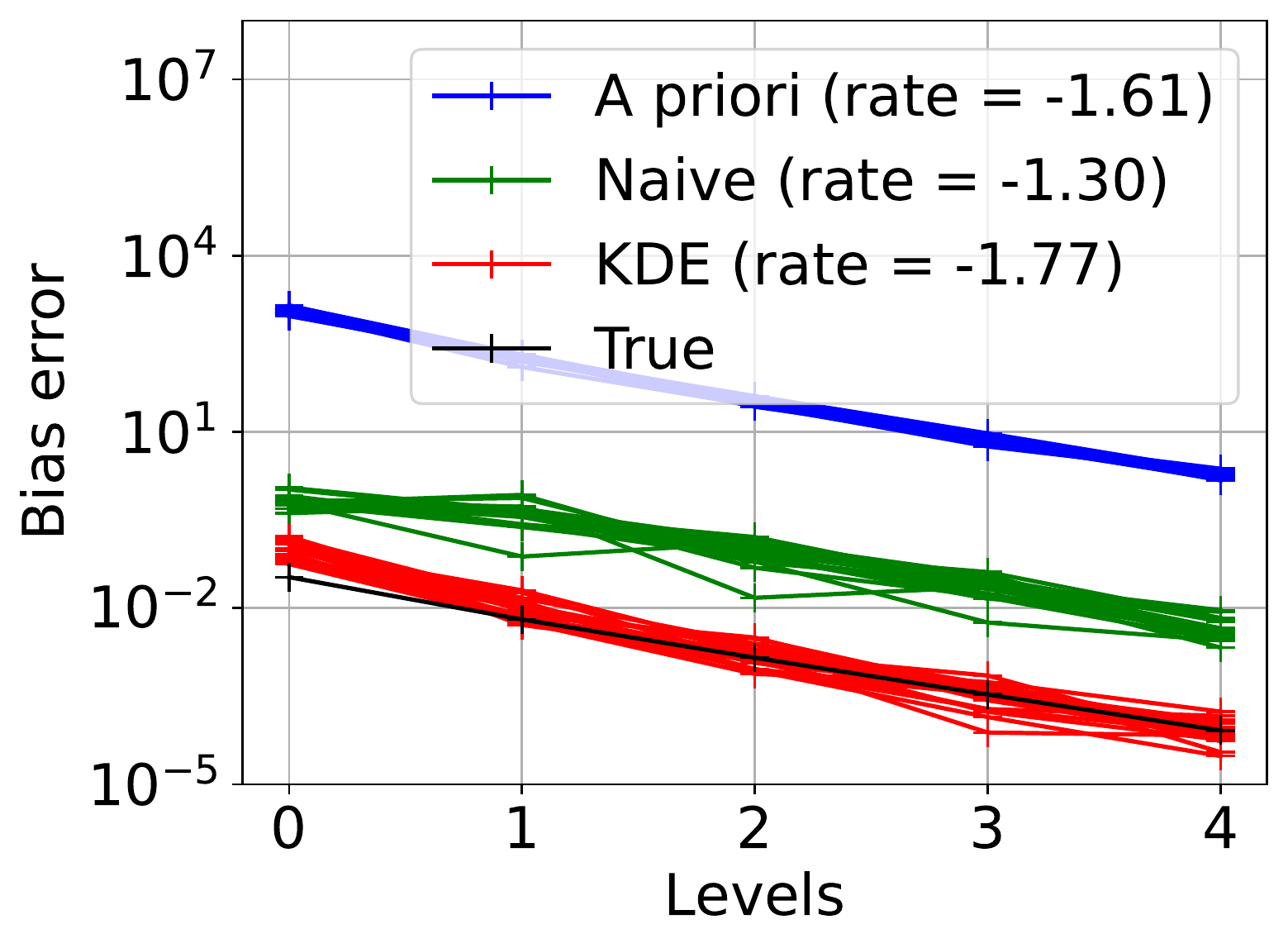}
    \caption{$m=2$}
  \end{subfigure}
  \caption{Comparison of bias decay over levels for different derivatives of $\Phi$.}
  \label{fig:bias_error_decay}
\end{figure}

\subsubsection{Statistical error comparison} \label{sec:stat_err_comparison}
The statistical error is controlled by the level-wise sample sizes.
To assess the quality of the novel statistical error estimator that was introduced in Section~\ref{sec:stat_err_est}, we consider three types of hierarchies, all of the general form 
\begin{align}
N_l = \left\lfloor N_0 2^{rl} \right\rfloor, \quad l \in \{0,1,...,L\}.\label{eq:hierarchy}
\end{align}
For the simulations performed in this section, we fix $L=5$ and consider $r\in \{-1,0,1\}$. 
These values of $r$ generate hierarchies where $N_l$ decreases, stays the same, and increases with $l$, respectively. 
Although hierarchies where $N_l$ increases with $l$ do not occur in practice, they are nevertheless investigated here to assess the robustness of the error estimator. 
By selecting different values of $N_0$, one can determine the hierarchy fully.

We propose estimating the statistical error as follows.
The true statistical error $e^{(m)}_{s,tru}$, which is not readily computable in this case, is estimated using a brute force strategy where we repeat the \gls{mlmc} procedure $N_{ref}=10^4$ times.
The squared true error is then estimated as:
\begin{align}
(e^{(m)}_{s,tru})^2 \approx \frac{1}{N_{ref}} \sum_{i=1}^{N_{ref}} \norm{\interpm{\hat{\Phi}_{L,i}-\Phi_L}}_{\linf(\Theta)}^2,
\end{align}
where $\hat{\Phi}_{L,i}$ denotes the $i^{\text{th}}$ simulation of an \gls{mlmc} hierarchy whose parameters are given by Eq.~\eqref{eq:hierarchy} for a given $r$ and $N_0$.
For the novel bootstrapped statistical error estimate introduced in Section~\ref{sec:stat_err_est}, we create one realisation of the hierarchy in Eq.~\eqref{eq:hierarchy} and create $N_{bs}=100$ bootstrapped realisations of the resulting \gls{mlmc} estimator.
The statistical error estimate $\errest^{(m)}_{s,new}$ is then computed as described in Section~\ref{sec:stat_err_est} and in Eq.~\eqref{eq:se_bs}.
We do not change $N_{bs}$ adaptively in this demonstration and keep the value fixed.
We perform the above study for the three different types of hierarchy in Eq.~\eqref{eq:hierarchy}, and for the first two derivatives of $\Phi$, as well as for $\Phi$ itself. 
For each hierarchy shape and derivative, we test for different values of the hierarchy size parameter $N_0$.

The results are shown in Fig.~\ref{fig:stat_error}.
The a priori statistical error estimate given by Eq.~\eqref{eq:apriori_s} and the bootstrapped statistical error estimate given by Eq.~\eqref{eq:se_bs} are plotted alongside the true statistical error for decreasing, uniform and increasing hierarchies and for different derivatives of $\Phi$.
As can be observed from the figure, the bootstrapped statistical error estimate provides a tight bound on the true error for the range of hierarchies and derivatives tested, whereas the a priori statistical error estimator defined in Eq.~\eqref{eq:apriori_s} clearly provides overly conservative estimates for $m=1,2$.


\begin{figure}[H]
\centering
\begin{subfigure}{.32\textwidth}
    \centering
    \includegraphics[width=\textwidth]{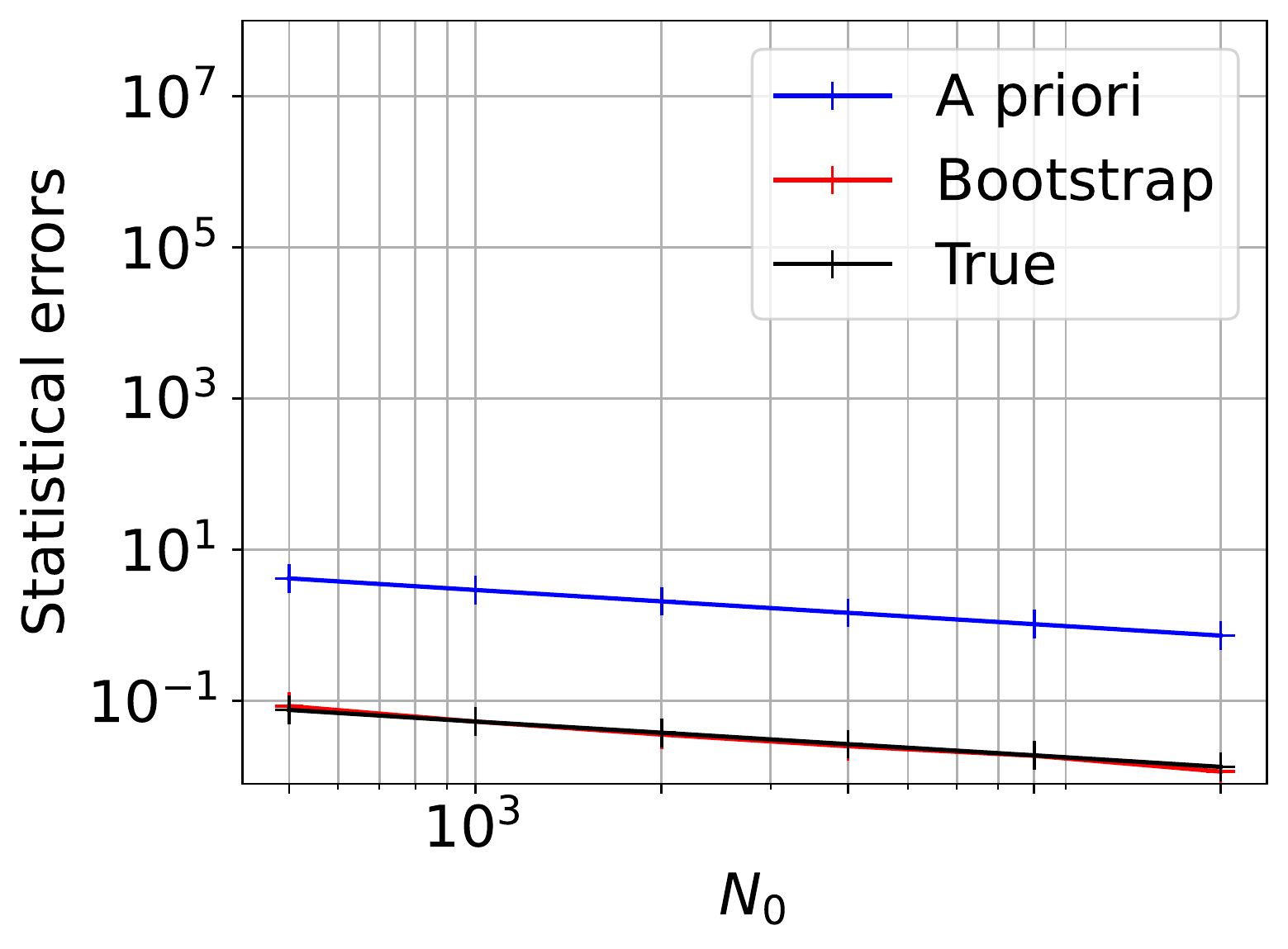}
    \caption{$m=0, r=-1$}
\end{subfigure}
\begin{subfigure}{.32\textwidth}
    \centering
    \includegraphics[width=\textwidth]{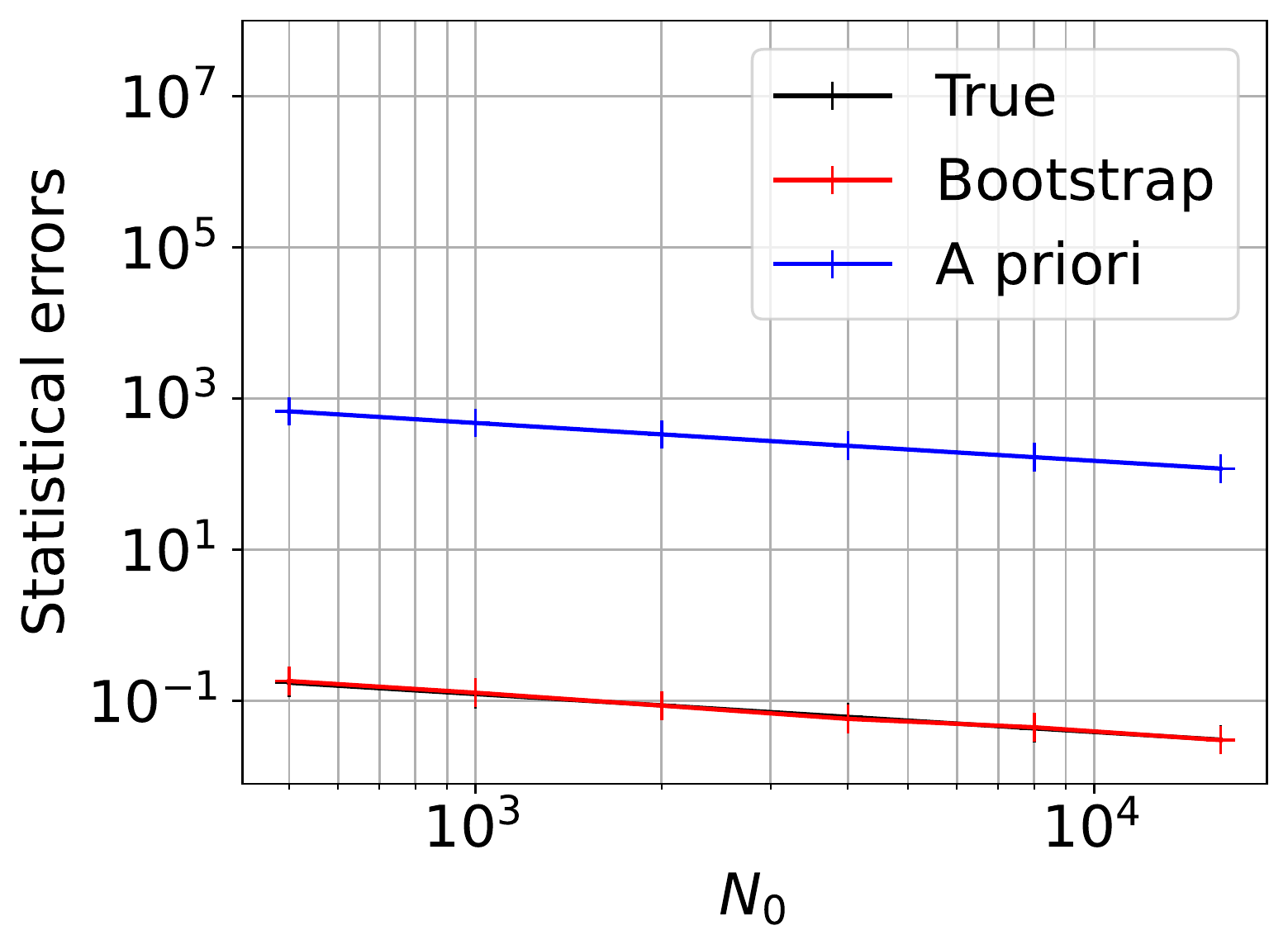}
    \caption{$m=1, r=-1$}
\end{subfigure}
\begin{subfigure}{.32\textwidth}
    \centering
    \includegraphics[width=\textwidth]{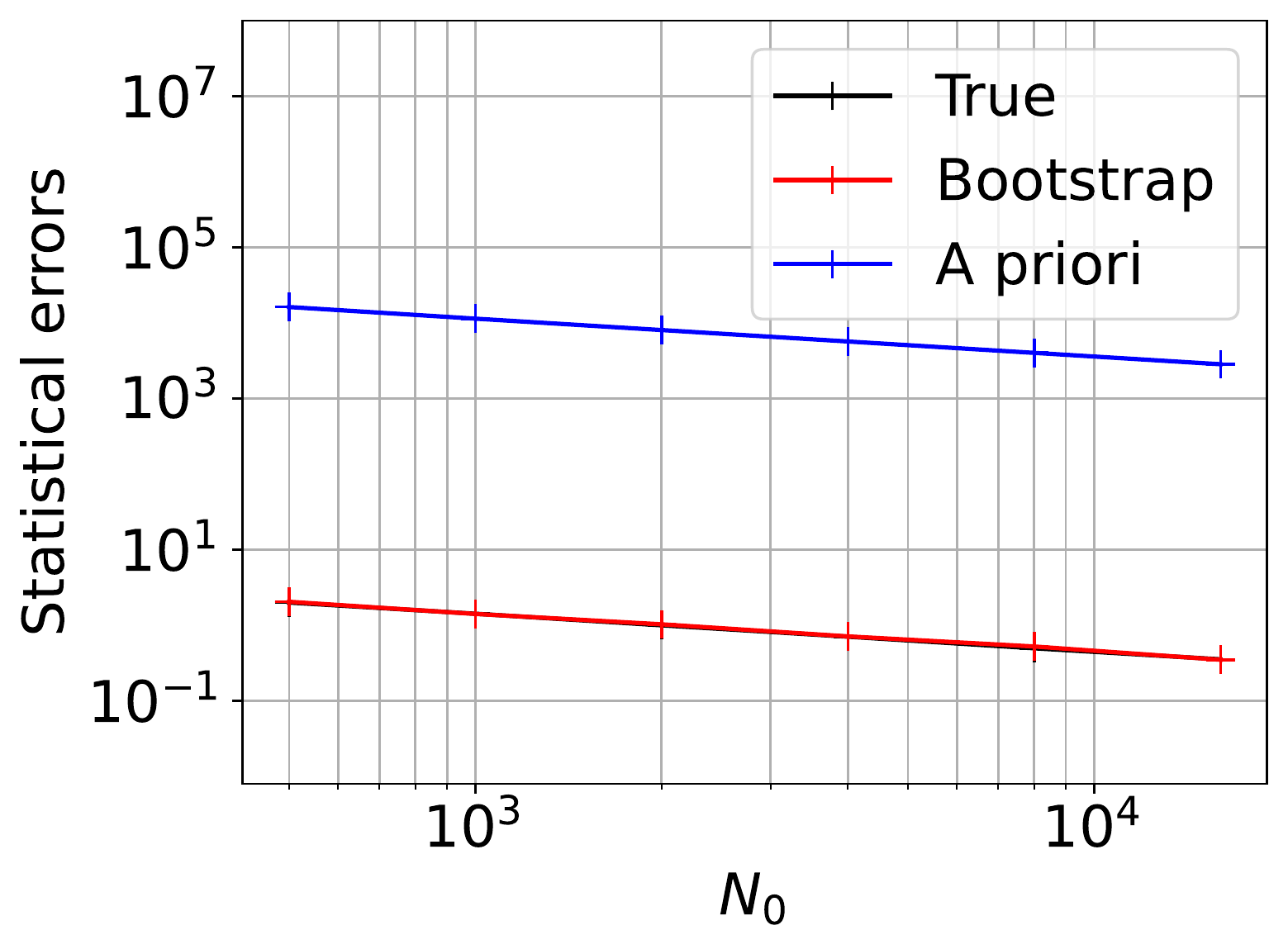}
    \caption{$m=2, r=-1$}
\end{subfigure}

\begin{subfigure}{.32\textwidth}
    \centering
    \includegraphics[width=\textwidth]{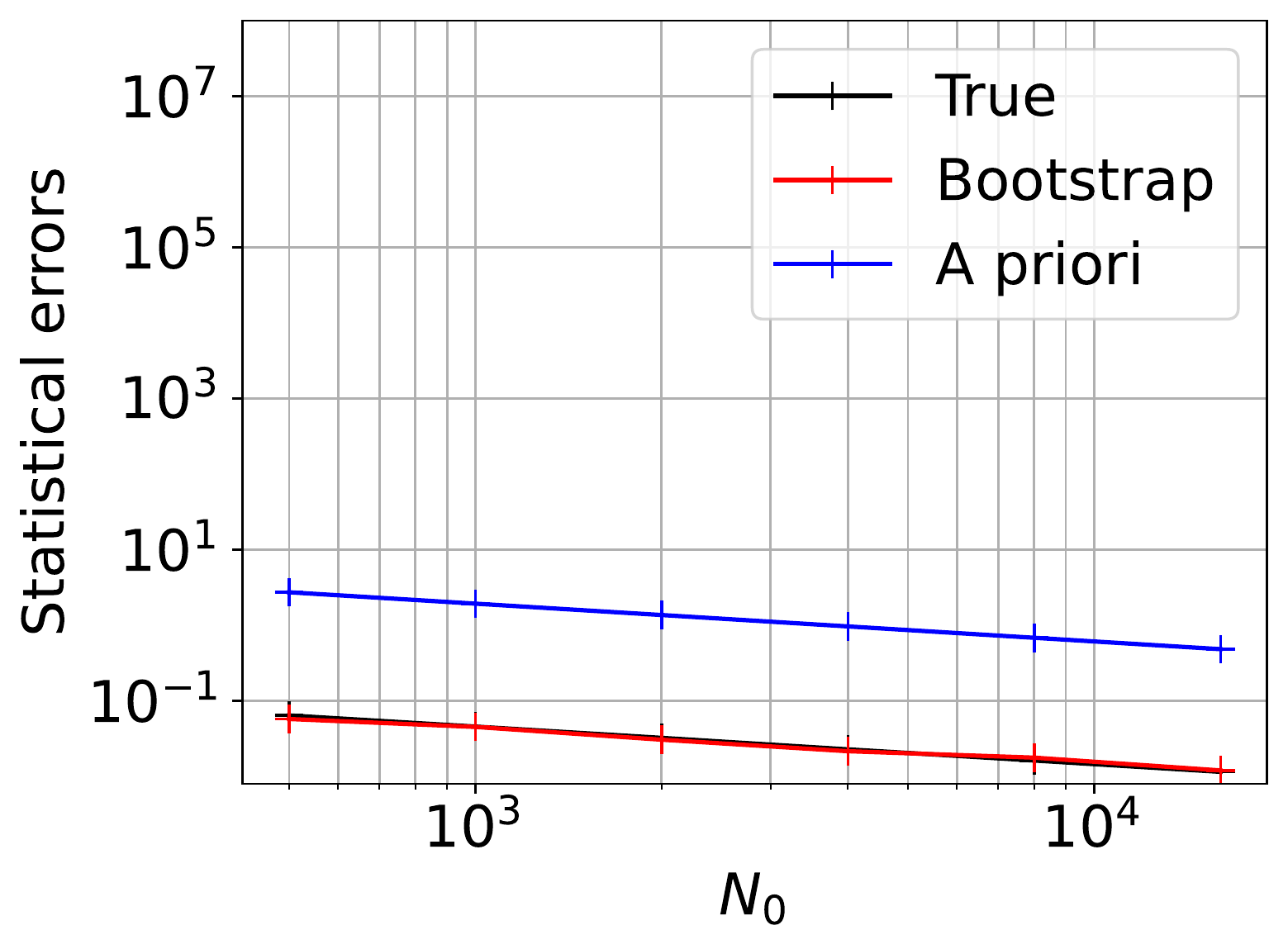}
    \caption{$m=0, r=0$}
\end{subfigure}
\begin{subfigure}{.32\textwidth}
    \centering
    \includegraphics[width=\textwidth]{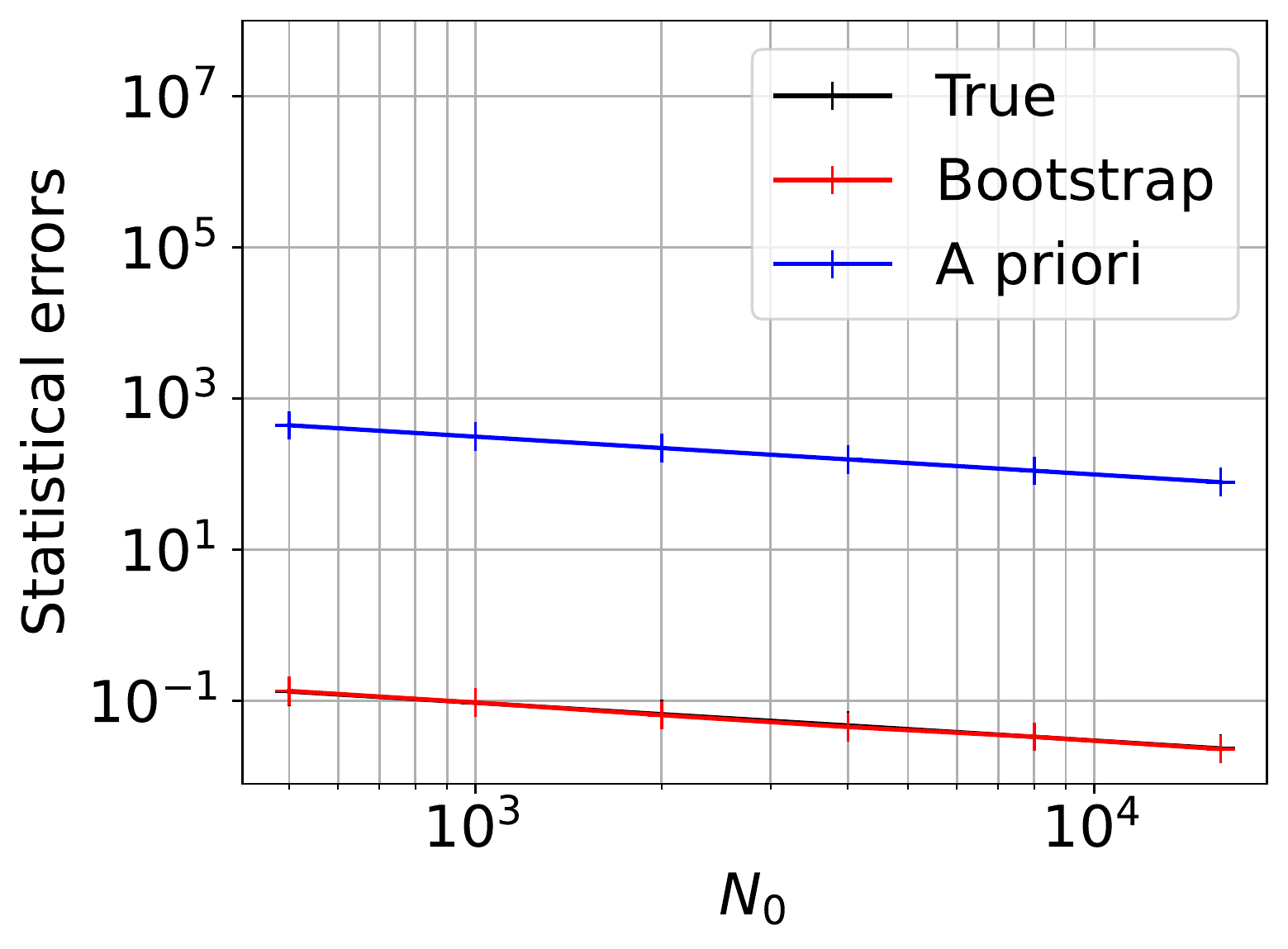}
    \caption{$m=1, r=0$}
\end{subfigure}
\begin{subfigure}{.32\textwidth}
    \centering
    \includegraphics[width=\textwidth]{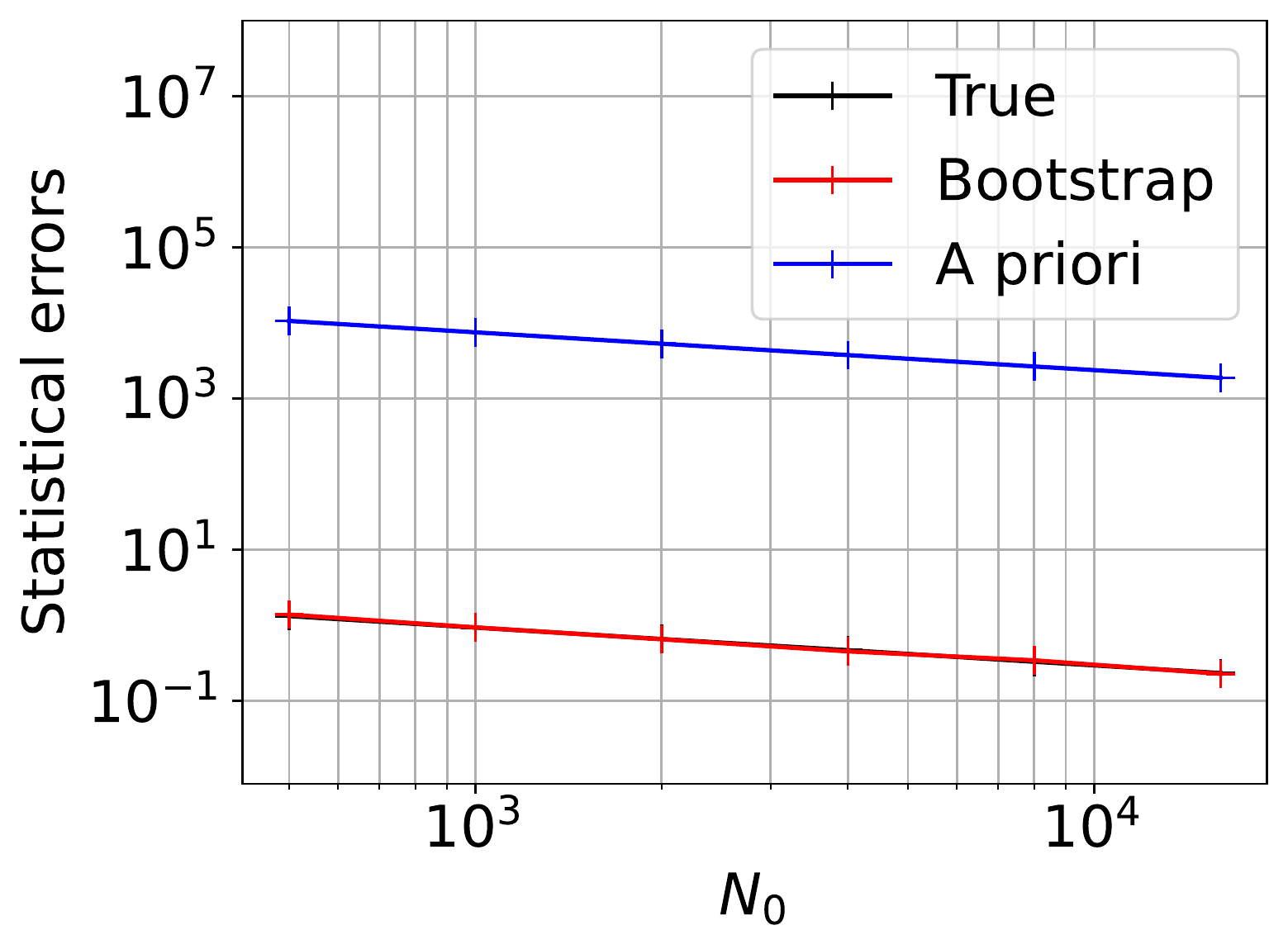}
    \caption{$m=2, r=0$}
\end{subfigure}

\begin{subfigure}{.32\textwidth}
    \centering
    \includegraphics[width=\textwidth]{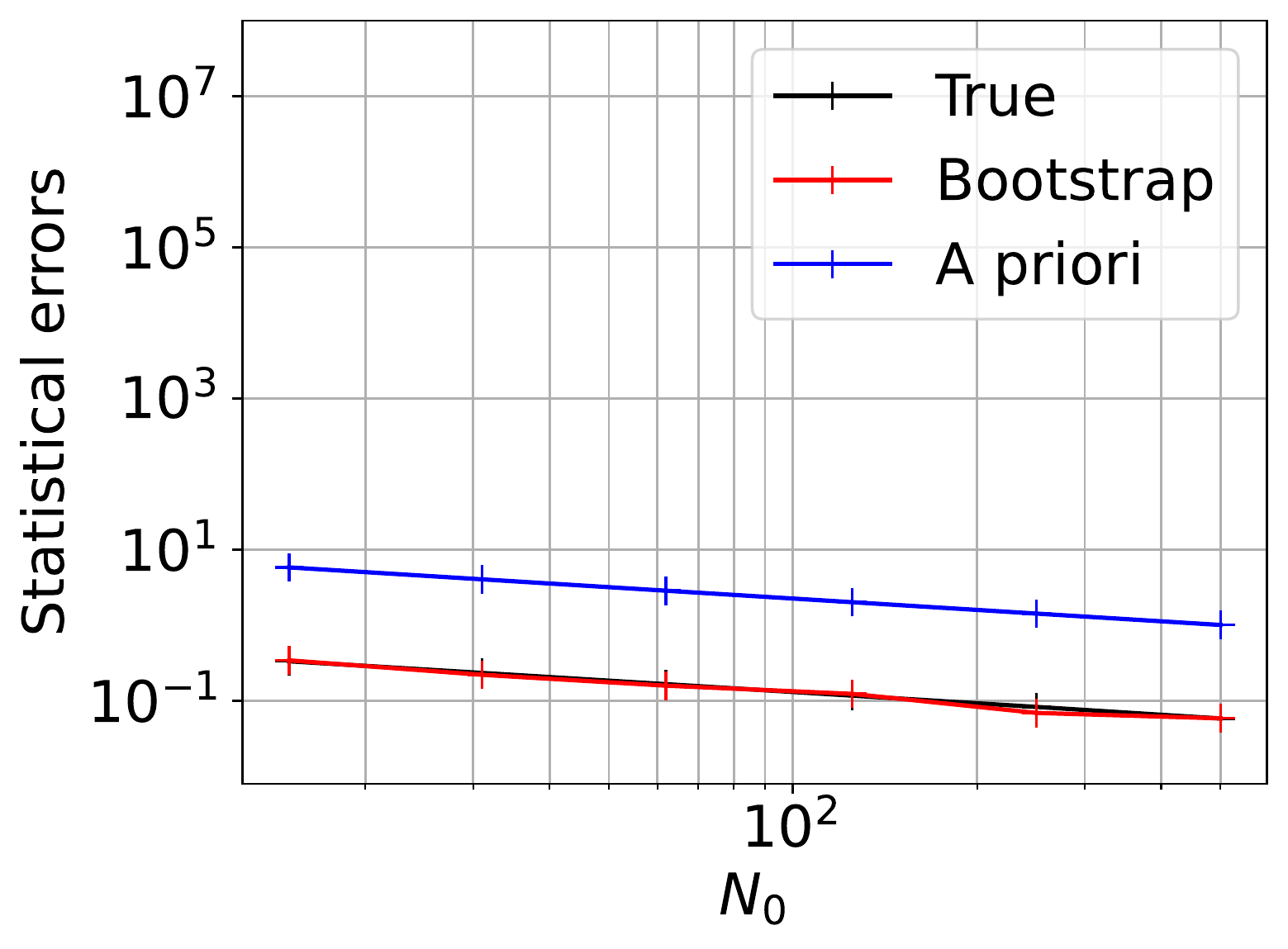}
    \caption{$m=0, r=1$}
\end{subfigure}
\begin{subfigure}{.32\textwidth}
    \centering
    \includegraphics[width=\textwidth]{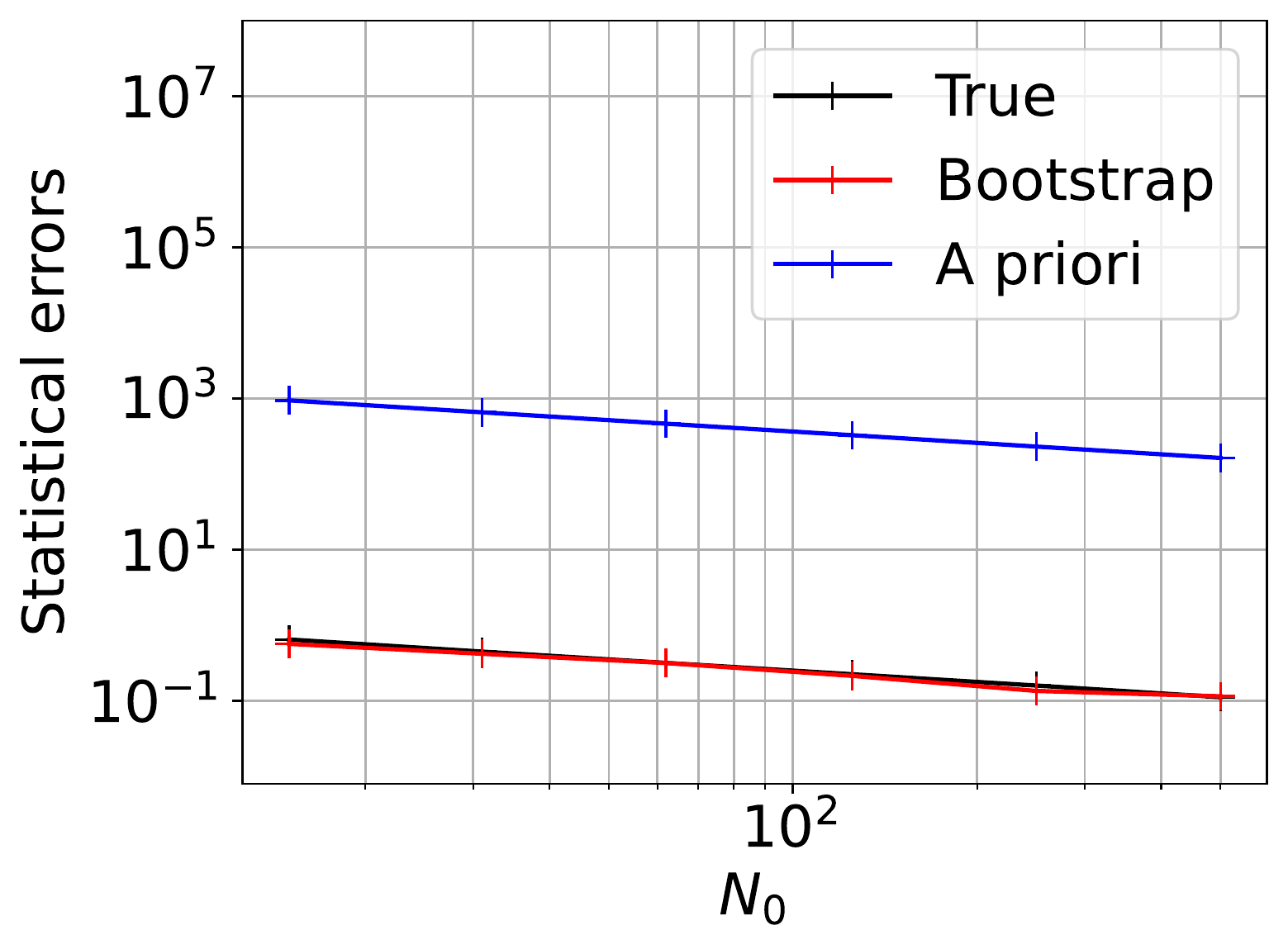}
    \caption{$m=1, r=1$}
\end{subfigure}
\begin{subfigure}{.32\textwidth}
    \centering
    \includegraphics[width=\textwidth]{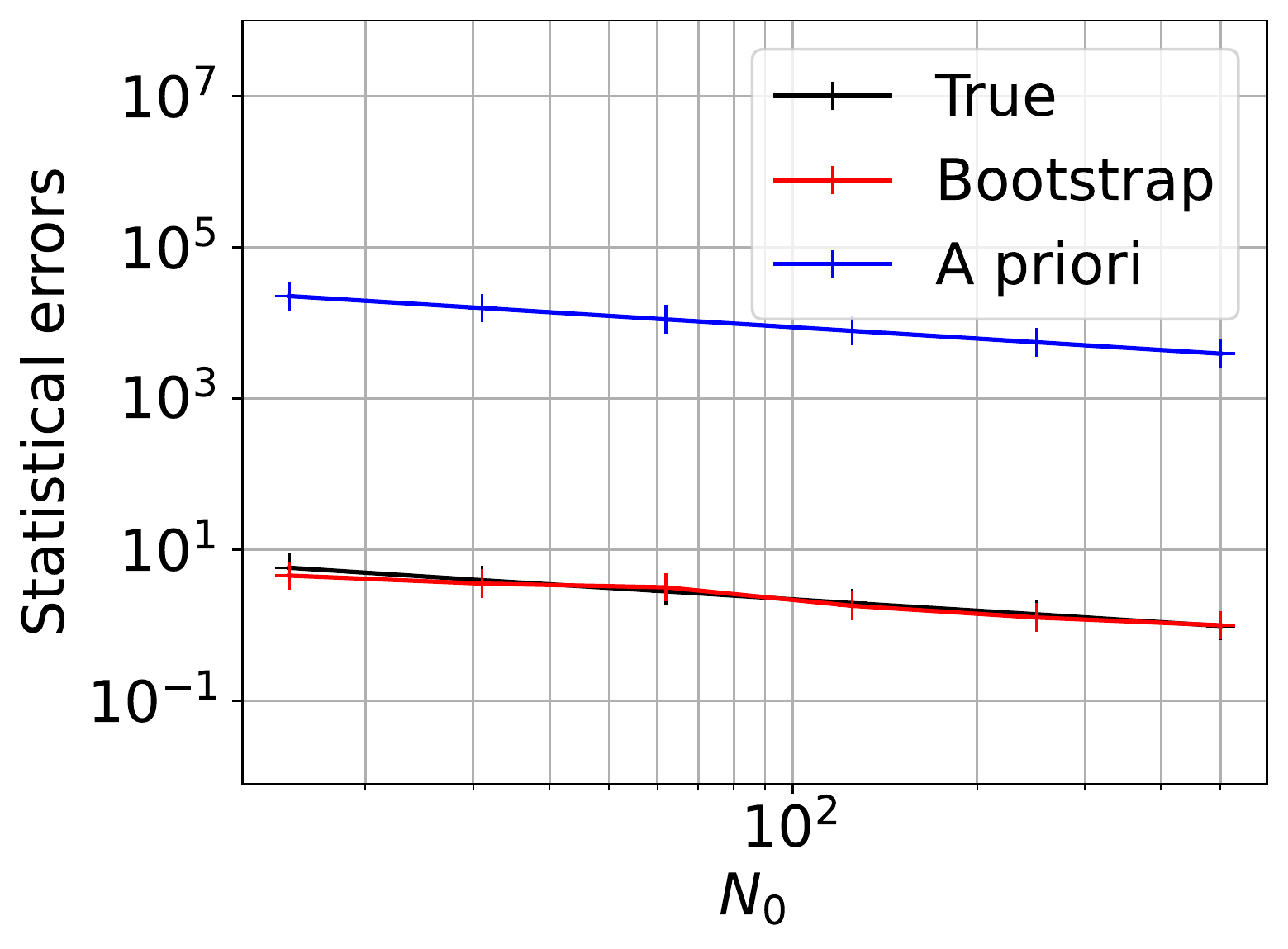}
    \caption{$m=2, r=1$}
\end{subfigure}
\caption{Statistical error estimator comparison of $\errest^{(m)}_{s}$ and $\errest^{(m)}_{s,new}$. From left to right, increasing order of derivative $m$ of $\Phi$. From top to bottom, decreasing, uniform and increasing hierarchies of different sizes.}
\label{fig:stat_error}
\end{figure}

\section{Tuning the MLMC hierarchy using the novel error estimators} \label{sec:adaptivity}
In the previous sections, we have presented effective error estimators for the \gls{mse} contributions of the \gls{mlmc} estimator of $\Phim$.
The next step will be to adapt the \gls{mlmc} hierarchy based on such an error estimator to achieve a prescribed tolerance on the \gls{mse} in a cost-optimal way.
This implies that one should choose adaptively the number of interpolation points $n$, the maximum discretisation level $L$ and the level-wise sample sizes $N_l$.
We discuss a possible way to do this for the \gls{mse} on $\Phiem$ and quantities derived from it, focussing on the \gls{pdf}, the \gls{cdf}, the \gls{var} and the \gls{cvar} as defined in Eq.~\eqref{eq:stat_defs}.
The procedure is, of course, similar to the adaptive procedure described in Section~\ref{sec:apriori} for the fully a priori error estimators, but tailored here to the current setting.


\subsection{\gls{mlmc} tuning procedure for linear combinations of \glsplural{mse}}
It can be seen from Eq.~\eqref{eq:stat_defs} that the \gls{mse} of the \gls{cdf} and the \gls{pdf} are directly proportional to the \gls{mse} of $\Phieone$ and $\Phietwo$ respectively.
In addition, as will be shown in Section~\ref{sec:var_cvar_error_bound}, Lemma~\ref{lemma:accuracy:derived}, the \gls{mse} of the \gls{var} and of the \gls{cvar} can be bounded by a linear combination of the \glsplural{mse} of $\Phie$ and $\Phieone$.
For these reasons, we first present a method to calibrate the \gls{mlmc} estimator of any arbitrary quantity $\stat$ with corresponding estimate $\statest$, whose \gls{mse} can be bounded by linear combinations of the form
\begin{align}
\mse{\statest} \leq k_0 \mse{\Phie} + k_1 \mse{\hat{\Phi}_L^{(1)}} + k_2 \mse{\hat{\Phi}_L^{(2)}},\; k_0, k_1, k_2 \geq 0.\label{eq:three_sums_mse}
\end{align}
Each of the three terms decomposes into its three respective error contributions, which can then be combined to yield global interpolation, bias and statistical error contributions: 
\begin{align}
\mse{\statest} \leq 3 \bigg\{ &\underbrace{\left[\sum_{m=0}^2 k_m (e^{(m)}_{i})^2 \right]}_{\text{Squared interpolation error}} +\underbrace{\left[\sum_{m=0}^2 k_m (e^{(m)}_{b})^2 \right]}_{\text{Squared bias error}} +\underbrace{\left[\sum_{m=0}^2 k_m (e^{(m)}_{s})^2 \right]}_{\text{Squared statistical error}} \bigg\}.\label{eq:general_mse}
\end{align}
We require the \gls{mse} to satisfy a tolerance $\epsilon^2$ with each of the three squared error contributions on the right-hand side of Eq.~\eqref{eq:general_mse} satisfying their corresponding tolerances $\epsilon_i^2$, $\epsilon_b^2$ and $\epsilon_s^2$ as defined in Eq.~\eqref{eq:tol_split}.
In addition, each of the terms $\err_{i}^{(m)}$, $\err_{b}^{(m)}$ and $\err_{s}^{(m)}$ is estimated using the error estimators $\errest_{i}^{(m)}$, $\errest_{b,new}^{(m)}$ and $\errest_{s,new}^{(m)}$ defined in Eqs.~\eqref{eq:apriori_i}, \eqref{eq:blm_def} and \eqref{eq:se_bs} respectively.

As described in Section~\ref{sec:apriori}, the interpolation error is controlled solely by the number of interpolation points.
To determine it, we require that the squared interpolation error term in Eq.~\eqref{eq:general_mse} satisfies the condition:
\begin{align}
\norm{\Upsilon^{(4)}_{\lceil L/2 \rceil}}_{\linf(\Theta)}^2 \left[\sum_{m=0}^2 k_m C^2_1(m) \left( \frac{|\Theta|}{n} \right)^{2(4-m)}\right] \leq \epsilon^2_{i}, \label{eq:interp_points_all}
\end{align}
which results from each term $\err^{(m)}_i$ in the interpolation error contribution of Eq.~\eqref{eq:general_mse} being bounded by its estimator $\errest^{(m)}_i$ in Eq.~\eqref{eq:apriori_i}.
Determining $n$ such that equality holds in Eq.~\eqref{eq:interp_points_all} requires finding the roots of a polynomial of the form $a_8 n^{-8} + a_6 n^{-6} +a_4 n^{-4}+a_0=0$ for $a_8, a_6, a_4, a_0\geq 0$.
In case of multiple real roots, the smallest positive root is taken and the optimal number of interpolation points $n^*$ is, in practice, taken to be the smallest integer larger than this root. 

The bias error is controlled by the number of levels $L$.
To determine it, we first estimate the level-wise bias terms $\hat{b}_{l,new}^{(m)},\;l \in \{1,...,L\}$ for $m\in\{0,1,2\}$.
We then enforce the condition that the squared bias error term in Eq.~\eqref{eq:general_mse}, with each of the terms $\err_b^{(m)}$ replaced by the corresponding estimator $\errest_{b,new}^{(m)}$, satisfies the tolerance $\epsilon_b^2$.
However, we recall that bias estimates are unavailable for levels $l > L$ and are available only on levels where samples have already been computed.
As a result, to determine the optimal choice of level $L^*$, the bias decay models $c_{\alpha_m} e^{-l \alpha_m}$ are first constructed by least squares fits respectively on the level-wise bias estimates $\hat{b}_{l,new}^{(m)},\;l \in \{1,...,L\}$ for each $m\in\{0,1,2\}$.
We then require the squared bias error term, with the terms $\hat{b}_{l,new}^{(m)}$ in $\errest_{b,new}^{(m)}$ replaced by the corresponding model $c_{\alpha_m}e^{-l\alpha_m}$, to satisfy the following conditions:
\begin{align}
\sum_{m=0}^2 \frac{k_m c_{\alpha_m}^2 e^{-2L \alpha_m}}{(e^{\alpha_m}-1)^2} \leq \epsilon_b^2, \label{eq:all_biases_adapt}
\end{align}
on $L$. 
The appropriate choice of level $L^*$ is selected to be the minimum level that satisfies the above condition in Eq.~\eqref{eq:all_biases_adapt}.
Although all three rates $\alpha_m, \; m \in \{0,1,2\}$ are expected to be the same in most cases, small differences can exist in practice due to estimation errors.



To select the appropriate level-wise sample sizes $N_l$, we are required to localize the bootstrapped statistical error estimator $\errest^{(m)}_{s,new}$ over the levels $l$.
We propose an algorithm to accomplish this based on rescaling the level-wise variances $\hat{V}_l$ defined in Eq.~\eqref{eq:vl_est}, thus preserving the same squared statistical error splitting between levels as in the case of the a priori error estimator.
We first discuss the case of a single term in Eq.~\eqref{eq:general_mse} ($k_m = 1, k_j = 0, j \neq m$).
Recall from Section~\ref{sec:apriori} that the level-wise sample sizes $N_l$ are selected to minimise the total cost subject to the constraint that the statistical error $\err^{(m)}_{s}$ satisfies a given tolerance.
The a priori error estimator $\errest_s^{(m)}$ in Eq.~\eqref{eq:apriori_s} is naturally split over the levels as 
\begin{align}
(\errest^{(m)}_s)^2 = \sum_{l=0}^{L} K(n,m) \frac{\hat{V}_l}{N_l}, \label{eq:apriori_gen_stat}
\end{align}
where $K(n,m) = C_2^{2}(m)C_3^2 (n-1)^{2m} c(n)$, $\{N_l\}_{l=0}^L$ denotes the hierarchy based on which $\errest^{(m)}_{s}$ is computed and $\hat{V}_l$ is defined as in Eq.~\eqref{eq:vl_est} and hence, is independent of $m$.
On the other hand, the new error estimator $\hat{\err}^{(m)}_{s,new}$ based on bootstrapping does not provide such an immediate notion of how each level contributes to the overall statistical error and as such provides only a global statistical error estimate. 
To overcome this limitation, we propose using the error splitting structure of the a priori statistical error estimator in Eq.~\eqref{eq:apriori_gen_stat}, however replacing the large constant $K(n,m)$, which is responsible for the overly conservative error bound of $\hat{\err}^{(m)}_s$ exemplified in Section~\ref{sec:stat_err_comparison}, with a new one so that the total error matches the computable a posteriori estimator $\hat{\err}_{s,new}^{(m)}$.
In particular, we introduce the redefined level-wise variances $\tilde{V}_l$, computed such that 
\begin{align}
(\hat{\err}^{(m)}_{s,new})^2=\sum_{l=0}^L \frac{\tilde{V}_l}{N_l},
\end{align}
where each $\tilde{V}_l$ is a rescaled version of $\hat{V}_l$ and the same scaling constant is used across levels.
It follows that the rescaled variances $\tilde{V}_l$ decay at the same rate over levels as the a priori variances $\hat{V}_l$.
Specifically, we define the rescaled variances $\tilde{V}_l$ as follows:
\begin{align}
\tilde{V}_l = r_e \hat{V}_l,\quad \text{ where } r_e \coloneqq \frac{(\errest^{(m)}_{s,new})^2}{\sum_{k=0}^L \hat{V}_k/N_k}.\label{eq:vl_rescaling_single}
\end{align}
We refer to $r_e$ as the rescaling ratio.
The formulation of the cost optimisation problem then proceeds similarly to Section~\ref{sec:apriori} with $\hat{V}_l$ replaced by $\tilde{V}_l$.
As before, we neglect the evaluation and interpolation costs $\costfunc$ and $\costint$ since they are negligible in comparison to $\costsim$ for the type of applications considered in this work.
We require that the statistical error satisfies the prescribed tolerance $\epsilon_{s}^2$ while minimising the total cost of the simulation. 
Similar to the constrained optimisation problem for the a priori estimators described in Section~\ref{sec:apriori}, the approach here yields the following optimal level-wise sample sizes:
\begin{align}
N_l^* = \left\lceil \frac{1}{\epsilon_{s}^2} \sqrt{\frac{\tilde{V}_l}{\costsim}} \sum_{k=0}^{L^*} \sqrt{\tilde{V}_k \mathfrak{c}_k}\right\rceil = \left\lceil \frac{r_e}{\epsilon_{s}^2}\sqrt{\frac{\hat{V}_l}{\costsim}} \sum_{k=0}^{L^*} \sqrt{\hat{V}_k \mathfrak{c}_k}\right\rceil. \label{eq:opt_samples} 
\end{align}
We note here that the hierarchy $\{N_l\}_{l=0}^L$ is used to compute the estimator $\hat{V}_l$, $\costsim$ and $\errest^{(m)}_{s,new}$, whereas the hierarchy $\{N_l^*\}_{l=0}^L$ is the cost-optimal hierarchy computed based on these estimators that will achieve a tolerance of $\epsilon^2_s$ on the statistical error.
Finally, we also note that the variance rescaling proposed in Eq.~\eqref{eq:vl_rescaling_single} can be extended to a linear combination of errors as follows:
\begin{align}
\tilde{V}_l = r_e \hat{V}_l, \quad \text{ where } r_e \coloneqq \frac{\left[\sum_{m=0}^2 k_m(\errest^{(m)}_{s,new})^2 \right]}{\sum_{k=0}^L \hat{V}_k/N_k}, \quad 0\le l\le L^*.\label{eq:rescaling_all}
\end{align}

We now justify the choice of rescaling factor with the following theoretical result.
We first establish upper and lower bounds on the true squared statistical error $\sum_{m=0}^2 k_m (\err^{(m)}_s)^2$ in terms of the true level-wise variances $V_l$ defined in Eq.~\eqref{eq:bl_vl}.

\begin{lemma}\label{lemma:nested_bounds}
Let $\Phie$ be the estimator defined in Eq.~\eqref{eq:mlmc_est_1} to approximate $\Phi \in C^4(\Theta)$ and $\{V_l\}_{l=0}^L$ be the true corresponding level-wise variances as defined in Eq.~\eqref{eq:bl_vl}.
Then there exist positive constants $\lambda(n)$ and $K_{low}(n,m)$ such that:
\begin{align}
\frac{\lambda(n)}{|\Theta|} \sum_{l=0}^L \frac{V_l}{N_l} \leq \sum_{m=0}^2 k_m(\err^{(m)}_{s})^2 \leq \left( \sum_{m=0}^2 k_m K(n,m) \right) \sum_{l=0}^L \frac{V_l}{N_l}.
\end{align}
\end{lemma}
\begin{proof}
The upper bound follows directly from the a priori statistical error bound introduced in Eq.~\eqref{eq:apriori_bound_s} in Section~\ref{sec:apriori}.
The lower bound is derived as follows.
We first define the function $\Gamma \coloneqq \Phie - \expec{\Phie} : \Theta \to \setR$.
We then have:
\begin{align}
(\err^{(m)}_{s})^2 = \expec{\norm{\interpm{\Gamma(\thetab)}}^2_{\linf(\Theta)}} \geq \frac{1}{|\Theta|}\expec{\norm{\interpm{\Gamma(\thetab)}}^2_{\ltwo(\Theta)}},
\end{align}
since $\Theta$ is bounded.
The cubic spline interpolant $\interp{\Gamma(\thetab)}$ over a set of point evaluations $\Gamma(\thetab) \in \setR^n,\;\thetab = \{\theta_1, ... , \theta_n\}$, can be written as a linear combination of suitable basis functions $\psi_i(\theta),\;i\in\{1,...,n\}$:
\begin{align}
\interp{\Gamma(\thetab)} = \sum_{i=1}^n \Gamma(\theta_i) \psi_i(\theta).
\end{align}
This then implies that 
\begin{align}
\norm{\interpm{\Gamma(\thetab)}}^2_{\ltwo(\Theta)} &= \sum_{i,j=1}^n \Gamma(\theta_i) \Gamma(\theta_j) \int_{\Theta} \psi^{(m)}_i(\theta)\psi^{(m)}_j(\theta) d \theta\\
&= \Gamma(\thetab)^T B^{(m)} \Gamma(\thetab),
\end{align}
where $B^{(m)} \in \setR^{n \times n}$ is a matrix whose entries are given by $B^{(m)}_{ij} = \int_{\Theta} \psi^{(m)}_i(\theta)\psi^{(m)}_j(\theta) d \theta$.
It then follows that:
\begin{align}
\sum_{m=0}^2 k_m \expec{\norm{\interpm{\Gamma(\thetab)}}^2_{\ltwo(\Theta)}} &= \expec{ \Gamma(\thetab)^T \left( \sum_{m=0}^2 k_m B^{(m)}\right) \Gamma(\thetab) } \geq \lambda \expec{ \norm{ \Gamma(\thetab) }^2_{\sltwo}},
\end{align}
where $\lambda = \lambda(n) > 0$ denotes the minimum eigenvalue of the positive definite matrix $B = \sum_{m=0}^2 k_m B^{(m)}$, which is non-zero since $B$ is non-singular.
We then finally have that 
\begin{align}
\sum_{m=0}^2 k_m (\err^{(m)}_{s})^2 \geq \frac{\lambda(n)}{|\Theta|} \expec{\norm{\Phie(\thetab)-\expec{\Phie(\thetab)}}^2_{\sltwo}} \geq \frac{\lambda(n)}{|\Theta|} \sum_{l=0}^L \frac{V_l}{N_l},
\end{align}
where in the final inequality, we have used the level-wise independence of samples of the \gls{mlmc} estimator.
This concludes the proof.
\end{proof}
Lemma~\ref{lemma:nested_bounds} shows that the true global squared statistical error $\sum_{m=0}^2 k_m (\err^{(m)}_{s})^2$ can be both lower and upper bounded by a constant times the quantity $\sum_{l=0}^L V_l/N_l$.
Therefore, we expect the rescaling ratio
\begin{align*}
r_e = \frac{\sum_{m=0}^2 k_m (\errest_{s,new}^{(m)})^2}{\sum_{l=0}^L \frac{\hat{V}_l}{N_l}} \approx \frac{\sum_{m=0}^2 k_m (\err_{s}^{(m)})^2}{\sum_{l=0}^L \frac{V_l}{N_l}}
\end{align*} 
to remain bounded independent of the hierarchy $\{N_l\}_{l=0}^L$.

Lastly, since the variances $\hat{V}_l$ are estimated using Monte Carlo sampling, the estimates on finer levels typically have a larger error due to smaller sample-sizes.
In addition, estimates of $\hat{V}_l$ and $\costsim$ may not be available for unexplored levels $l$.
To alleviate this problem, we fit the exponential models $c_{\beta} e^{-\beta l}$ and $c_{\gamma} e^{\gamma l}$ on the variances $\tilde{V}_l$ and costs $\costsim$ respectively for $l\in\{1,...,L\}$ using a least-squares fit, similar to the procedure described in Section~\ref{sec:practical}.
We use the costs and variances predicted by these models instead of $\tilde{V}_l$ and $\costsim$ in Eq.~\eqref{eq:opt_samples} for the optimal level-wise sample sizes.
This stabilises the estimates computed on finer levels, and the models can also be extrapolated for levels where estimates are not available yet.
The expression for the optimal level-wise sample sizes is then given by
\begin{align}
N_l^* = \left\lceil \frac{c_{\beta}}{\epsilon_{s}^2} \sqrt{\frac{e^{-\beta l}}{e^{\gamma l}} \sum_{k=0}^{L^*} \sqrt{e^{(\gamma-\beta) k} }}\right\rceil, \quad 0\le l\le L^*. \label{eq:nl_adapt_model}
\end{align}

\subsection{Assessment of the stability and behaviour of the rescaling ratio} \label{sec:rescaling_ratio}
We now wish to numerically study the behaviour of the rescaling ratio $r_e$.
We therefore consider once again the Poisson problem from Section~\ref{sec:est_tests} and focus instead on the computation of the $70\%$-\gls{cvar}.
It will be shown in Section~\ref{sec:var_cvar_error_bound} that the \gls{mse} of the \gls{pdf}, the \gls{cdf}, the \gls{var} and the \gls{cvar} can all be written in the form of Eq.~\eqref{eq:three_sums_mse} with appropriately chosen values of $k_0$, $k_1$ and $k_2$.
Particularly, Lemma~\ref{lemma:accuracy:derived} in that section derives the values of $k_0$, $k_1$ and $k_2$ for the \gls{var} and the \gls{cvar}.
Fig.~\ref{fig:ratio_study_poisson} shows the variation of $r_e$ for different hierarchy shapes $N_l = N_0 2^{rl},\;l\in\{0,...,L\},\; r \in \{-1,0,1\}$ with $L=5$ and for different interval sizes $\Theta$ centred approximately around the $70\%$-\gls{var}.
For each value of $r$, $\Theta$ and $N_0$, we simulate 20 independent random realizations of the hierarchy and plot the values of $r_e$ along with the sample average over the 20 values of $r_e$.
We observe that for nearly all choices of hierarchy, the rescaling ratio $r_e$ is stable, in the sense that the realizations are clustered about a mean value with a relatively small variance.
However, for hierarchies with very small sample sizes $N_0$ on coarser levels, which contribute proportionately more to the overall statistical error, and for cases with smaller intervals, we observe sporadic large values of $r_e$.
These observations indicate that one needs to select an adequately large sample size and/or interval size in order for the rescaling ratio $r_e$ to be numerically stable. 
It can be seen from Eq.~\eqref{eq:nl_adapt_model} that larger values of $r_e$ in practice lead to larger hierarchies and, hence, more conservative statistical error estimates.
It is therefore important to select the interval $\Theta$ and sample sizes appropriately.

\begin{figure}[ht]
  \begin{subfigure}{0.32\textwidth}
    \includegraphics[width=\textwidth]{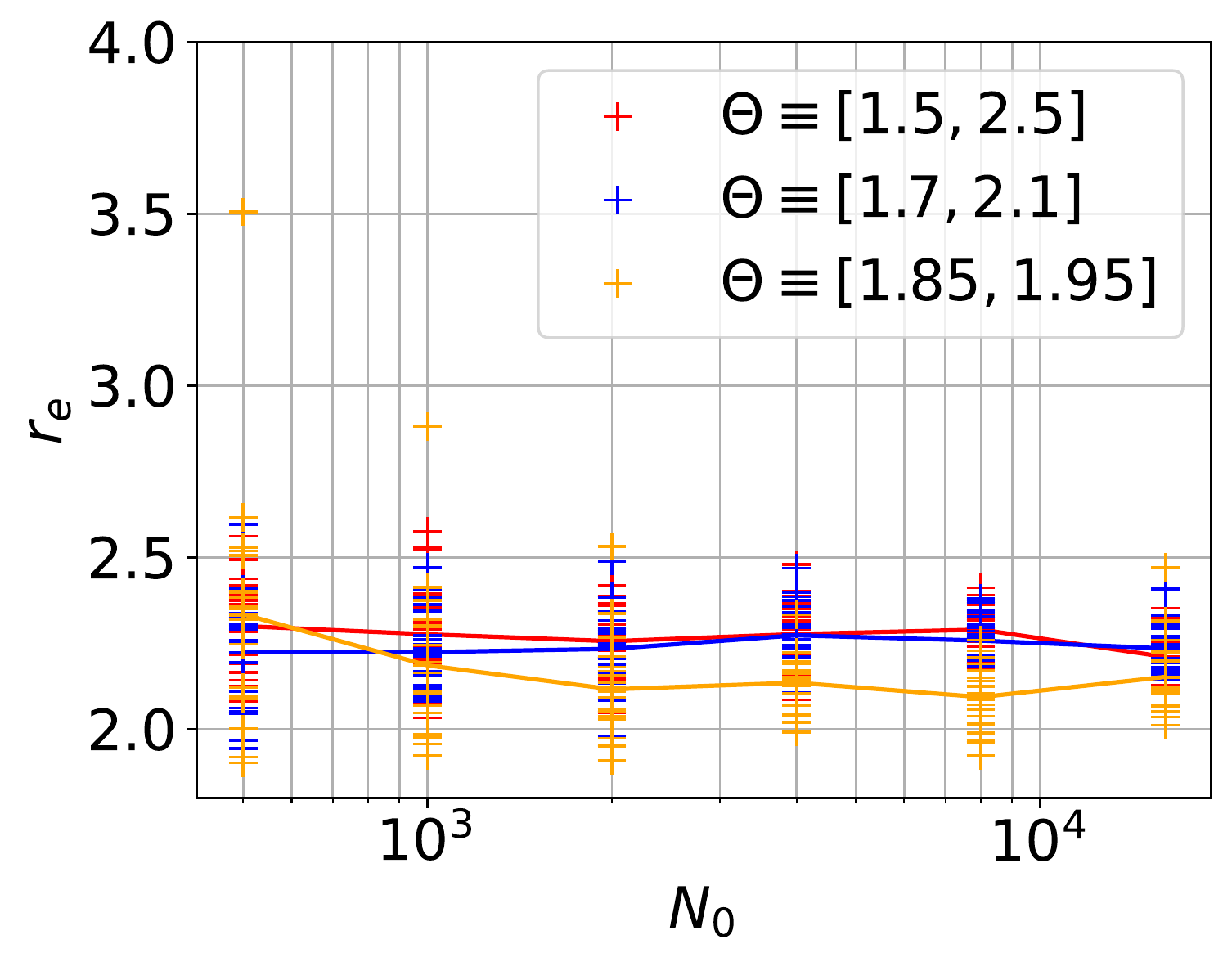}
    \caption{$r=-1$}
  \end{subfigure}
    \begin{subfigure}{0.32\textwidth}
    \includegraphics[width=\textwidth]{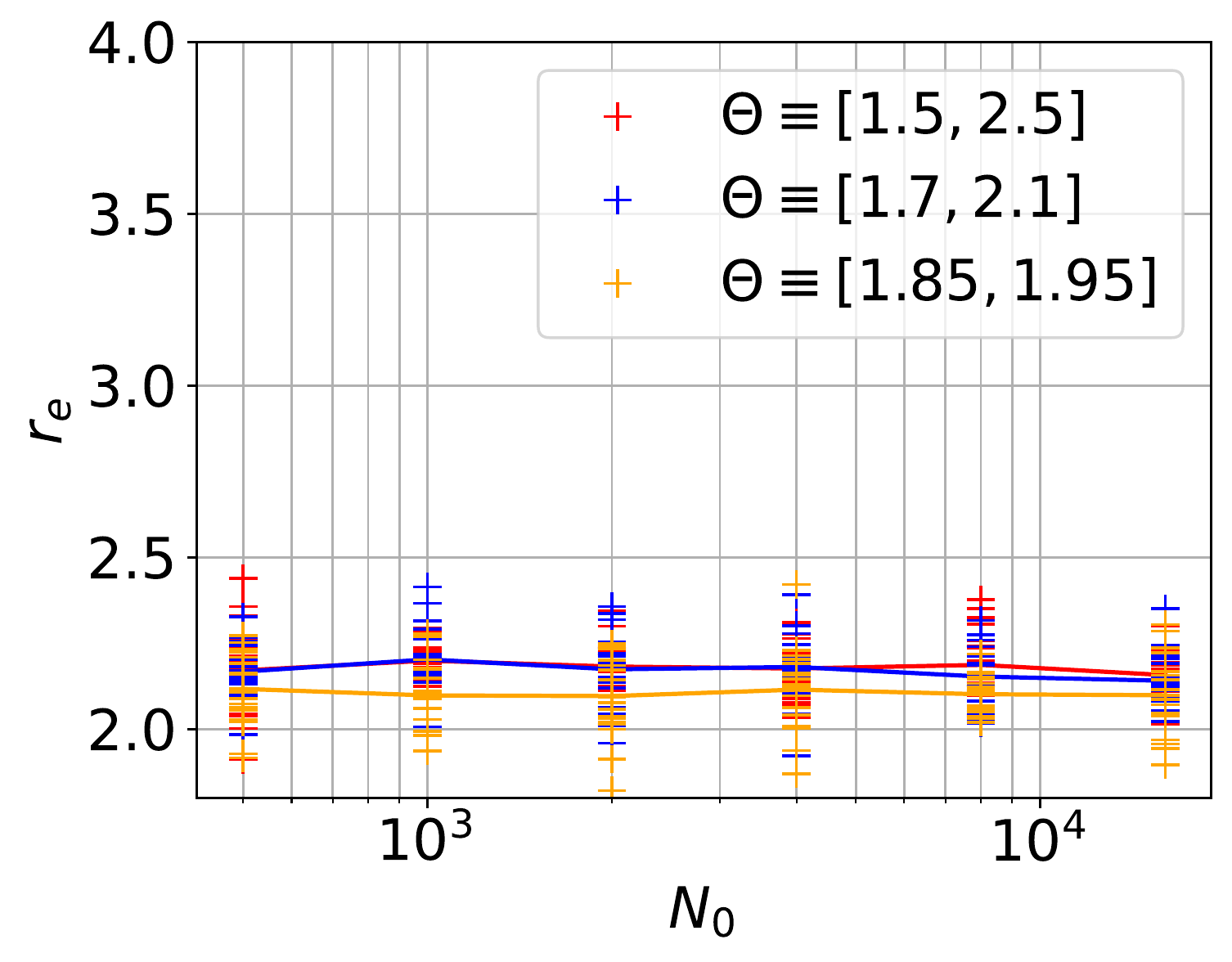}
    \caption{$r=0$}
  \end{subfigure}
    \begin{subfigure}{0.32\textwidth}
    \includegraphics[width=\textwidth]{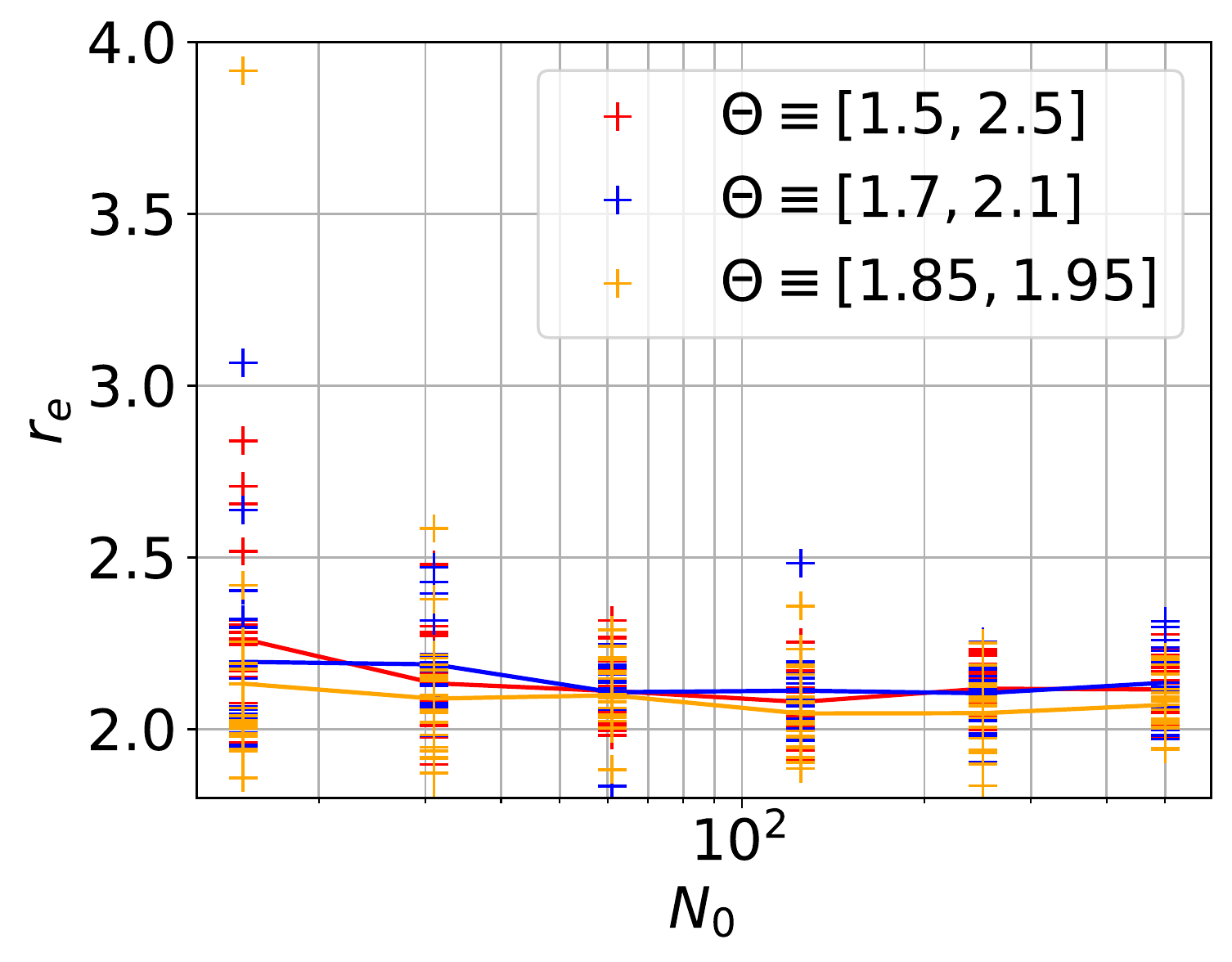}
    \caption{$r=1$}
  \end{subfigure}
  \caption{Behaviour of $r_e$ for different hierarchy shapes and interval sizes for Poisson problem}
  \label{fig:ratio_study_poisson}
\end{figure}
\FloatBarrier

\subsection{Adaptive \gls{mlmc} algorithm} \label{sec:mlmc_alg}
It was shown earlier that the cost optimal number of interpolation points $n$, level-wise sample sizes $N_l$ and the number of levels $L$ could be calculated according to Eqs.~\eqref{eq:interp_points_all}, \eqref{eq:all_biases_adapt} and \eqref{eq:nl_adapt_model} respectively, with knowledge of the quantities $\hat{b}_{l,new}^{(m)}$, $\hat{V}_l$ and $\costsim$, their corresponding decay rates and constants, as well as $\errest_{s,new}^{(m)}$.
Since estimates of these quantities are computed using a posteriori computations and require that samples have already been computed, we propose the use of a variation of the \gls{cmlmc} algorithm introduced in \cite{collier2014continuation} and further adapted for complex simulation problems in \cite{pisaroni2016continuation}.
The \gls{cmlmc} algorithm begins with a small pre-set initial hierarchy, typically called the ``screening'' hierarchy, and a geometrically decreasing sequence of tolerances $\epsilon_0^2 > \epsilon_1^2 > \dots > \epsilon_d^2 = \epsilon^2$, where $\epsilon^2$ is the target tolerance to be achieved on the \gls{mse}.
The method then adapts for the tolerance $\epsilon^2_j, \; j \in \{1,...,d\}$ based on the estimates from the hierarchy tuned on $\epsilon^2_{j-1}$.
For $\epsilon_0$, one uses the estimates from the screening hierarchy.
The advantage of this method is that the estimators $b^{(m)}_l$, $\hat{V}_l$,  $\errest_{s,new}^{(m)}$ and $\costsim$ are successively improved.
This makes the algorithm more robust to inaccurate initial estimates from the screening hierarchy, since the screening hierarchy is typically selected to be much smaller than the optimal hierarchy.
The algorithmic description of the \gls{cmlmc} algorithm is presented in Alg.~\ref{alg:xmc} for a general statistic $s_{\tau}$ whose \gls{mse} decomposes as in Eq.~\eqref{eq:general_mse}.
The reader is referred to \cite{collier2014continuation} for a more detailed exposition. 

\begin{algorithm}[ht]
\begin{algorithmic}
	\STATE Input: Target tolerance $\epsilon > 0$, Number of \gls{cmlmc} iterations $d \in \setN$, Tolerance refinement ratios $\lambda > \kappa > 1$, Error parameters $k_0$, $k_1$ and $ k_2$. Set $j=1$, $\epsilon_a = \epsilon_0$.
	\STATE Launch screening hierarchy.
	\STATE Compute estimators $\hat{b}_{l,new}^{(m)}, \hat{V}_l, \costsim, \errest_{s,new}^{(m)}$ and model parameters $c_{\alpha}, c_{\beta},c_{\gamma}, \alpha, \beta, \gamma$.
	\STATE Compute $\mse{\statest}$ based on Eq.~\eqref{eq:general_mse}.
	\WHILE{$j\leq d$ \textbf{or} $\mse{\statest}\geq \epsilon^2$}
		\STATE Launch hierarchy with $n^*(\epsilon_a)$, $L^*(\epsilon_a)$, $\{N_l^*(\epsilon_a)\}_{l=0}^{L^*}$ computed based on Eqs.~\eqref{eq:interp_points_all},~\eqref{eq:all_biases_adapt} and~\eqref{eq:nl_adapt_model}
		\STATE \textbf{if} $j\leq d$ $\{$ Set $\epsilon_a = \epsilon \lambda^{(d-j)}$ $\}$ \textbf{else} $\{$ Set $\epsilon_a = \epsilon \kappa^{(d-j)}$ $\}$
		\STATE Compute estimators $\hat{b}_{l,new}^{(m)}, \hat{V}_l, \costsim, \errest_{s,new}^{(m)}$ and model parameters $c_{\alpha}, c_{\beta},c_{\gamma}, \alpha, \beta, \gamma$
		\STATE Compute $\mse{\statest}$ based on Eq.~\eqref{eq:general_mse}
		\STATE Update $j \leftarrow j + 1$
	\ENDWHILE

\end{algorithmic}
\caption{\gls{cmlmc} Algorithm}
\label{alg:xmc}
\end{algorithm}

\subsection{Error bounds on the \gls{pdf}, the \gls{cdf}, the\gls{var} and the \gls{cvar}} \label{sec:var_cvar_error_bound}
It follows directly from Eq.~\eqref{eq:three_sums_mse} that the \gls{mse} of the \gls{pdf} $f_{\qoi}(\theta)$ and the \gls{cdf} $F_{\qoi}(\theta)$ can be written as follows:
\begin{align}
\mse{\hat{F}_{\qoi}} = (1-\tau)^2 \mse{\Phieone}, \qquad \mse{\hat{f}_{\qoi}} = (1-\tau)^2 \mse{\Phietwo},
\end{align}
where $\hat{F}_{\qoi}(\theta) \coloneqq \tau + (1-\tau) \Phieone$ and $\hat{f}_{\qoi}(\theta) \coloneqq (1-\tau) \Phietwo$.
As a result, Eq.~\eqref{eq:three_sums_mse} can be used to bound the error on these quantities by selecting $k_1$ and $k_2$ appropriately.
We now present a simple result to demonstrate that the general form in Eq.~\eqref{eq:three_sums_mse} can also be used also to bound the \gls{mse} of the \gls{var} and the \gls{cvar}.
\begin{lemma}\label{lemma:accuracy:derived}
Let $\Phie$ be the multilevel Monte Carlo estimator defined in Eq.~\eqref{eq:mlmc_est_1} to approximate $\Phi$. 
If there exist $\quantest,\quant \in\Theta$ such that $\Phieone(\quantest)=\Phione(\quant)=0$ for some given $\tau\in\setR$, then it holds that
\begin{equation}
\expec{ \left|\quantest - \quant\right|^2}\le \norm{ \frac{1}{\Phi^{(2)}}}_{\linf([\quantest, \quant])}^2\mse{\hat{\Phi}_L^{(1)}},\label{eq:err_var}
\end{equation}
as well as that
\begin{equation}
\expec{\left|\cvarest - \cvar\right|^2}\le 
2 \norm{\Phi^{(1)}}_{\linf([\quantest, \quant])}^2 \norm{ \frac{1}{\Phi^{(2)}}}_{\linf([\quantest, \quant])}^2\mse{\hat{\Phi}_L^{(1)}} + 2\mse{\Phie},\label{eq:err_cvar}
\end{equation}
where $\cvarest = \Phie(\quantest)$ and $\cvar= \Phi(\quant)$.
\end{lemma}
\begin{proof}
Let $\quantest,\quant\in\Theta$ be such that $\hat{\Phi}^{(1)}(\quantest) =  \Phi^{(1)}(\quant) = 0$.
It then follows from Taylor's theorem that
\begin{equation}
\lvert\quantest - \quant  \rvert \lvert\Phi^{(2)}(\xi) \rvert = \lvert \hat{\Phi}^{(1)}_L(\quantest) - \Phi^{(1)}(\quantest)  \rvert \label{eq:var_bound}
\end{equation}
for some $\xi$ between $\quantest$ and $\quant$, so that the first claim follows. 
The second claim follows from the first claim upon noting that
\begin{align}
\lvert \cvarest-\cvar \rvert = \bigl\vert \hat{\Phi}_L(\quantest) - \Phi(\quant) \bigr\vert &\le \bigl\vert \Phi(\quantest) - \Phi(\quant) \bigr\vert + \bigl\vert \hat{\Phi}_L(\quantest) - \Phi(\quantest) \bigr\vert \nonumber \\
&\le {\left\vert \Phi^{(1)}(\zeta)\right\vert} \lvert\quantest - \quant  \rvert +  \bigl\vert \hat{\Phi}_L(\quantest) - \Phi(\quantest) \bigr\vert
\end{align}
in view of Taylor's theorem for some $\zeta$ between $\quantest$ and $\quant$.
\end{proof}
From Lemma~\ref{lemma:accuracy:derived}, it is evident that $\quant$ and $\quantest$ represent the true and estimated \gls{var}, while $\cvar$ and $\cvarest$ represent the true and estimated \gls{cvar}, respectively. 
We can then derive \gls{mse} bounds for the \gls{var} and the \gls{cvar} by setting the constants $k_0$, $k_1$ and $k_2$ in Eq.~\eqref{eq:general_mse} based on Lemma~\ref{lemma:accuracy:derived}.
A closed form expression for the solution of Eq.~\eqref{eq:interp_points_all} for the number of interpolation points can be derived since Lemma~\ref{lemma:accuracy:derived} implies that some of the constants $k_0, k_1$ and $k_2$ are zero for each of the \gls{var} and the \gls{cvar}.
For the number of levels $L$ and level-wise sample sizes $N_l$, the methods described earlier in this section can be directly used with the appropriate values of the constants $k_0, k_1$ and $k_2$.
Since we expect the interval $[\quant, \quantest]$ to be small, we replace each of the constants $\norm{\Phi^{(1)}}_{\linf([\quant, \quantest])}$ with $|\Phie^{(1)}(\quantest)|$ and $\norm{1/\Phi^{(2)}}_{\linf([\quant, \quantest])}$ with $|1/\Phie^{(2)}(\quantest)|$ in practice.
Lastly, we note that although $k_0$, $k_1$ and $k_2$ in Eq.~\eqref{eq:general_mse} and in Algorithm~\ref{alg:xmc} are constants, we use the function estimator $\Phiem$ to estimate and update them iteratively within the continuation framework in Algorithm~\ref{alg:xmc}.

\section{Numerical Experiments}\label{sec:results}
We now demonstrate the performance of the above combination of novel error estimators, adaptive strategy and \gls{cmlmc} algorithm on a set of test cases. 
Firstly, we consider again the simple Poisson problem introduced in Section~\ref{sec:est_tests}.
We then study a problem of options contract pricing using the Black-Scholes \gls{sde}. 
Lastly, we study a case of laminar steady fluid flow over a cylinder placed in a channel governed by the Navier-Stokes equations, which demonstrates the methodology on a more applied problem.

\subsection{Poisson Problem}\label{sec:poisson}
We consider the same random Poisson equation in two spatial dimensions described in Section~\ref{sec:est_tests}.
We recall that we have an explicit dependence of $\qoi$ on the random input $\xi$ for this example and hence, it is straightforward to compute reference values for the \gls{var} and \gls{cvar} of different significances (See Tab.~\ref{tab:ellip2d:quantiles}).

The details of the input uncertainties, numerical scheme and discretisation hierarchy are described in Section~\ref{sec:est_tests}.
We are interested particularly in the estimation of the \gls{cvar} with significance $\tau = 0.7$, and hence consider the interval $\Theta = [1.5,2.5]$ as before.
The \gls{mlmc} estimator proposed in Section~\ref{sec:mlmc_par_est} is used to estimate the parametric expectation $\Phi$.
The \gls{cvar} estimate is computed from $\Phie$ as described in Eq.~\eqref{eq:stat_defs}.
The hierarchy is adaptively calibrated as described in Section~\ref{sec:adaptivity} based on the novel error estimators described in Section~\ref{sec:novel}.
The \gls{cmlmc} algorithm described in Section~\ref{sec:mlmc_alg} is then used to successively improve the estimates required to compute the optimal hierarchy with $\lambda = 1.5$ and $\kappa = 1.1$ in Algorithm~\ref{alg:xmc}. 
To compute the statistical error estimate, we initially use $N_{bs}=100$ bootstrapped samples and then adapt $N_{bs}$ according to the procedure described in Section~\ref{sec:stat_err_est} to obtain a bootstrap error smaller than $1\%$ of the squared statistical error tolerance.
The above combination of problem simulations, \gls{mlmc} and error estimation have been implemented in the Python package XMC \cite{ExaQUte_XMC}.

To assess the robustness of the novel error estimators, a reliability study is conducted.
For a given tolerance, the entire \gls{mlmc} simulation is repeated 20 times independently.
For each simulation, an estimate of the \gls{cvar} and a corresponding estimate of the \gls{mse} are produced.
Since the true value of the \gls{cvar} is known for this example, we compute the corresponding true squared errors.
We expect the \gls{mse} estimates to be approximately equal to the sample average of the true squared errors, which we take here as the reference value for the true \gls{mse}.
The results of this reliability study are shown in Fig.~\ref{fig:poisson_reliability}.
As can be seen from the figure, the error estimates bound the true error on the \gls{cvar} and lead to practically computable \gls{mlmc} hierarchies for the \gls{cvar}.
For all the tolerances tested, the squared error estimate is not larger than 10 times the squared true error, which we consider acceptable for practical applications (cf. Section~\ref{sec:est_tests}).

To verify the predictions of Proposition~\ref{thm:complexity:MLMC:spline}, we also compute the cost of each \gls{mlmc} simulation according to Eq.~\eqref{eq:simple_cost_model}.
The time taken to compute each of the $N_l$ samples is measured, and $\costsim$ is taken to be their average.
The cost is computed using the level-wise sample sizes corresponding to the final iteration of the \gls{cmlmc} that satisfies the target tolerance and averaged over the 20 repetitions of the algorithm.
The results are summarised in Fig.~\ref{fig:poisson_complexity}, where the average cost over all the simulations for each final \gls{cmlmc} tolerance is plotted versus the final tolerance.

To compare the \gls{mlmc} estimator with the Monte Carlo method, we propose the following Monte Carlo estimator:
\begin{align}
\hat{\Phi}^{(m)}_{L,mc}(\theta) &\coloneqq \interpm{\hat{\Phi}_{L,mc}(\thetab)},\quad \theta \in \Theta, \nonumber\\
\text{where } \hat{\Phi}_{L,mc}(\theta_j) &\coloneqq \frac{1}{N} \sum_{i=1}^N \phi(\theta_j, \qoi_L^{(i)}). \label{eq:phi_mc}
\end{align}
To estimate the \gls{mse} of the \gls{cvar} computed from the estimator in Eq.~\eqref{eq:phi_mc}, we utilise the general \gls{mse} form in Eq.~\eqref{eq:general_mse} but for an \gls{mlmc} estimator with a single level, i.e., without the telescoping summation term in Eq.~\eqref{eq:mlmc_est_1}.
The constants $k_0$, $k_1$ and $k_2$ are chosen according to the results of Lemma~\ref{lemma:accuracy:derived} for the \gls{cvar}.
We now describe a procedure to select the parameters of this estimator such that a prescribed tolerance can be obtained on the corresponding \gls{mse} of the \gls{cvar}.
The number of interpolation points $n$ used in the Monte Carlo estimator is selected to be the same as for the \gls{mlmc} estimator, since the \gls{cmlmc} method leaves the number of interpolation points unchanged for all tested tolerances.
The discretisation level $L$ is selected to be the same as the finest level predicted by the \gls{cmlmc} algorithm for the \gls{mlmc} estimator over all repetitions of the \gls{cmlmc} algorithm for the given tolerance, although nearly all repetitions predict the same level $L$ for a given tolerance.
To predict the correct sample size $N$, we first note that the squared statistical error term of the \gls{cvar} for the Monte Carlo estimator in Eq.~\eqref{eq:phi_mc} contains only the single level contribution and hence, is inversely proportional to the sample size $N$, where the numerator is independent of $N$ and can be estimated using a sample variance estimator.
The sample size is then selected such that this statistical error term satisfies the same squared statistical error tolerance $\epsilon^2_s$ as the \gls{mlmc} estimator. 
The cost can then be computed in a straightforward manner from $N$ and the cost of a single simulation at level $L$.
The estimated Monte Carlo cost is shown as well in Fig.~\ref{fig:poisson_complexity}, together with a least squares fit rate over the estimated Monte Carlo cost. 
We observe that the predictions made by Proposition~\ref{thm:complexity:MLMC:spline} are observed here, namely that the \gls{mlmc} cost grows as $\mathcal{O}(\epsilon^{-2})$ and that the Monte Carlo cost grows as $\mathcal{O}(\epsilon^{-3})$ for a prescribed tolerance $\epsilon^2$ on the \gls{mse} of the \gls{cvar}.

\begin{figure}[H]
  \centering
  \begin{subfigure}{0.45\textwidth}
    \centering
    \includegraphics[width=\textwidth]{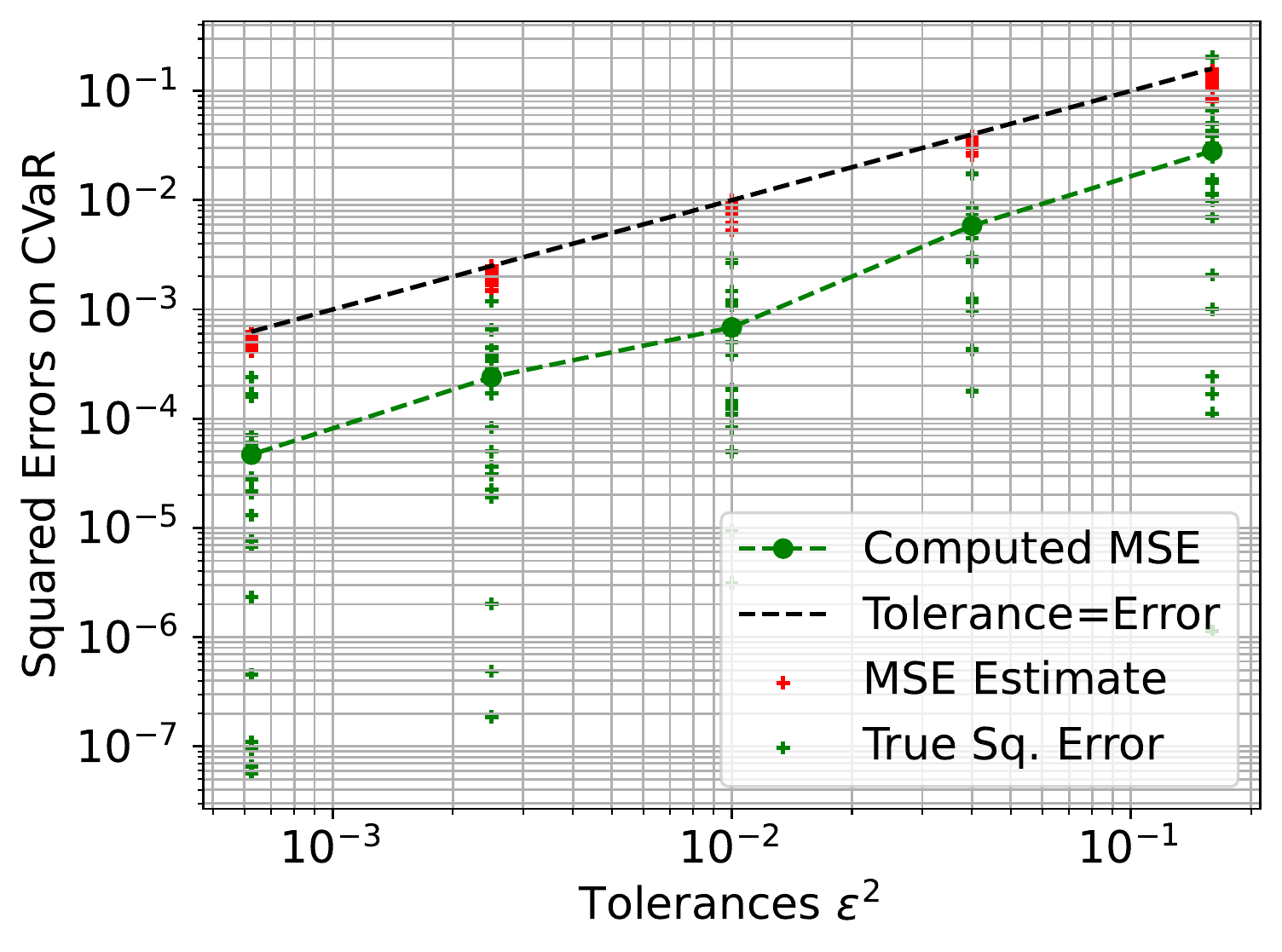}
    \caption{Reliability of error estimator}
    \label{fig:poisson_reliability}
  \end{subfigure}
  \begin{subfigure}{0.45\textwidth}
    \centering
    \includegraphics[width=\textwidth]{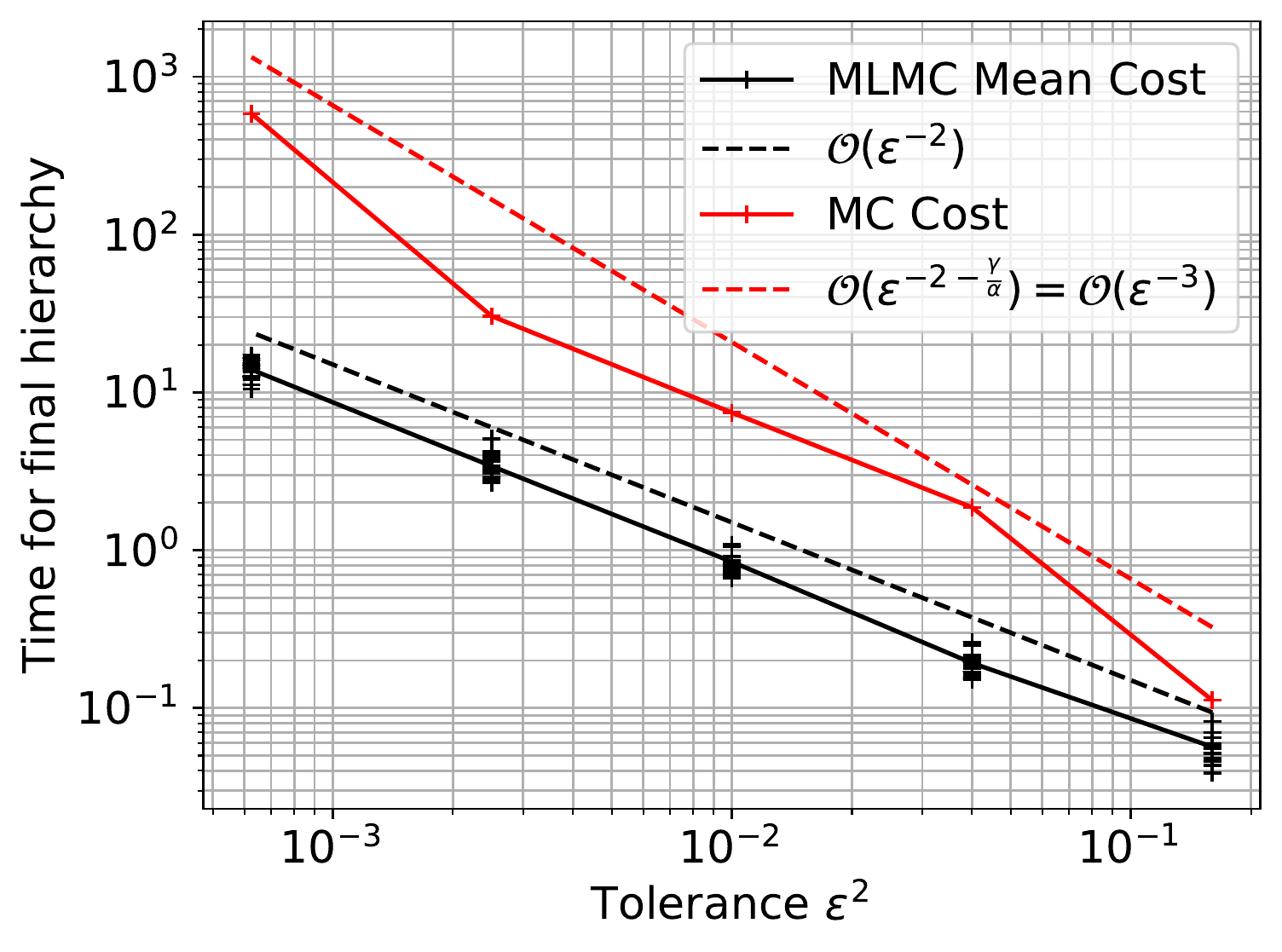}
    \caption{Complexity behaviour}
    \label{fig:poisson_complexity}
  \end{subfigure}
  \caption{Summary of results for the Poisson problem}
  \label{fig:poisson_performance}
\end{figure}

Fig.~\ref{fig:phim_var_error_poisson} compares the true and estimated squared errors on $\Phim, \; m \in \{0,1,2\}$ and the \gls{var} for the same set of simulations as in Fig.~\ref{fig:poisson_performance} plotted against the prescribed tolerance used in the \gls{cvar} calculation.
As can be seen in this figure, a tight bound is obtained on $\Phim, m \in \{0,1,2\}$, with a comparatively more conservative estimate on the \gls{var}.
The reason for this discrepancy can be explained with Lemma~\ref{lemma:accuracy:derived}; although the equality in Eq.~\eqref{eq:var_bound} holds true for the function derivative evaluated at the \gls{var}, Eq.~\eqref{eq:err_var} in turn bounds this with the $\linf$ norm over the entire interval $\Theta$.
Finally, Fig.~\ref{fig:poisson_var} shows the \gls{mse} of the \gls{var} estimated from the same \gls{qoi} realizations corresponding to the optimal hierarchy for the interval $\Theta = [1.5,2.5]$, but using a smaller interval $\Theta = [1.87,1.89]$ such that the $70\%$-\gls{var} is contained within the interval.
It can be seen that choosing a smaller interval around the \gls{var} results in a tighter bound on the true \gls{mse}.
However, this choice needs to be balanced with the numerical stability of the rescaling ratio $r_e$ in Eq.~\eqref{eq:vl_rescaling_single} (cf. discussion in Section~\ref{sec:rescaling_ratio}).
The choice of interval $\Theta$ hence may have an important effect on the tightness of the error bounds on the \gls{var} and \gls{cvar}.
In practical applications, however, one does not know a priori the location of the \gls{var}.
For the purposes of this study, we only explore fixed intervals $\Theta$ and find that the resultant hierarchies are practically computable, leading to effective \gls{mlmc} estimators.
In future works, we plan to explore algorithms that adaptively select $\Theta$.

\begin{figure}[t]
\centering
\begin{subfigure}{.48\textwidth}
    \centering
    \includegraphics[width=\textwidth]{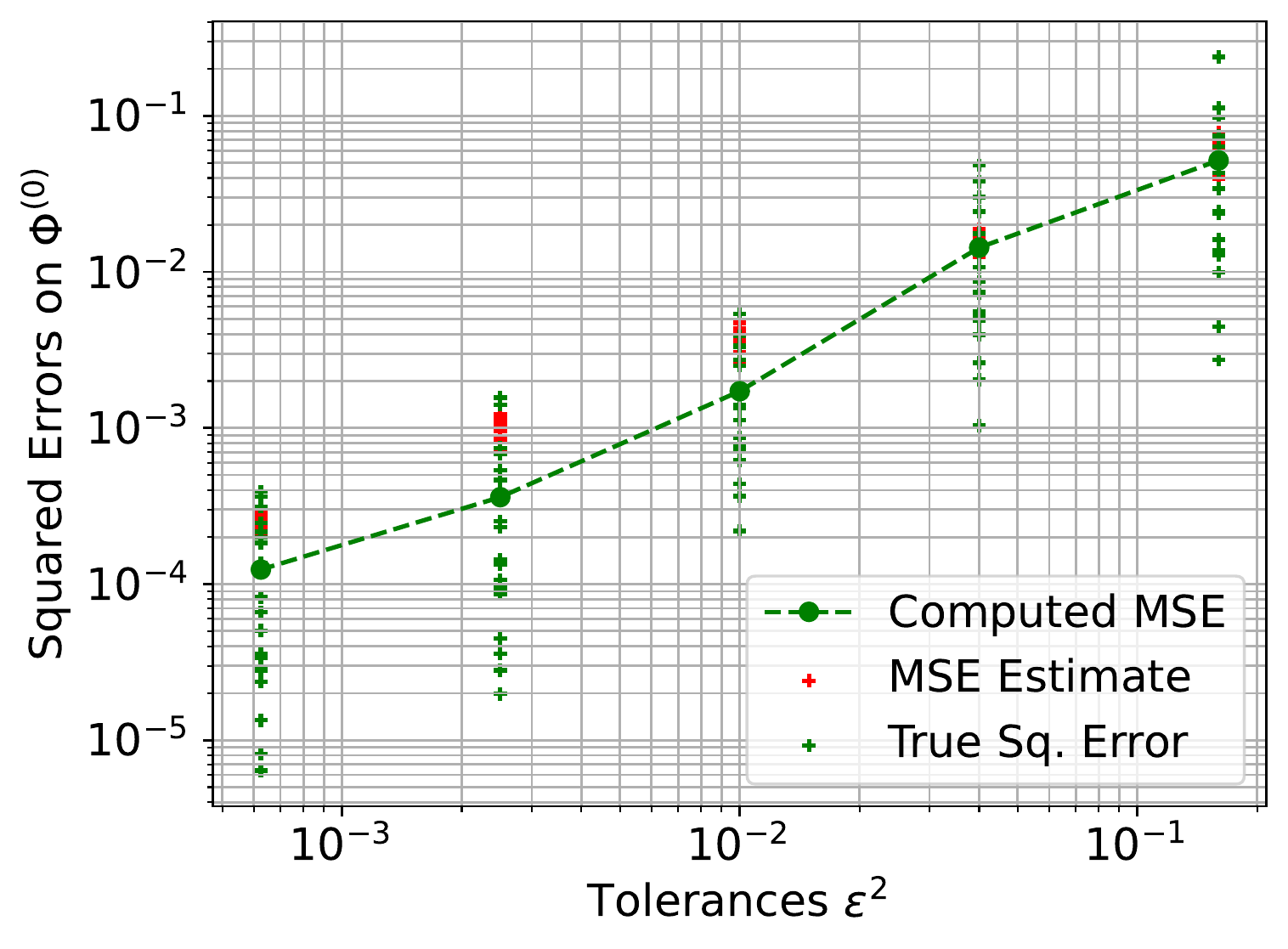}
    \caption{$m=0$}
\end{subfigure}
\begin{subfigure}{.48\textwidth}
    \centering
    \includegraphics[width=\textwidth]{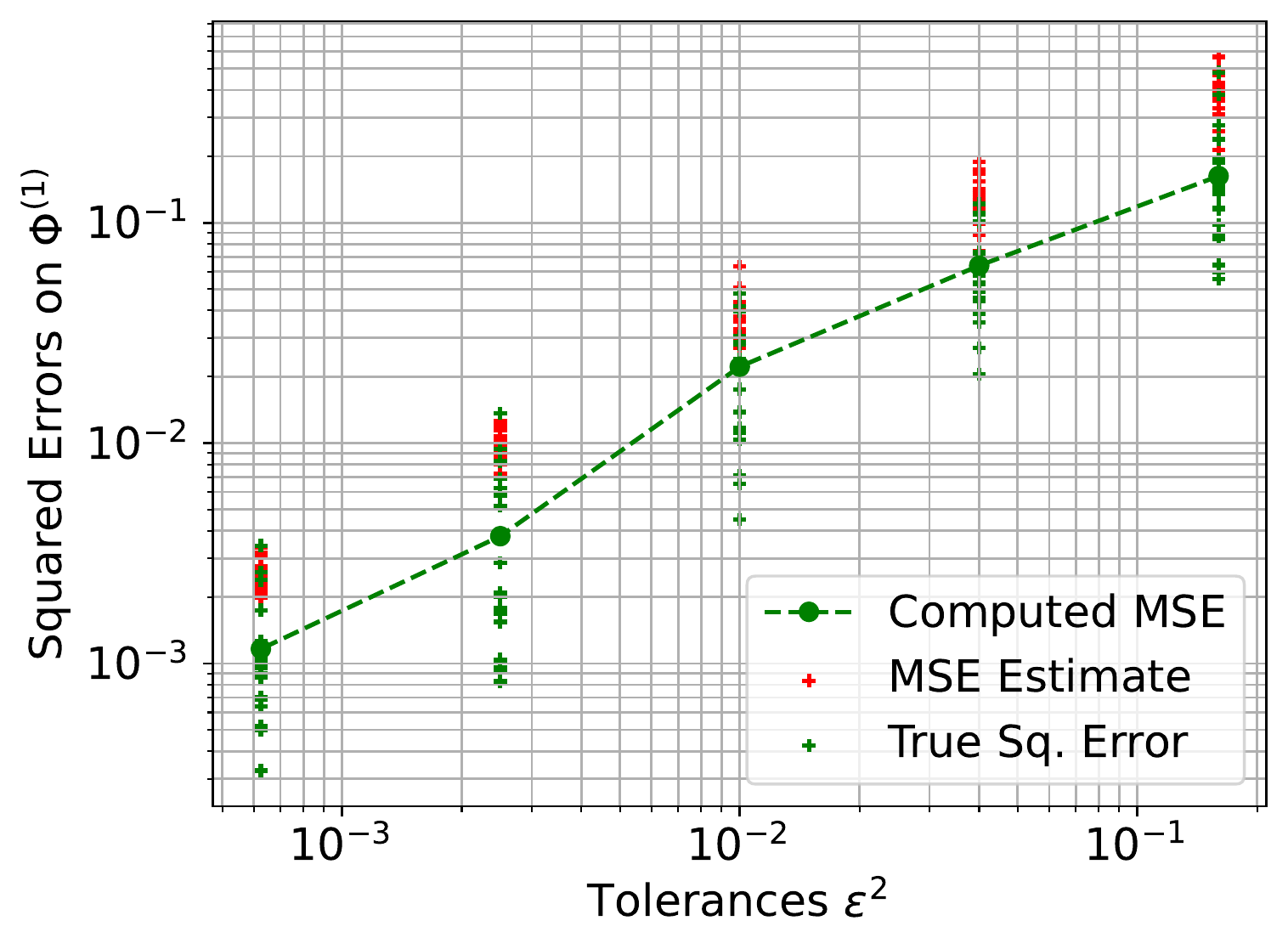}
    \caption{$m=1$}
\end{subfigure}

\begin{subfigure}{.48\textwidth}
    \centering
    \includegraphics[width=\textwidth]{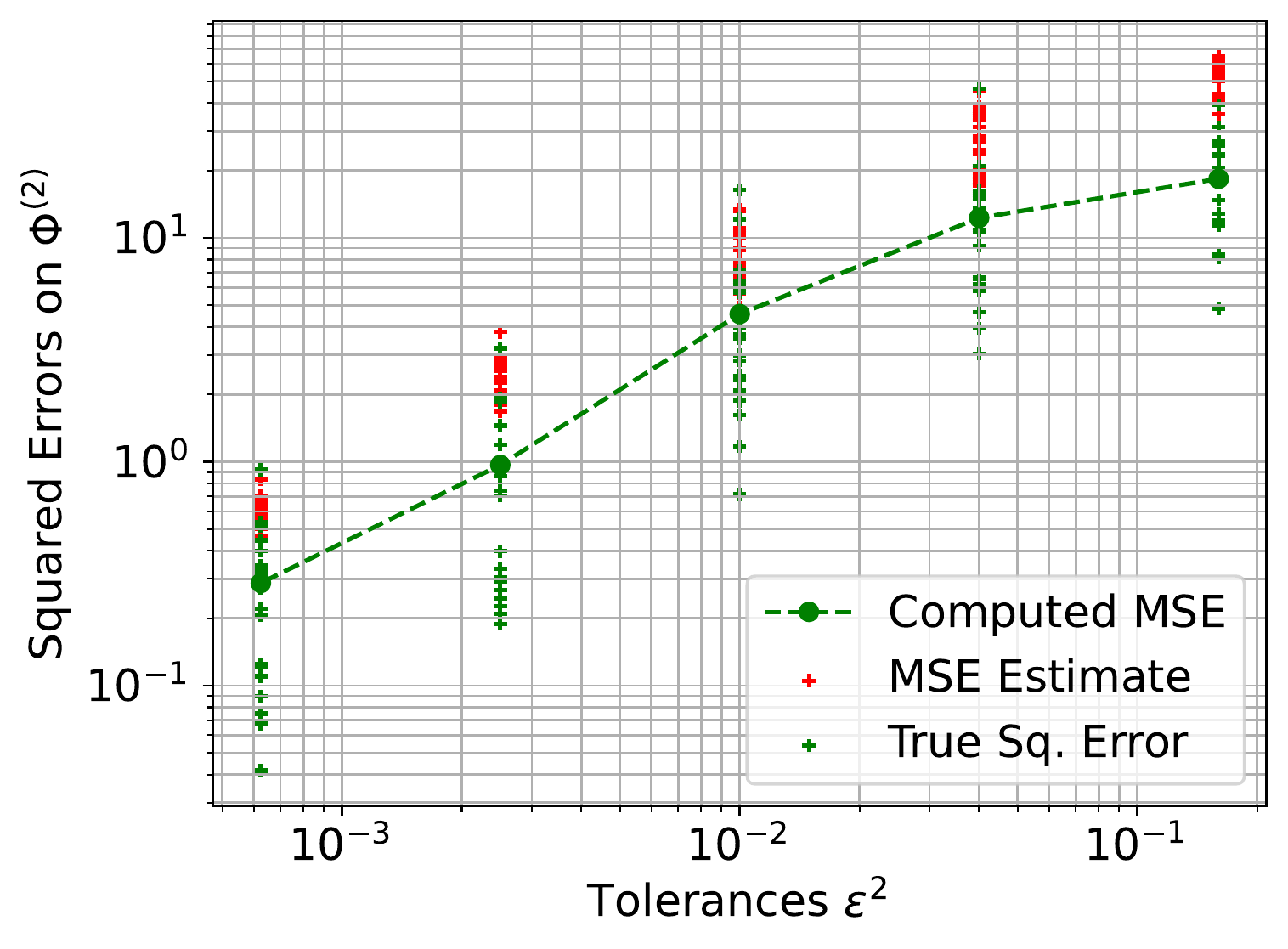}
    \caption{$m=2$}
\end{subfigure}
\begin{subfigure}{.48\textwidth}
    \centering
    \includegraphics[width=\textwidth]{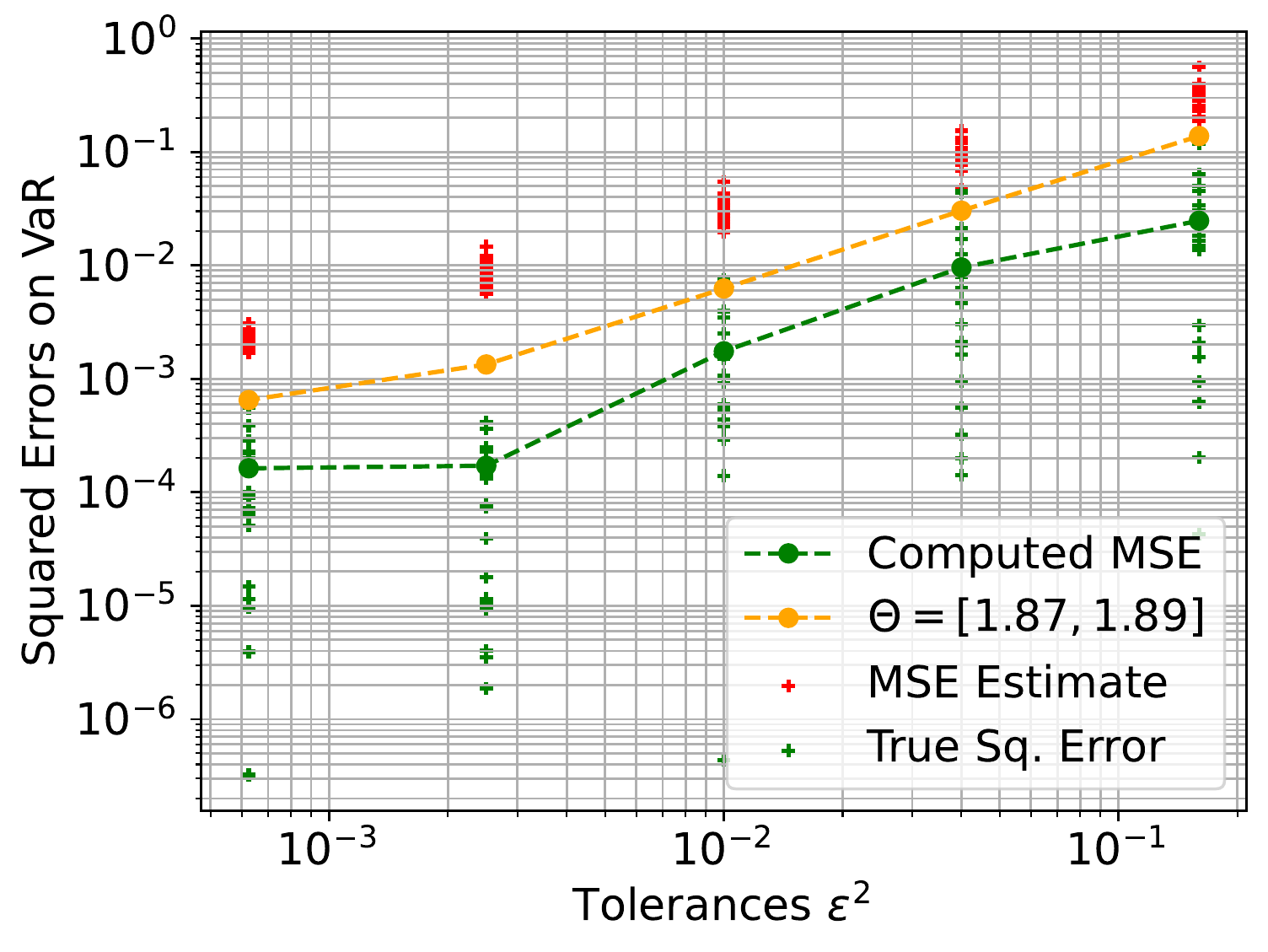}
    \caption{\gls{var}}
    \label{fig:poisson_var}
\end{subfigure}
\caption{Reliability of error estimator for $\Phim$ and \gls{var} for the Poisson problem}
\label{fig:phim_var_error_poisson}
\end{figure}

In Fig.~\ref{fig:poisson_complexity_all_stats}, we conduct a similar complexity study as the one shown in Fig.~\ref{fig:poisson_complexity}.
We adapt the \gls{mlmc} hierarchy to achieve a particular relative tolerance on the \gls{mse} of each of $\Phi^{(m)}, m \in \{0,1,2\}$, as well as the \gls{var} and \gls{cvar}.
The relative error $\err_{r}^{(m)}$ of $\Phiem$ and the relative errors $\err_{r}^{\quant}$ and $\err_{r}^{\cvar}$ of the \gls{var} and the \gls{cvar} respectively, are computed as follows:
\begin{align}
(\err_{r}^{(m)})^2 = \frac{\mse{\Phiem}}{\norm{\Phiem}^2_{\linf(\Theta)}}, \qquad (\err_{r}^{\quantest})^2 = \frac{\mse{\quantest}}{\quantest^2},\qquad (\err_{r}^{\cvarest})^2 &= \frac{\mse{\cvarest}}{\cvarest^2}.
\end{align}
We run 20 independent runs of the \gls{cmlmc} algorithm each for a given statistic and a given relative tolerance.
We plot the average of the cost to compute the optimal hierarchy over these 20 simulations versus the corresponding relative tolerance in Fig.~\ref{fig:poisson_complexity_all_stats}.
As can be seen from Fig.~\ref{fig:poisson_complexity_all_stats}, higher derivatives of $\Phi$ are more expensive to compute for a certain relative tolerance.
In addition, for each simulation that was adapted on $\Phi^{(1)}$ for all tolerances considered, we plot the cost of the simulation versus the \gls{mse} estimate on the \gls{var}.
It can be observed from Fig.~\ref{fig:poisson_complexity_all_stats} that for a given budget, adapting the hierarchy on the \gls{var} directly leads to a much lower relative error than adapting on $\Phi^{(1)}$ and computing the \gls{var} as a postprocessing step.

\begin{figure}[H]
  \centering
  \includegraphics[width=0.65\textwidth]{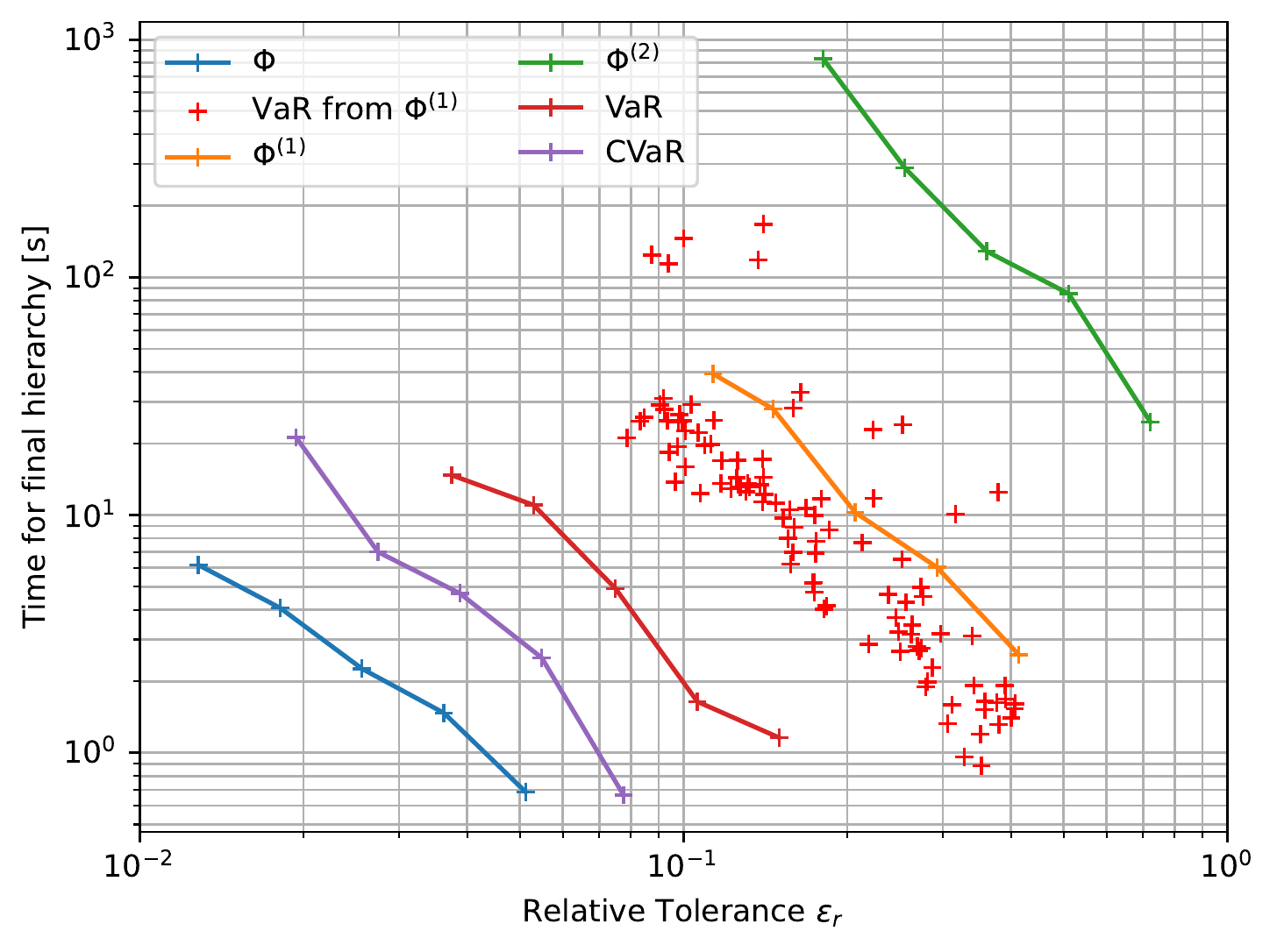}
  \caption{Complexity for different statistics}
  \label{fig:poisson_complexity_all_stats}
\end{figure}

Lastly, Fig.~\ref{fig:hierarchies_cmlmc} shows the optimal level-wise sample sizes computed for each intermediate tolerance of one simulation of the \gls{cmlmc} algorithm aimed at estimating the $70\%$-\gls{cvar} to the finest tolerance simulated.
This demonstrates the successive refinement strategy of the \gls{cmlmc} algorithm, where the hierarchy is calibrated based on a decreasing sequence of tolerances.
This can be seen in the increased level-wise sample sizes in the hierarchy with successive iterations.
In addition, the red dashed line shows the expected decay rate of $N_l$ over the levels $l$ as predicted by Eq.~\eqref{eq:opt_samples}, for the variance decay and cost growth exponents $\beta$ and $\gamma$ obtained from least squares fits on the estimates of $\tilde{V_l}$ and $\costsim$ over the levels $l$, respectively.

\begin{figure}[H]
\centering
\includegraphics[width=0.6\textwidth]{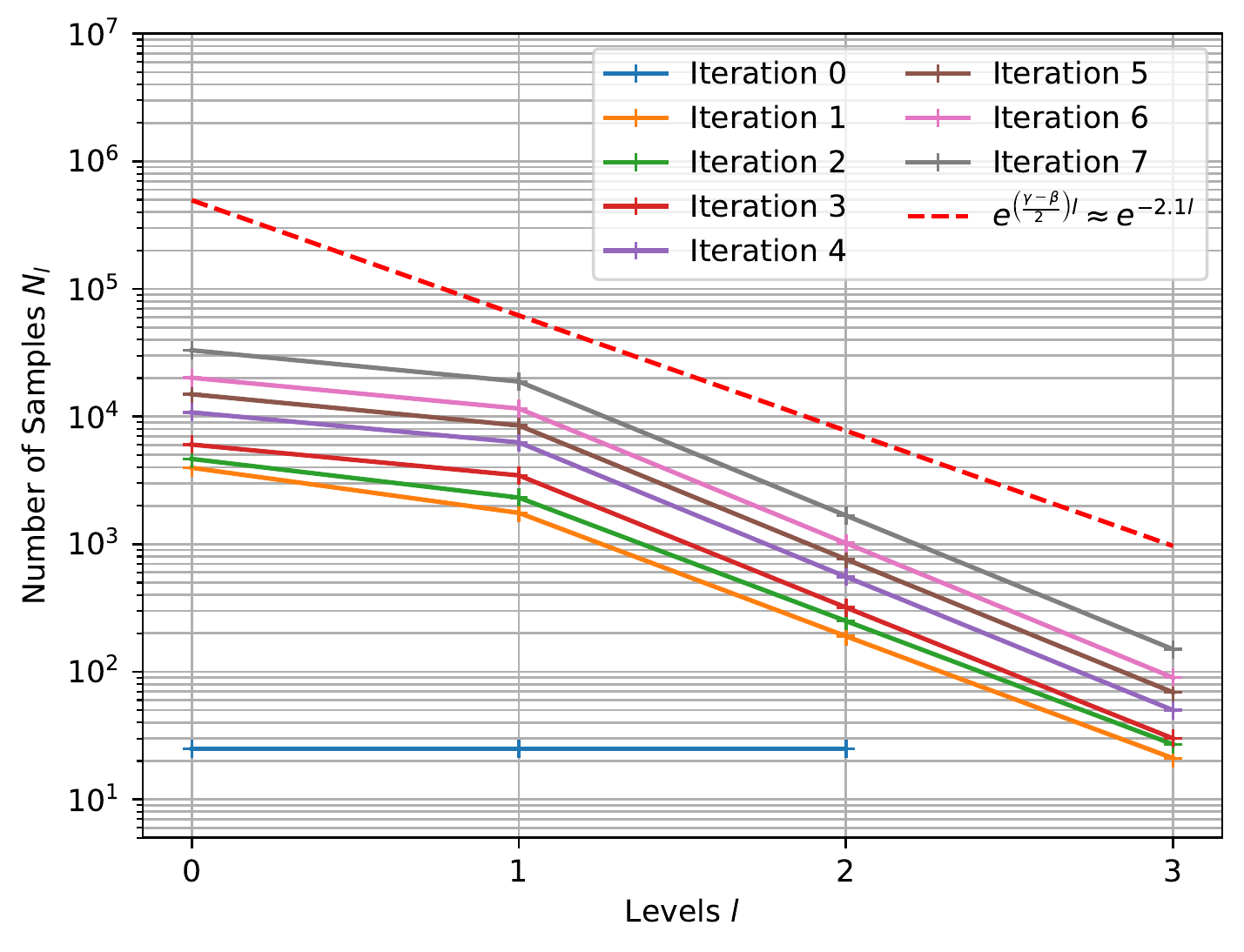}
\caption{Hierarchy evolution over \gls{cmlmc} iterations}
\label{fig:hierarchies_cmlmc}
\end{figure}
\FloatBarrier

\subsection{Black Scholes Stochastic Differential Equation} \label{sec:black_scholes}
We consider in this section the simulation of the price of a financial asset in time using the Black-Scholes \gls{sde}. 
The price of the asset is modelled as a geometric Brownian motion:
\begin{equation}
dS = r S\,dt + \sigma S\,dW\;,\quad S(0)=S_0,\;t \in (0,T].\label{eq:gbm:prob_b}
\end{equation}
with $r, \sigma, S_0 > 0$.
The \gls{qoi} for this example, whose distribution we wish to quantify, is the discounted European call option, defined as follows:
\begin{equation} 
\qoi := e^{-rT}\max\bigl(S(T)-K,0\bigr)\;,
\end{equation}
where $K>0$ denotes the agreed strike price and $T>0$ the pre-defined
expiration date. 
It is known that the solution $S(T)$ to Eq.~\eqref{eq:gbm:prob_b} at time $T$ is a log-normally distributed random variable with mean $S_0e^{rT}$ and variance ${S_0}^2e^{2rT}\left(e^{\sigma^2 T}-1\right)$.
Hence, the \gls{cdf} of $\qoi$ is:
\begin{align}
F_\qoi(\theta) &= \begin{cases}
\frac{1}{2} + \frac{1}{2}\erf\biggl(\frac{\sqrt{2}\bigl(\sigma^2 T - 2rT + 2\ln{\bigl(K+e^{rT}\theta\bigr)} - 2\ln{(S_0)}\bigr)}{4\sigma\sqrt{T}}\biggr) \;,& \theta\ge 0\;,\\
0\;,& \theta<0\;,
\end{cases} \label{eq:cdf_bs} \\
\text{where}\quad\erf(z) &= \frac{2}{\sqrt{\pi}}\int_0^z e^{-s^2}\, ds.
\end{align}
Using the \gls{cdf} in Eq.~\eqref{eq:cdf_bs}, it is then straightforward to compute reference values for the \gls{var} and \gls{cvar}.
Table \ref{tab:gbm:quantiles} lists the values of the \gls{var} and \gls{cvar} for different significances $\tau$ and for $r=0.05$, $\sigma = 0.2$, $T=1$, $K=10$, and $S_0=10$.
The corresponding \gls{cdf} is plotted in Fig.~\ref{fig:atom}.
We are interested in estimating the \gls{cvar} value corresponding to a significance of $\tau = 0.7$ while ensuring the numerical stability of the rescaling ratio $r_e$ in Eq.~\eqref{eq:rescaling_all} and hence, select an interval $\Theta = [0.5,2.0]$.
We utilize the \gls{cmlmc} algorithm coupled with the novel error estimators described in Section~\ref{sec:novel} in order to calibrate the \gls{mlmc} estimator for the \gls{cvar}. 

\begin{figure}[ht]
	\begin{minipage}[b]{0.35\textwidth}
		\centering
		\includegraphics[width=\textwidth]{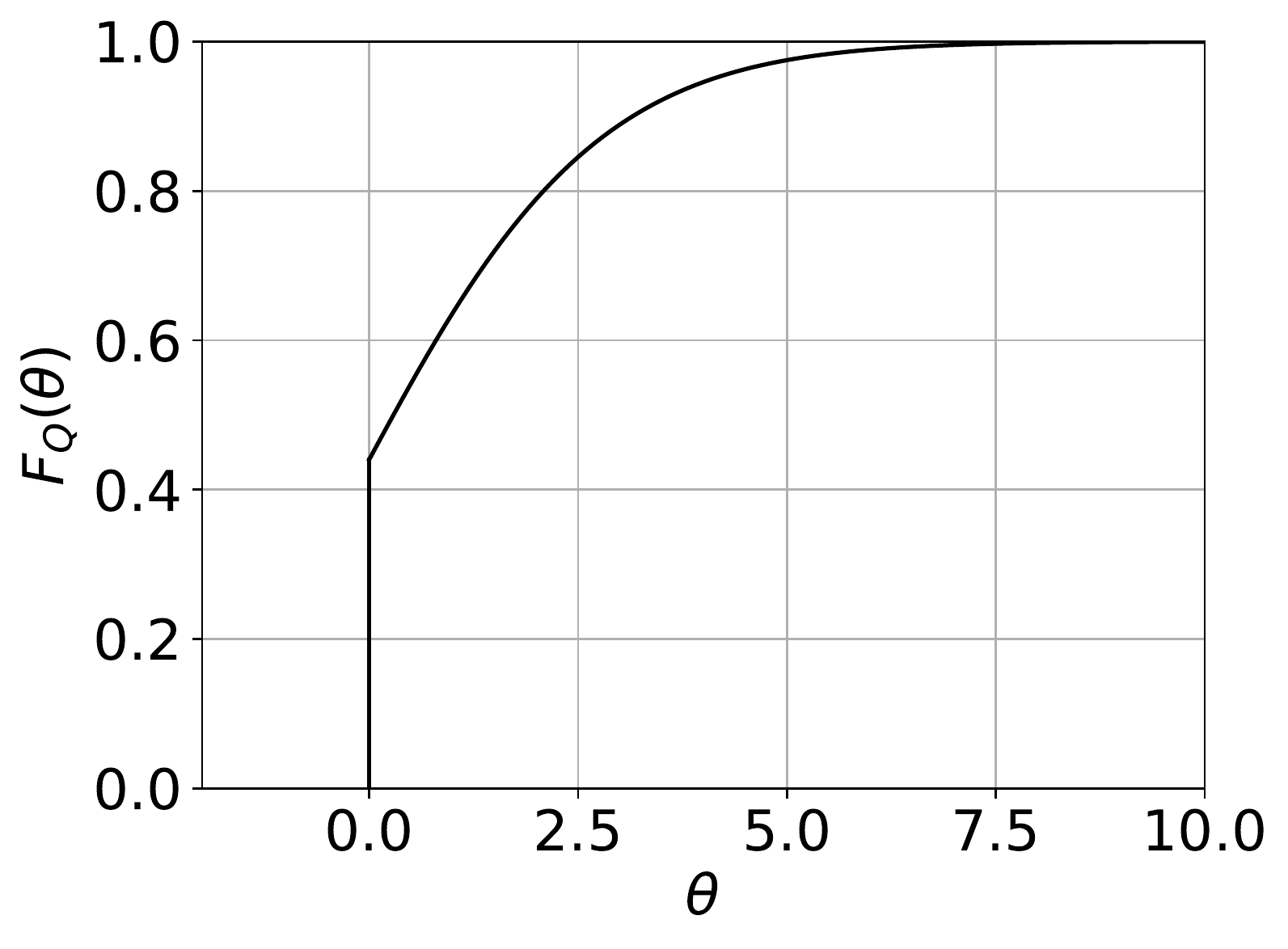}
		\captionof{figure}{\gls{cdf} $F_{\qoi}$ of $\qoi$ for the \gls{sde} problem.}
		\label{fig:atom}
	\end{minipage}
	\hfill
	\begin{minipage}[b]{0.55\textwidth}
		\centering
		\begin{tabular}{ccc}  
			\toprule
			$\tau$ & $q_{\tau} = F_{\qoi}^{-1}(\tau)$ & $c_{\tau}$\\
			\midrule
			$0.6$ & $ 0.799151$ & $  2.455898$\\
			$0.7$ & $ 1.373571$ & $  2.914953$\\
			$0.8$ & $ 2.086595$ & $  3.515684$\\
			$0.9$ & $ 3.153379$ & $  4.460298$\\
			\bottomrule\
			\vspace{0.12in}
		\end{tabular}
		\captionof{table}{\gls{var} and \gls{cvar} values for the \gls{qoi} associated with \gls{sde} problem.}
		\label{tab:gbm:quantiles}
	\end{minipage}
\end{figure}

For the numerical experiments, the \gls{sde} in Eq.~\eqref{eq:gbm:prob_b} is discretised on a hierarchy of uniform grids given by $t^l_i = i \Delta t_l$, for level $l \in \{0,...,L\}$, with $i \in \{0,1,...,N^l_T\}$ such that $N^l_T \Delta t_l = T$.
The grid sizes $\Delta t_l$ are selected such that $\Delta t_l = \Delta t_0 2^{-l}$, giving rise to a hierarchy of nested grids.
Furthermore, we use the Euler Maruyama scheme to discretise the problem on this uniform grid.
Denoting the discretised approximation of $S(t_i)$ on level $l$ as $S_{t^l_i}$, the scheme for level $l$ reads:
\begin{align}
S_{t^l_{i+1}} &= S_{t^l_i} + r \Delta t_l  S_{t^l_i} + \sigma \sqrt{\Delta t_l} S_{t^l_i} \xi_i,\quad \xi_i  \overset{\text{i.i.d}}{\sim} \mathcal{N}(0,1),\quad i\in\{0,...,N^l_T-1\},
\end{align}
with $S_{t^l_0} = S_0$.
Correlated realizations are generated on a pair of levels $l$ and $l-1$ by using the same realization of the Brownian path on both levels.

The performance of the \gls{mlmc} method can be summarized as follows. Fig.~\ref{fig:black_scholes_reliability} shows the results of a reliability study analogous to the one conducted in Section~\ref{sec:poisson} for the Poisson problem.
For each \gls{cmlmc} simulation, an estimate of the \gls{mse} of the \gls{cvar} is produced using the novel error estimation procedure described in Sections~\ref{sec:novel} and \ref{sec:adaptivity}.
The \gls{mse} estimates are compared with the corresponding true squared errors computed using the estimated value of the \gls{cvar} obtained from the \gls{cmlmc} algorithm and the reference value in Table~\ref{tab:gbm:quantiles} corresponding to $\tau = 0.7$.
As can be seen from Fig.~\ref{fig:black_scholes_reliability}, the estimated \gls{mse} is larger than the ``true'' \gls{mse} by a factor of approximately 10, which we consider acceptable for practical applications.
In addition, Fig.~\ref{fig:black_scholes_complexity} shows the computational cost to compute the optimal hierarchy for a given tolerance on the \gls{mse}. 
The plot demonstrates that the complexity behaviour matches the best case scenario predicted by Proposition~\ref{thm:complexity:MLMC:spline}.
It is also noteworthy that the \gls{mlmc} estimator not only provides a significantly improved computational complexity of $\mathcal{O}(\epsilon^{-2})$ compared to $\mathcal{O}(\epsilon^{-3})$ for the Monte Carlo method, but also that the computational cost of the \gls{mlmc} estimator is already smaller by one to two orders of magnitude even for the largest tolerance considered.
For comparison, the Monte Carlo cost is plotted in Fig.~\ref{fig:black_scholes_complexity} using the procedure as described in Section~\ref{sec:poisson}.

\begin{figure}[H]
  \centering
  \begin{subfigure}{0.45\textwidth}
    \centering
    \includegraphics[width=\textwidth]{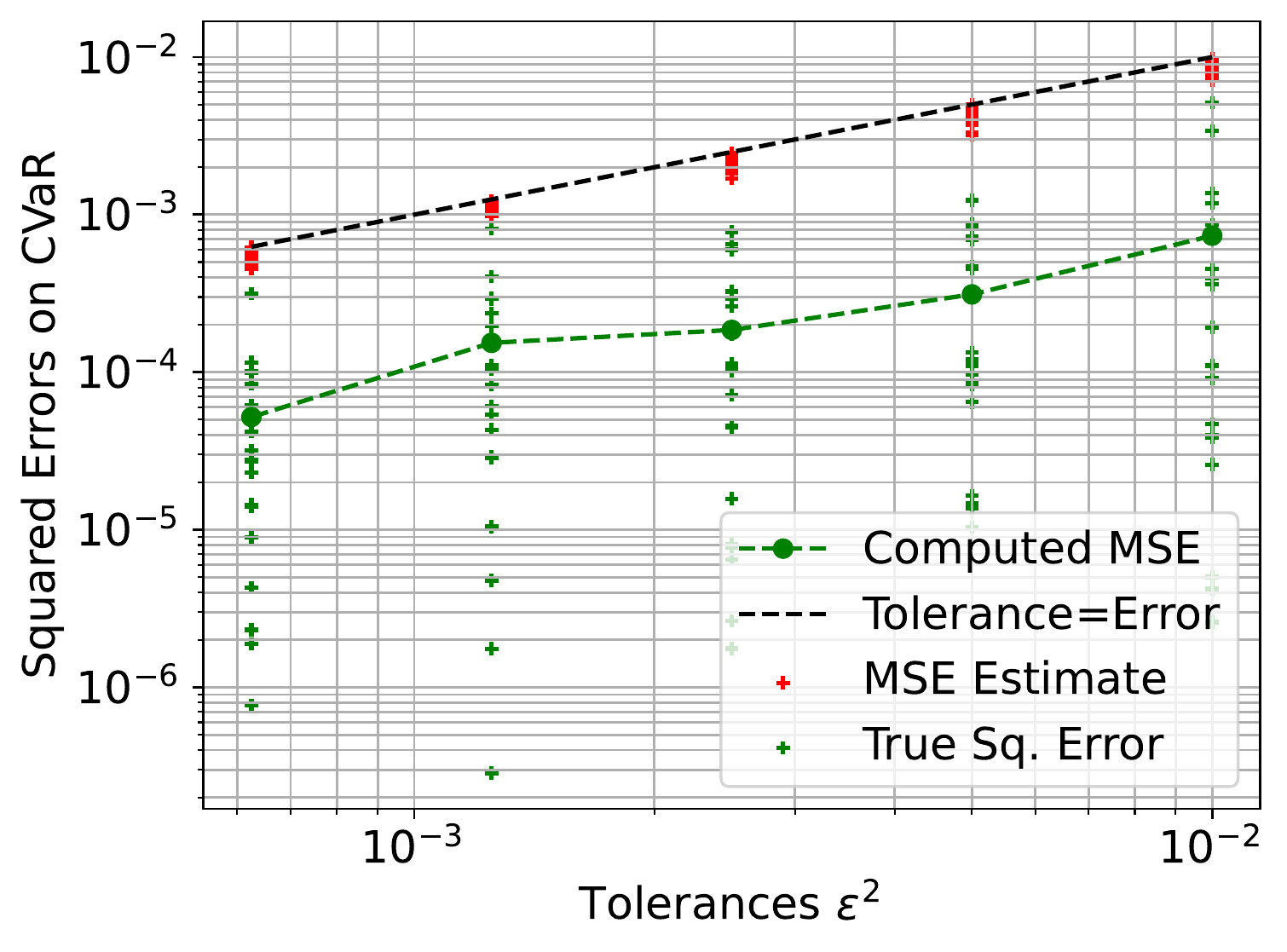}
    \caption{Reliability of error estimator}
    \label{fig:black_scholes_reliability}
  \end{subfigure}
  \begin{subfigure}{0.45\textwidth}
    \centering
    \includegraphics[width=\textwidth]{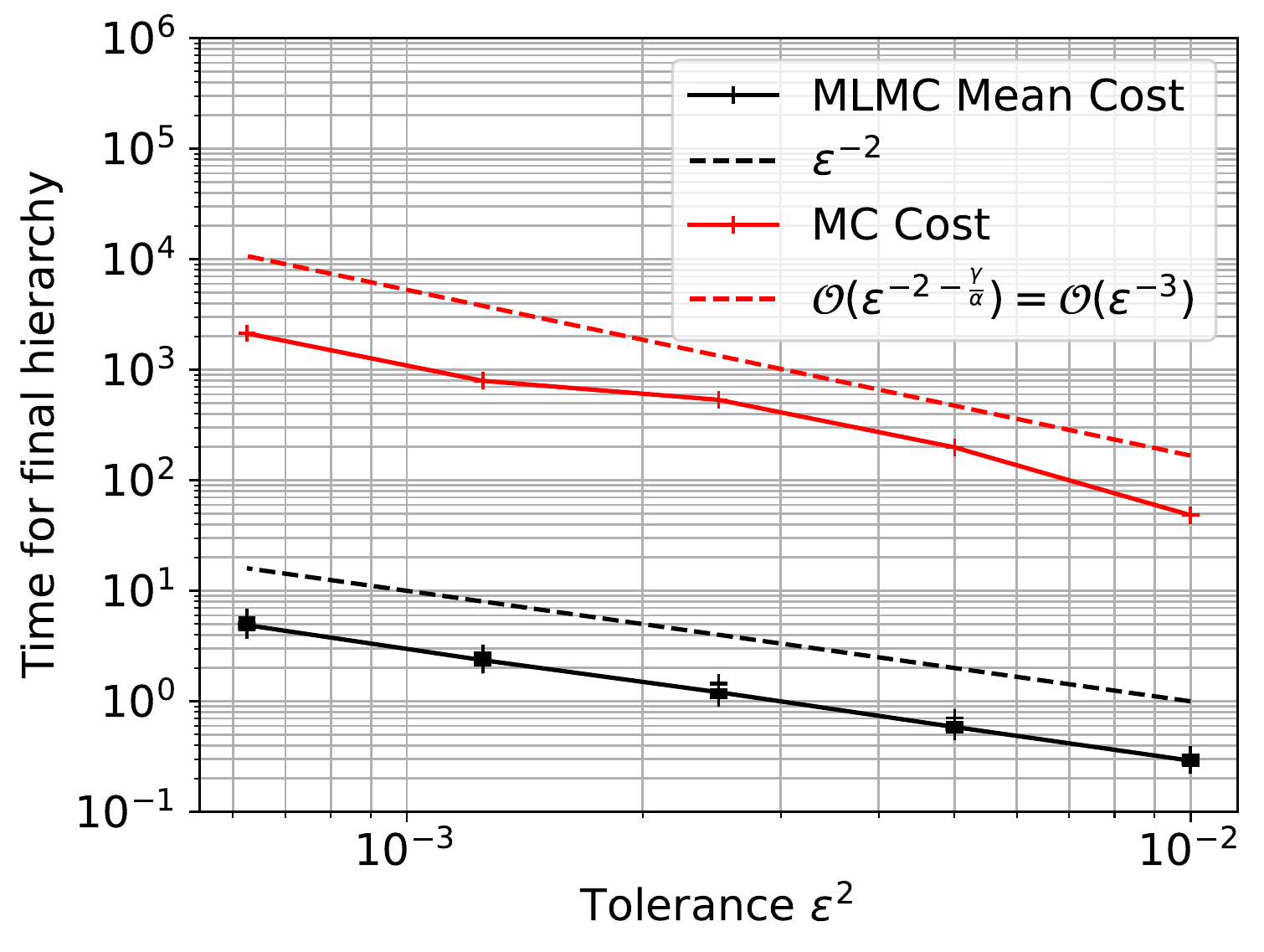}
    \caption{Complexity behaviour}
    \label{fig:black_scholes_complexity}
  \end{subfigure}
  \caption{Summary of results for the Black Scholes \gls{sde}}
  \label{fig:black_scholes_performance}
\end{figure}

Lastly, we present in Fig.~\ref{fig:phim_var_error_black_scholes} a similar study as in Fig.~\ref{fig:phim_var_error_poisson}. 
For the same set of reliability simulations as in Fig.~\ref{fig:black_scholes_reliability}, we compute the true squared error and \gls{mse} estimate for $\Phi^{(m)},\; m \in \{0,1,2\}$, as well as the \gls{var}, and plot it against the tolerance on the \gls{mse} used in the \gls{cvar} calculation.
In addition, for the \gls{var}, shown in Fig.~\ref{fig:var_black_scholes}, we estimate the \gls{mse} also on a smaller interval $\Theta = [1.35,1.4]$ than the one used for estimating the \gls{cvar} for comparison.
The results are comparable to the case of the Poisson problem in Section~\ref{sec:poisson}; although the novel error estimators provide accurate estimates of the \gls{mse} of $\Phi^{(m)},\;m\in\{0,1,2\}$, the error estimates for the \gls{var} are comparatively more conservative but improve with smaller interval size selection.

\begin{figure}[H]
\centering
\begin{subfigure}{.48\textwidth}
    \centering
    \includegraphics[width=\textwidth]{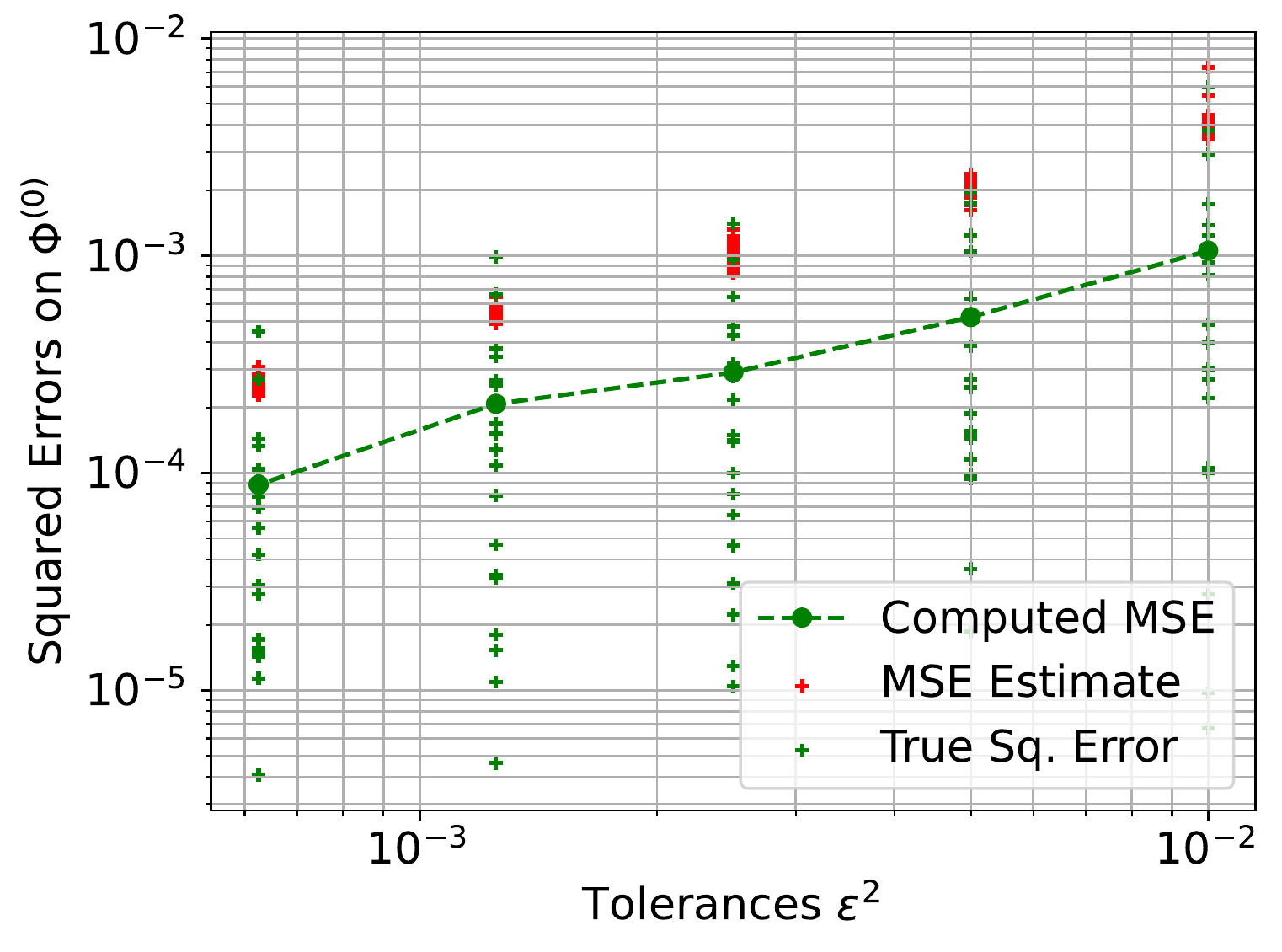}
    \caption{$m=0$}
\end{subfigure}
\begin{subfigure}{.48\textwidth}
    \centering
    \includegraphics[width=\textwidth]{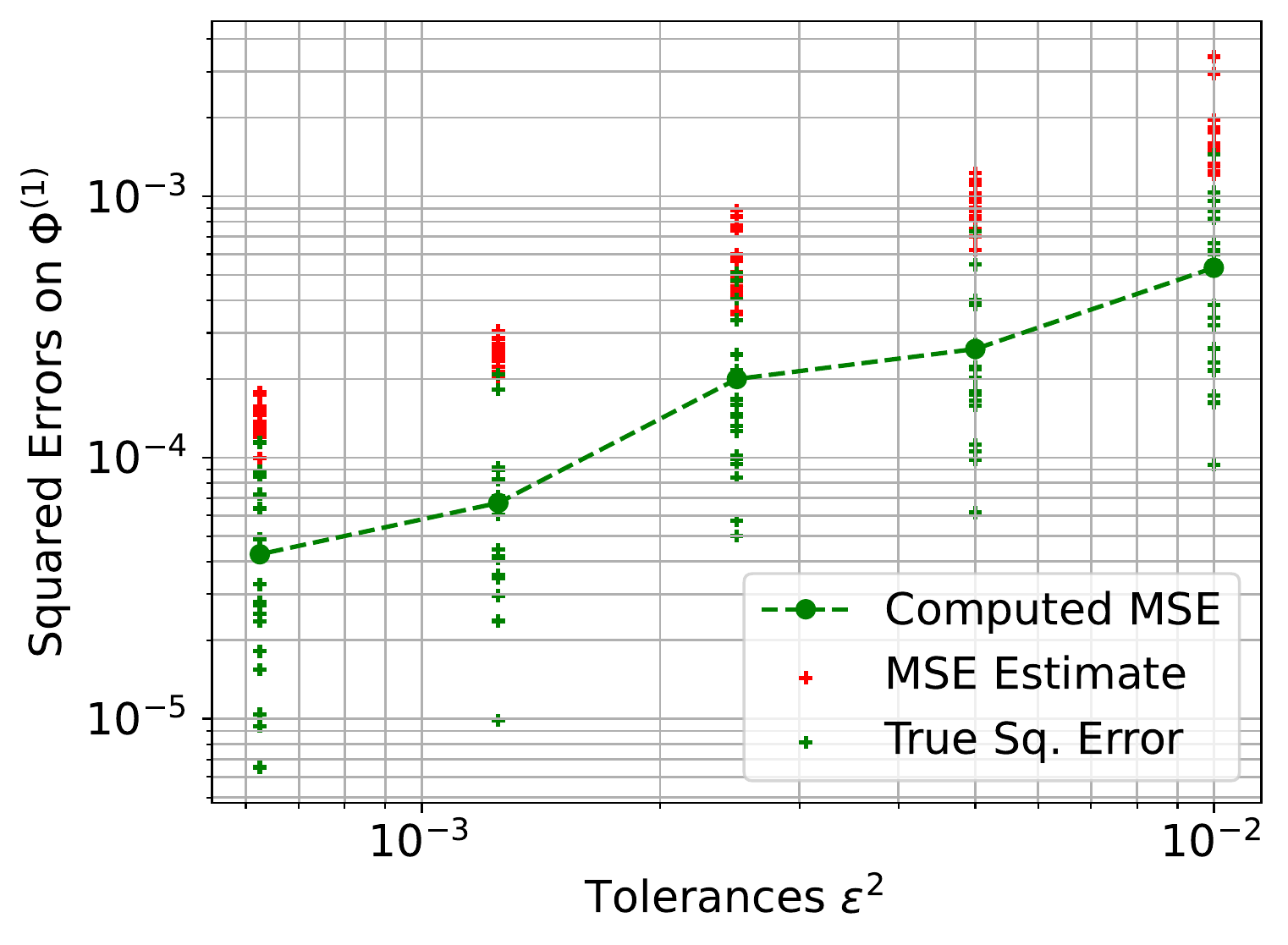}
    \caption{$m=1$}
\end{subfigure}

\begin{subfigure}{.48\textwidth}
    \centering
    \includegraphics[width=\textwidth]{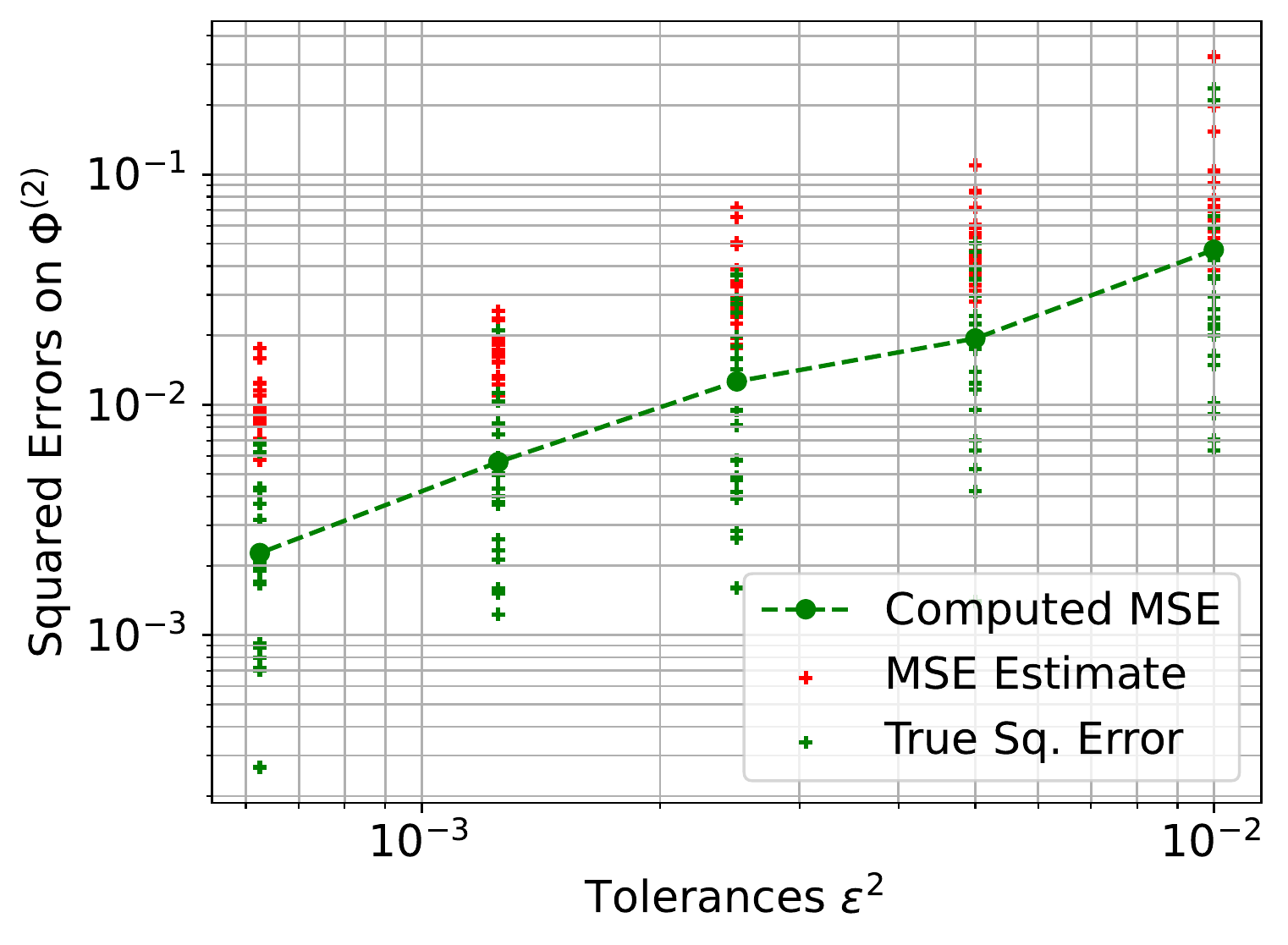}
    \caption{$m=2$}
\end{subfigure}
\begin{subfigure}{.48\textwidth}
    \centering
    \includegraphics[width=\textwidth]{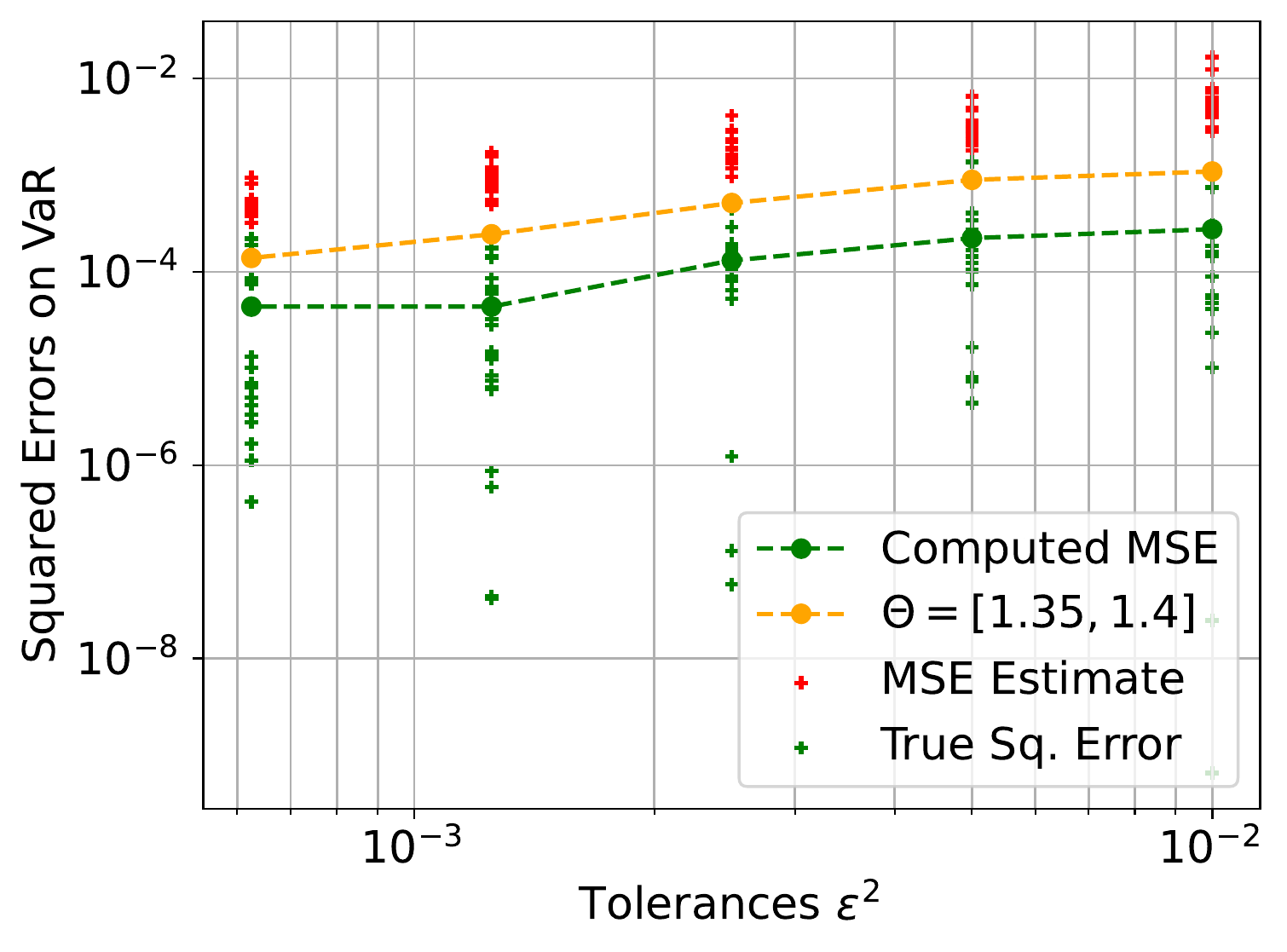}
    \caption{\gls{var}}
    \label{fig:var_black_scholes}
\end{subfigure}
\caption{Reliability of error estimator for $\Phim$ and \gls{var} for Black Scholes problems}
\label{fig:phim_var_error_black_scholes}
\end{figure}

\FloatBarrier
\subsection{Navier-Stokes flow over a cylinder in a channel}\label{sec:cylinder}
To demonstrate the performance of the \gls{mlmc} estimator and the novel error estimators on a more challenging problem, we consider a two-dimensional steady incompressible fluid flow over a cylinder placed asymmetrically in a channel.
The goal here is to study the effects of random inlet perturbations on the distributions of force and moment coefficients on the cylinder. 

The domain of the problem is a rectangle with a circular cylinder removed and can be defined as $D = [0, 2.2] \times [0, 0.41] \backslash B_r(0.2,0.2),\;r=0.05$, where $B_r(x,y)$ denotes a circle centred at the coordinate $(x,y) \in D \subset \setR^2$ with radius $r > 0$.
The flow is characterized by the velocity field $u:D \to \setR^2$ and pressure field $p: D \to \setR$.
The velocity and pressure fields are governed by the steady incompressible Navier-Stokes equations:
\begin{align}
(u\cdot \nabla)u - \nu \Delta u + \nabla p &= 0, \label{eq:ns1}\\ 
\nabla \cdot u &= 0, \quad\text{in } D, \label{eq:ns2}
\end{align}
where $\nu = 0.01$ denotes the kinematic viscosity.
The boundary conditions are as follows.
At the inflow boundary ($x=0$), we consider a random inlet profile, which consists of a parabolic mean profile on which harmonics with random amplitudes are added:
\begin{align}
u(0,y) &= \left( \frac{4 U y(0.41-y)}{0.41^2} + u_{r}, 0 \right),\\
u_r(y) &= \sigma \sum_{j=1}^{N_h} \xi_j e^{-j} \sin \left(\frac{j\pi y}{0.41}\right), \qquad \xi_j \overset{i.i.d}{\sim} \mathcal{N}(0,1), \label{eq:inlet}
\end{align}
where $U = 4.0$ is the peak velocity of the parabolic profile, $\sigma = 0.5$ denotes a strength parameter and $N_h = 8$ denotes the number of harmonics superimposed. 
Fig.~\ref{fig:inlet} shows 10 different realizations of the inlet profile given in Eq.~\eqref{eq:inlet}, plotted over the parabolic mean profile.
\begin{figure}[ht]
\centering
\includegraphics[width=0.6\textwidth]{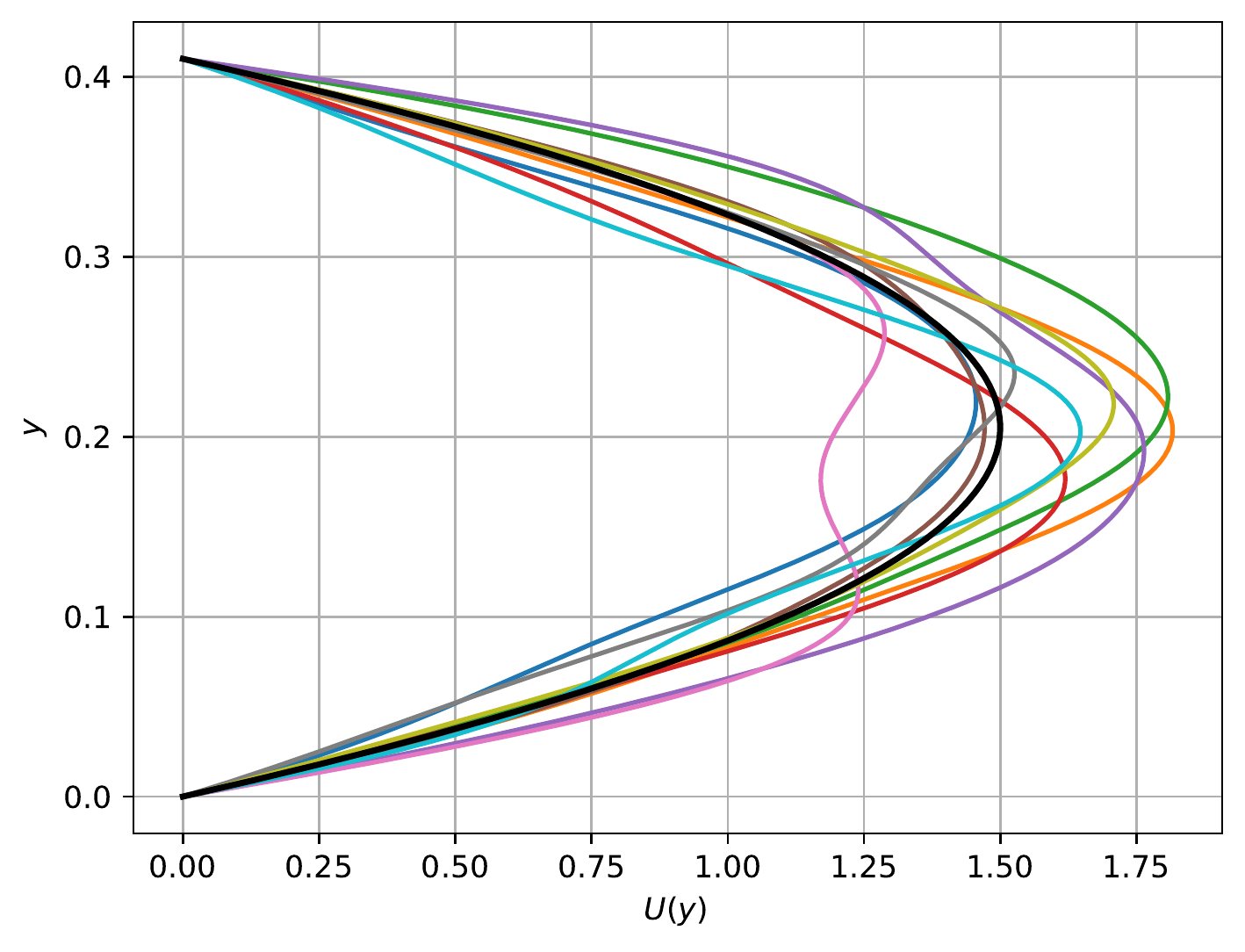}
\caption{Inlet profile realizations (in color) plotted over parabolic mean profile (in black)}
\label{fig:inlet}
\end{figure}

On the bottom and top channel walls ($y=0$ and $y=0.41$ respectively), no-slip boundary conditions are prescribed.
On the outlet ($x=2.2$), a zero-stress boundary condition is prescribed with the form
\begin{align}
(\nu \nabla u - pI)n = 0,
\end{align}
where $n =[1,0]^T$ denotes the outward boundary normal vector and $I \in \setR^{2 \times 2}$ denotes the identity matrix. 

For a peak velocity of $U=4.0$, the area-weighted inlet velocity is $U_{in}=2.667$.
Taking the reference length to be the diameter of the cylinder, the Reynolds' number is
\begin{align}
\text{Re} = \frac{2 U_{in} r}{\nu} = \frac{2.667 \times 0.1}{0.01} = 26.67.
\end{align}
The \gls{qoi} considered is the drag coefficient $C_d$, whose value is computed as follows.
First, we compute the drag and lift forces $F_d$ and $F_l$ on the cylinder, which are given respectively by:
\begin{align}
\left[\begin{matrix} F_d \\ F_l \end{matrix}\right] = \int_{\partial B_r} (\nu \nabla u - pI) n ds,
\end{align}
where $\partial B_r$ denotes the surface of the circle over which the stress is integrated.
The drag coefficient $C_d$ is then computed from the drag force as:
\begin{align}
C_{d} = \frac{F_{d}}{U_{in}^2 r}.
\end{align}
The domain $D$ is discretised with a non-uniform triangulation.
Reference mesh size values are prescribed on the surface of the circle, as well as at the corners of the domain.
The coarsest two meshes, corresponding to levels $l=0$ and $l=1$ of those simulated, are shown in Fig.~\ref{fig:mesh}. 
Each finer level is produced by reducing the prescribed reference mesh sizes by a factor $\sqrt{2}$ from the previous level and re-applying the triangulation.
The meshes computed as a result are non-nested.
Table~\ref{tab:mesh_sizes} shows the minimum and maximum mesh sizes $h_{min}$ and $h_{max}$, as well as the number of vertices for each of the meshes considered in the hierarchy. 
As can be seen from the table, the number of vertices approximately doubles with every level. 

The problem is implemented using the FEniCS finite element software \cite{logg2012automated}.
P2-P1 Taylor-Hood elements are used for the velocity and pressure fields.
The resulting non-linear problem is solved using Newton iterations with a relative tolerance of $10^{-10}$ on the residual.
Linear systems are solved using a sparse direct solver \cite{MUMPS:1, MUMPS:2}.

\begin{figure}[ht]
\centering
\includegraphics[width=0.8\textwidth]{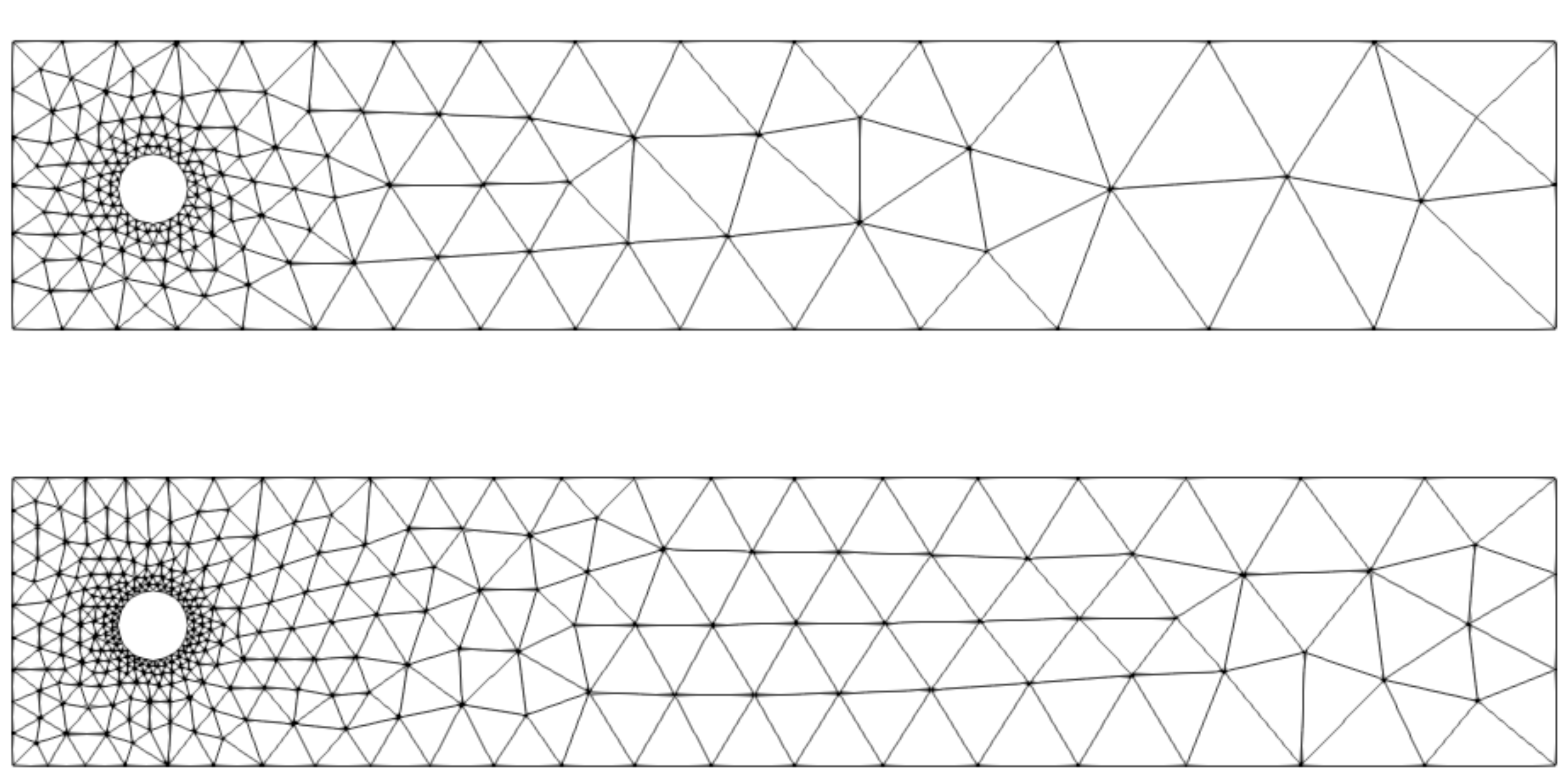}
\caption{Meshes for cylinder problem for level $l=0$ (top) and $l=1$ (bottom)}
\label{fig:mesh}
\end{figure}

\begin{table}[ht]
\centering
\begin{tabular}{cccc}
	\toprule
	Level & $h_{min}$ & $h_{max}$ & Vertices\\
	\midrule
	$0$ & $ 1.31 \times 10^{-2} $ & $ 2.65 \times 10^{-1} $ & $199$\\
	$1$ & $ 9.80 \times 10^{-3} $ & $ 1.87 \times 10^{-1} $ & $333$\\
	$2$ & $ 6.54 \times 10^{-3} $ & $ 1.44 \times 10^{-1} $ & $593$\\
	$3$ & $ 4.91 \times 10^{-3} $ & $ 9.82 \times 10^{-2} $ & $1073$\\
	$4$ & $ 2.91 \times 10^{-3} $ & $ 7.58 \times 10^{-2} $ & $2038$\\
	$5$ & $ 2.28 \times 10^{-3} $ & $ 5.62 \times 10^{-2} $ & $3857$\\
	\bottomrule
\end{tabular}
\caption{Mesh parameters for the Navier Stokes problem}
\label{tab:mesh_sizes}
\end{table}

As in previous sections, we aim to estimate the $70\%$-\gls{cvar} using the \gls{cmlmc} method presented in this work, for which we select the interval $\Theta = [0.3,0.8]$.
We perform a reliability study identical to the ones conducted in Section~\ref{sec:poisson} and Section~\ref{sec:black_scholes} for the Poisson and Black Scholes problems, respectively. 
The reference value of the $70\%$-\gls{cvar} is, however, not available in this case. 
Instead, a numerical reference is computed as follows.
We conduct 20 independent repetitions of the \gls{cmlmc} algorithm, tuned to a tolerance that is one quarter of the finest tolerance tested for in what follows. 
The reference value is then taken to be the average over the 20 estimates of the \gls{cvar} produced by these simulations. 
The resultant reliability plot is shown in Fig.~\ref{fig:navier_stokes_reliability}.
As can be seen in Fig.~\ref{fig:navier_stokes_reliability}, the estimated squared errors are approximately 1.5 orders conservative on the true \gls{mse}.
Although considerably more conservative than for the Poisson and Black-Scholes problems, we deem the error estimator still acceptable and leading to practically computable hierarchies.
In addition, the complexity plot is shown in Fig.~\ref{fig:navier_stokes_complexity}.
The reference Monte Carlo cost is computed using the same procedure as described in Section~\ref{sec:poisson} for the Poisson problem.
The figure once again demonstrates that the complexity behaviour matches the best case scenario predicted by Proposition~\ref{thm:complexity:MLMC:spline}.
The \gls{mlmc} estimator shows a complexity of $\mathcal{O}(\epsilon^{-2})$ as compared to a complexity of $\mathcal{O}(\epsilon^{-3.6})$ in the Monte Carlo case.
In addition, even for the largest tolerance considered, the \gls{mlmc} estimator is three orders faster than the corresponding Monte Carlo estimator.

\begin{figure}[H]
  \centering
  \begin{subfigure}{0.45\textwidth}
    \centering
    \includegraphics[width=\textwidth]{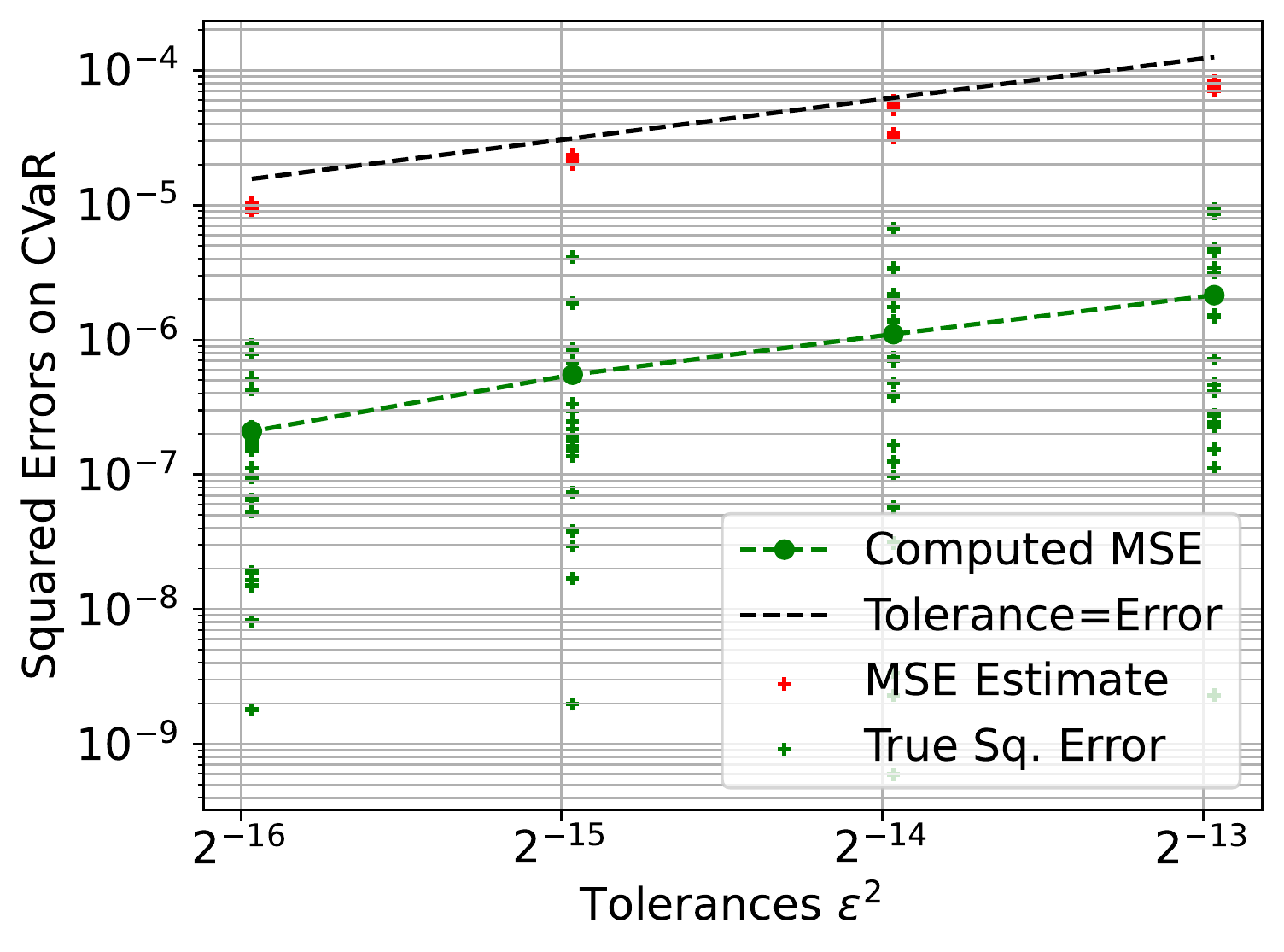}
    \caption{Reliability of error estimator}
    \label{fig:navier_stokes_reliability}
  \end{subfigure}
  \begin{subfigure}{0.45\textwidth}
    \centering
    \includegraphics[width=\textwidth]{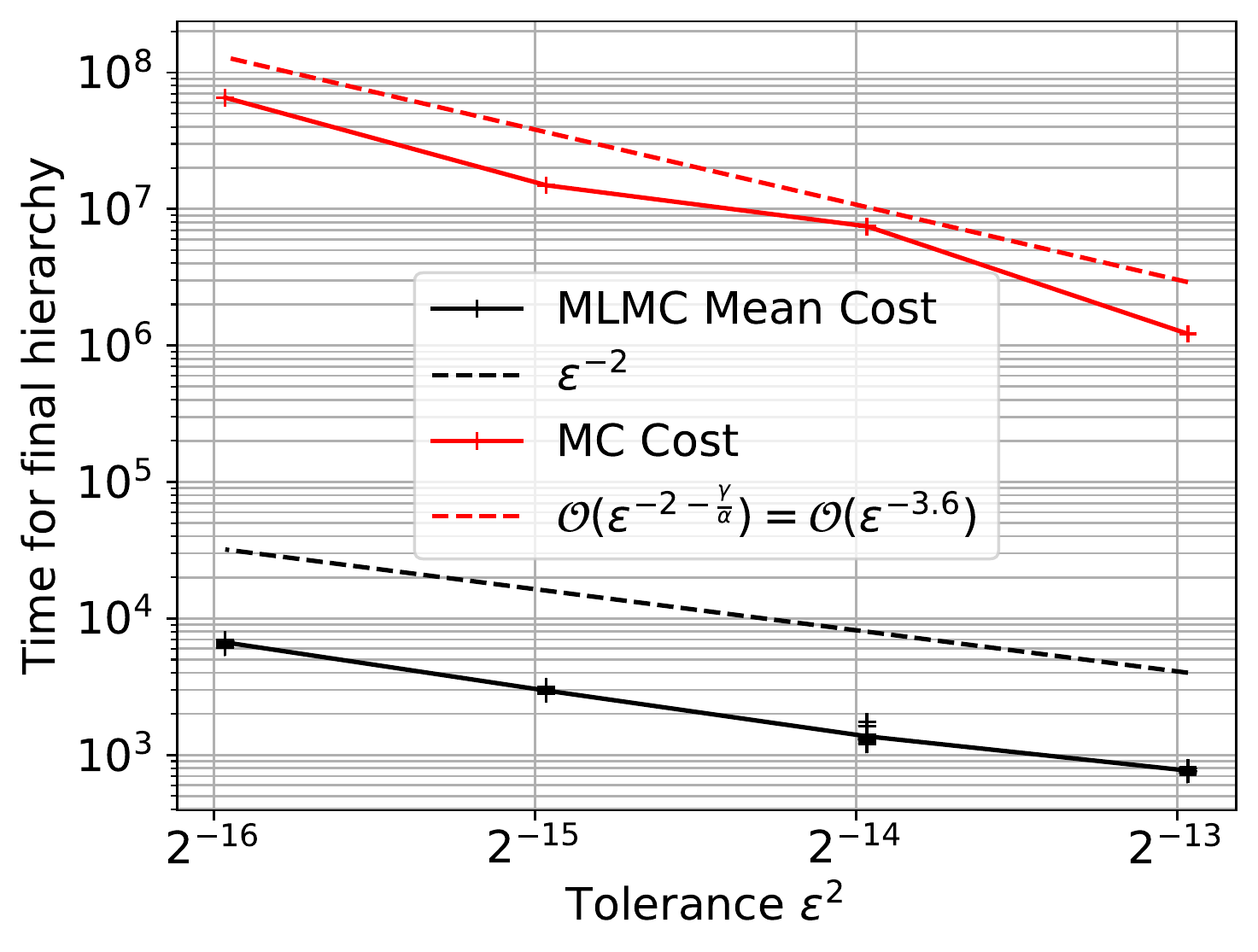}
    \caption{Complexity behaviour}
    \label{fig:navier_stokes_complexity}
  \end{subfigure}
  \caption{Summary of results for the Navier-Stokes problem}
  \label{fig:navier_stokes_performance}
\end{figure}

Fig.~\ref{fig:phim_var_error_navier_stokes} shows the true and estimated \glsplural{mse} on $\Phi^{(m)},\;m\in\{0,1,2\}$, as well as the \gls{var} for the same set of simulations as in Fig.~\ref{fig:navier_stokes_reliability}.
In addition, \gls{mse} estimates are computed for the \gls{var} case with a smaller interval $\Theta = [0.49,0.51]$.
We note that the results are similar to those seen for the Poisson and Black Scholes problems, although the error estimators for all four statistics are relatively more conservative when compared to those problems.
In addition, although the reduction of interval size leads to a reduction in the \gls{mse} estimates predicted for the \gls{var}, the reduction is not as significant as in the case of the Poisson and Black-Scholes problems.

\begin{figure}[H]
\centering
\begin{subfigure}{.48\textwidth}
    \centering
    \includegraphics[width=\textwidth]{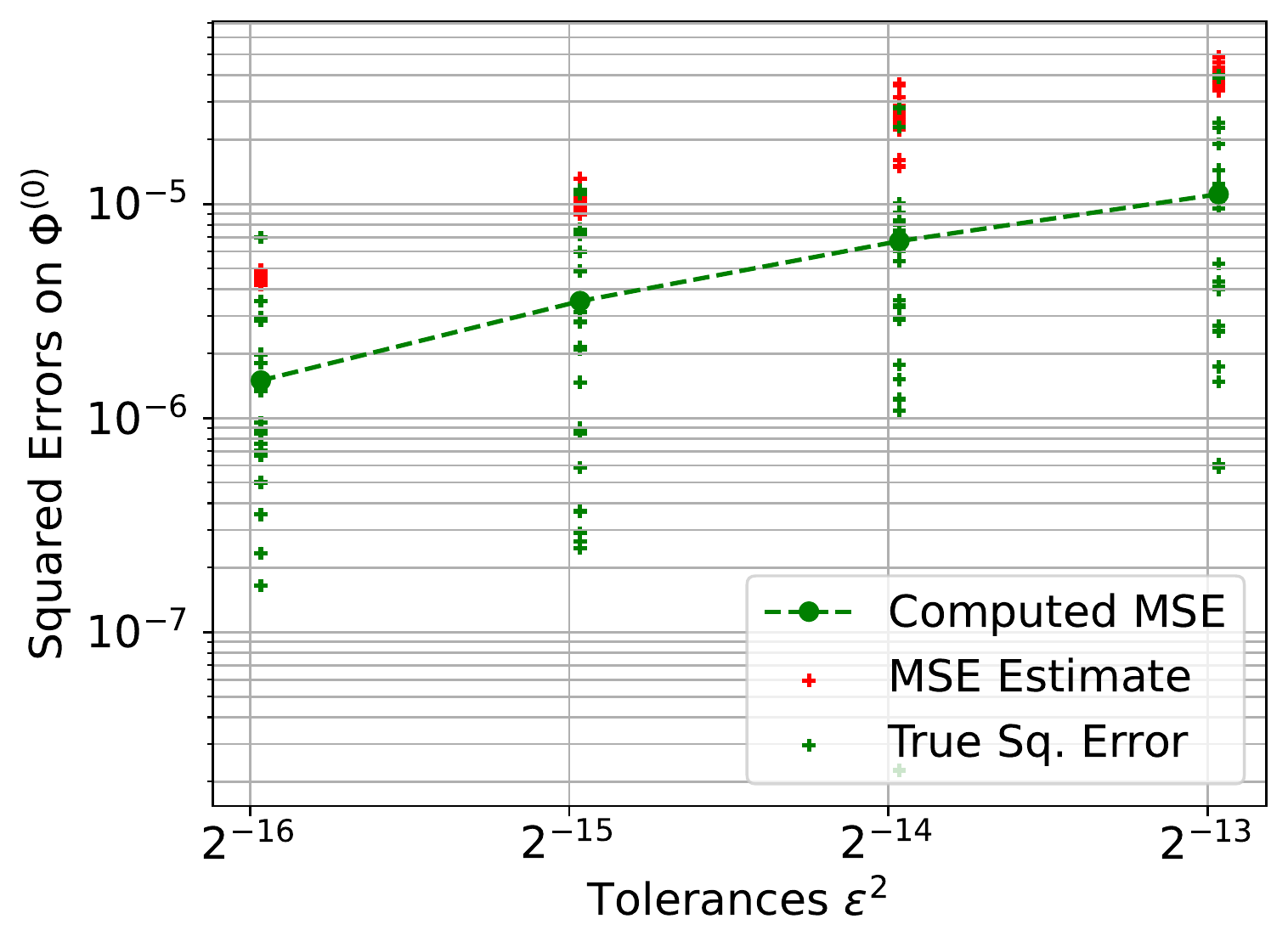}
    \caption{$m=0$}
\end{subfigure}
\begin{subfigure}{.48\textwidth}
    \centering
    \includegraphics[width=\textwidth]{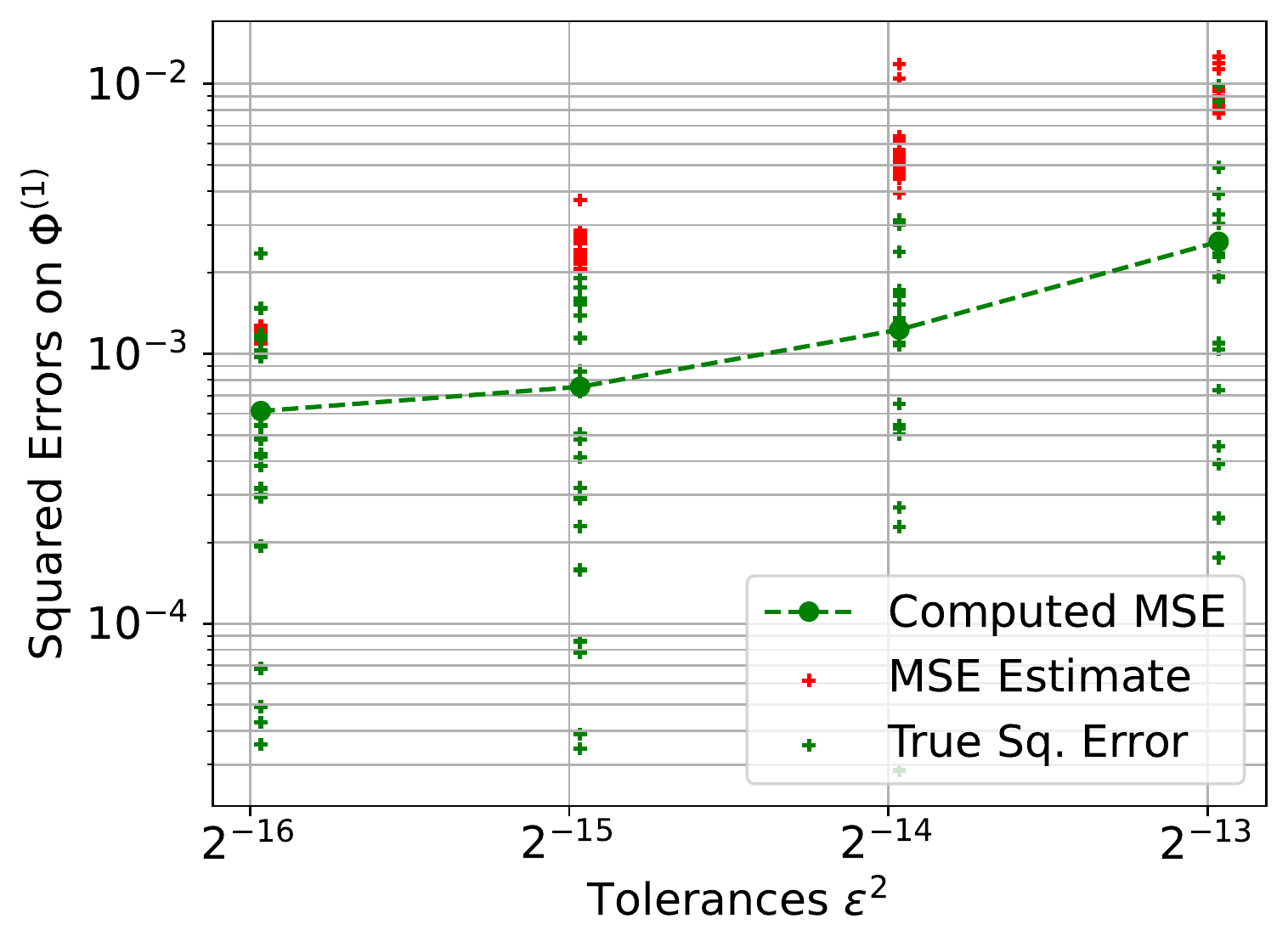}
    \caption{$m=1$}
\end{subfigure}

\begin{subfigure}{.48\textwidth}
    \centering
    \includegraphics[width=\textwidth]{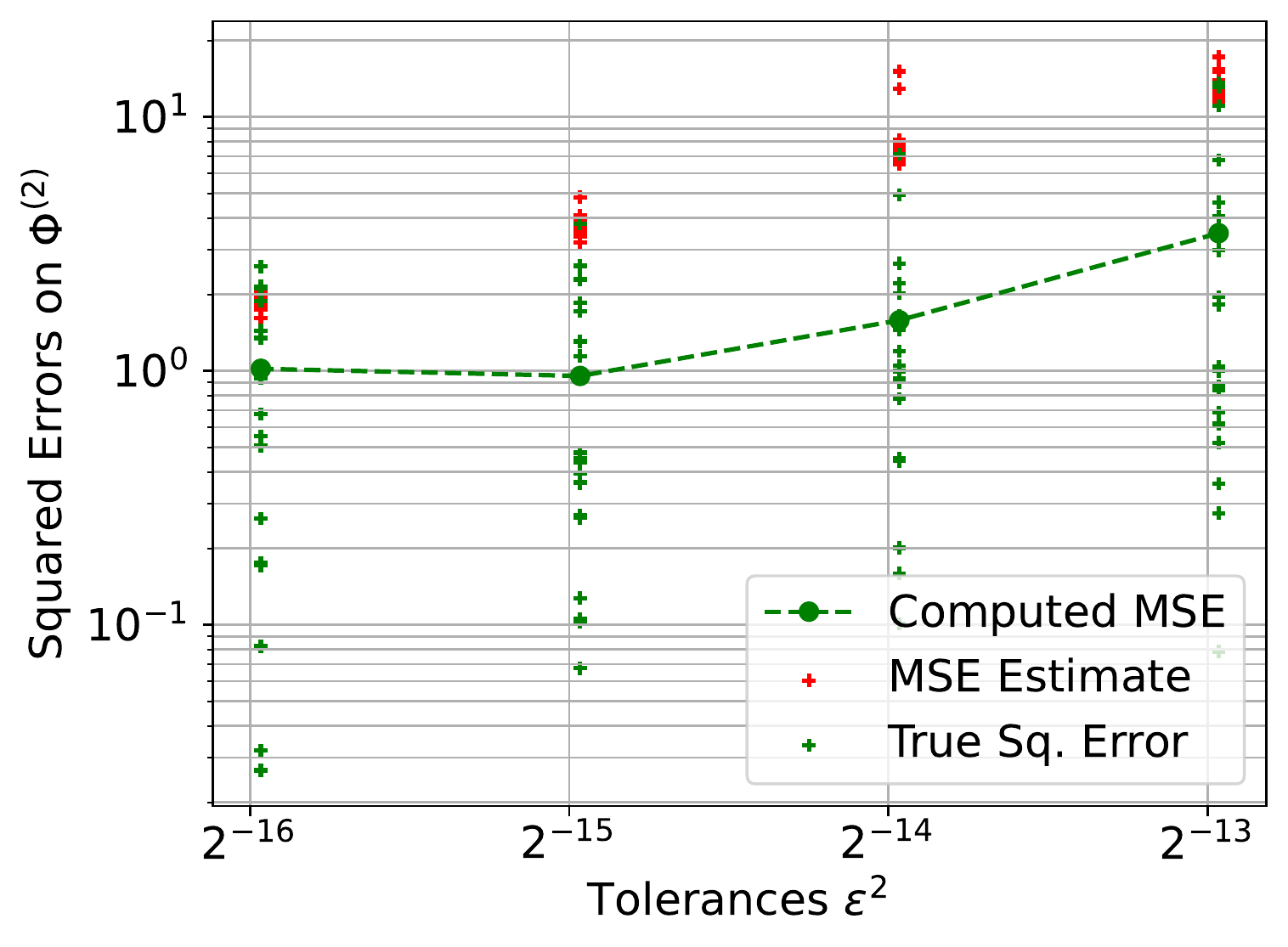}
    \caption{$m=2$}
\end{subfigure}
\begin{subfigure}{.48\textwidth}
    \centering
    \includegraphics[width=\textwidth]{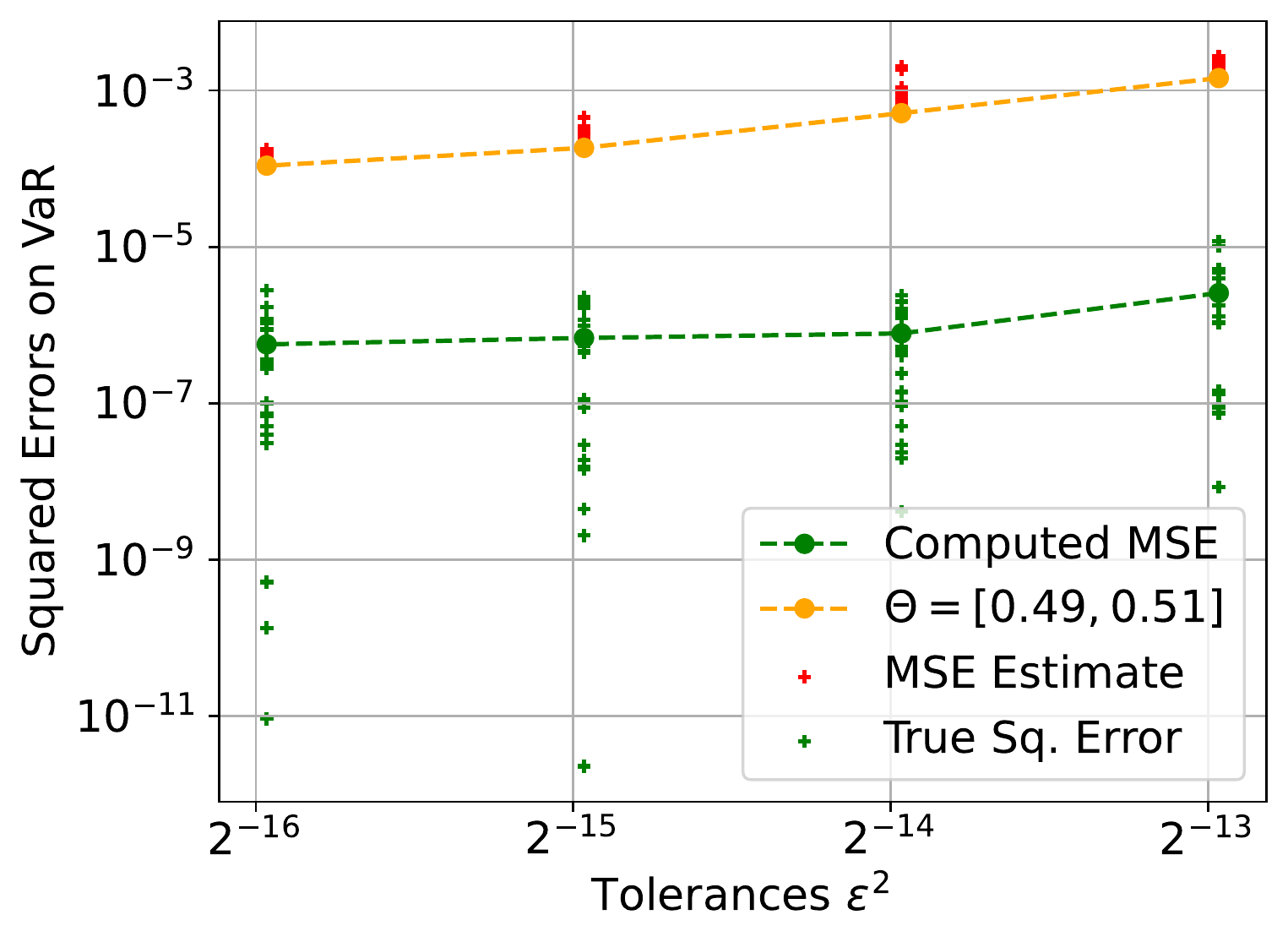}
    \caption{\gls{var}}
\end{subfigure}
\caption{Reliability of error estimator for $\Phim$ and \gls{var} for Navier Stokes problem}
\label{fig:phim_var_error_navier_stokes}
\end{figure}

Finally, we recall that the rescaling ratio for the Poisson problem was shown in Fig.~\ref{fig:ratio_study_poisson} to be relatively stable with respect to different interval sizes and hierarchy shapes.
However, we conducted a similar study for the Navier-Stokes problem and observed that the rescaling ratio was drastically more sensitive to the choice of these parameters than in the Poisson problem case.
The study is summarized in Fig.~\ref{fig:ratio_study_ns}; namely that we observe the behaviour of the variance rescaling ratio $r_e$ for different hierarchy shapes and interval sizes around the $70\%$-\gls{var}.
We conduct the study only for hierarchies as in Eq.~\eqref{eq:hierarchy} with $r=0$, that is, for a hierarchy with the same sample sizes across all levels. 
In contrast to the results of Fig.~\ref{fig:ratio_study_poisson}, we observe that the rescaling ratio very strongly depends on the hierarchy size $N_0$, as well as the interval $\Theta$, and increases significantly for smaller values of $N_0$ and shorter intervals $\Theta$.
These observations demonstrate the imperative need for an adaptive selection algorithm for the interval $\Theta$.
We plan to explore this direction in a future work. 

\begin{figure}[ht]
\centering
\includegraphics[width=0.6\textwidth]{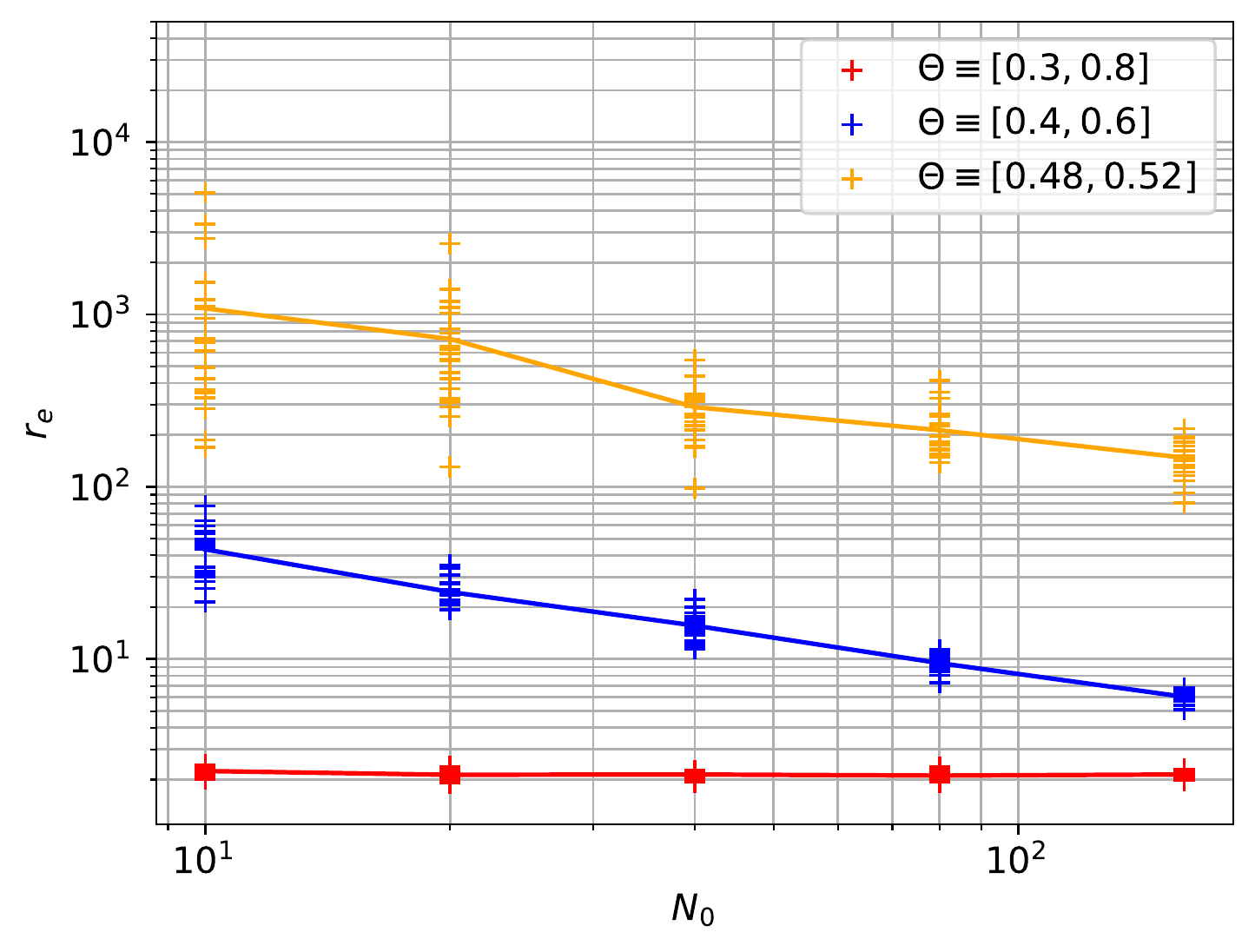}
 \caption{Behaviour of $r_e$ for different hierarchy shapes and interval sizes for Navier Stokes problem}
\label{fig:ratio_study_ns}
\end{figure}
\FloatBarrier

\section{Conclusions} \label{sec:conclusion}
The aim of this work was to tackle the problem of estimating summary statistics of a random \gls{qoi} which was an output of a complex differential model with random inputs. 
Namely, \gls{mlmc} estimators for the \gls{var}, the \gls{cvar}, the \gls{cdf} and the \gls{pdf} were proposed based on the concept of parametric expectations proposed in \cite{krumscheid2017multilevel}.
In this past theoretical work, a priori error estimates and complexity results were proposed for \gls{mlmc} estimators of parametric expectations, laying the foundation for the current work. 
However, the a priori estimates previously proposed were found to be highly conservative due to the presence of large leading constants and hence, practically unusable.

A completely practical modification was presented in this work by developing novel error estimators combined with an adaptive strategy for selecting the \gls{mlmc} hierarchy parameters and a \gls{cmlmc} framework for these summary statistics. 
The novel developments entail the following. 
Novel error estimators were presented for parametric expectations of the form in Eqs.~\eqref{eq:parest_form} and~\eqref{eq:phi_form} in Section~\ref{sec:novel}.
In Section~\ref{sec:adaptivity}, we have subsequently derived novel error estimators on the \gls{var} and the \gls{cvar} based on the novel error estimators on parametric expectations.
The error estimators presented in this work are an important improvement from the error bounds presented in \cite{krumscheid2017multilevel}; namely that they eliminate large leading constants that led to conservative error estimates while preserving decay rate properties important for the optimal performance of the \gls{mlmc} algorithm. 
Novel practical methods were presented for estimating the bias and statistical error components; the bias error is estimated using a \gls{kde}-smoothened density and the statistical error is estimated using bootstrapping and localised using rescaled local variances. 
Adaptive strategies were also presented for selecting the parameters of the \gls{mlmc} estimator for parametric expectations based on these error estimates. 
In particular, a \gls{cmlmc} algorithm was described to successively calibrate the \gls{mlmc} estimator on iteratively improved estimates of the errors. 
The above combination of error estimators, adaptive strategy and \gls{mlmc} algorithm were demonstrated on a simple problem whose analytical solution was known. 
It was shown that the error estimators provided practical error bounds on the true error and resulted in practically computable hierarchies for the test problems, ranging from a simple Poisson problem to the steady Navier-Stokes equations for flow over a cylinder, demonstrated in this study.

The numerical examples considered here indicate that the performance of the novel approach sensitively depends on the choice of interval over which to construct the parametric expectation.
It was shown that the choice of interval was important to the tightness of the novel error estimators for derived quantities such as the \gls{var} and the \gls{cvar}.
We plan to explore this and related improvements in future works.

\section*{Acknowledgments}
This project has received funding from the European Union's Horizon 2020 research and innovation programmed under grant agreement No. 800898. 

\begin{appendices}
\section{Spline Intepolator Property}\label{app:err_est}
We recall here basic results on error bounds for the use of cubic spline interpolation operators to approximate a function and its derivatives.
\begin{lemma}[Cubic spline interpolation operator]
\label{lemma:spline:prop}
Let $\interp{ \boldsymbol{f}(\thetab)} \in C^2(\Theta)$ be the cubic spline interpolation operator acting on the function values $\boldsymbol{f}(\thetab) \in \setR^n$ consisting of the function $f:\Theta \to \setR$ evaluated at the $n$ uniform nodes $\thetab = [\theta_1,...,\theta_n]^T$ such $[\theta_1 ,\theta_n] = \Theta$. 
The interpolation operator satisfies
\begin{enumerate}[label=\textbf{(S.\arabic*)}]
\item\label{spline:prop:deriv:error} for $m\in\{0,1,2\}$ and for any $f\in C^{4}(\Theta)$
\begin{equation*}
\norm{ f^{(m)} - \frac{d^m}{d\theta^m}\interp{ \boldsymbol{f}(\thetab)} }_{\linf(\Theta)} \leq C_1(m) \norm{ f^{(4)} }_{\linf(\Theta)} \left( \frac{|\Theta|}{n}\right)^{(4-m)}\;,
\end{equation*}
with $C_1(0) = 5/384$, $C_1(1) = 1/24$, and $C_1(2) = 3/8$,
\item\label{spline:prop:deriv:stable} for $m\in\{1,2\}$ and for any $\boldsymbol{x}\in\setR^n$
\begin{equation*}
\norm{\frac{d^m}{d\theta^m}\interp{\boldsymbol{x}}}_{\linf(\Theta)} \leq C_{2}(m) {(n-1)}^m \norm{ \interp{ \boldsymbol{x}}}_{\linf(\Theta)},
\end{equation*}
with $C_2(1) = 18/\left|\Theta\right|$ and $C_2(2) = 48/|\Theta|^2$,
\item \label{spline:prop:spline:stable} $\norm{ \interp{ \boldsymbol{x} }}_{\linf(\Theta)} \leq C_3 \norm{ \boldsymbol{x} }_{\slinf}$ for any $\boldsymbol{x}\in\setR^n$, with $C_3 = \frac{7(2\sqrt{7}+1)}{27}$,
\end{enumerate}
for all $n\in\setN$. 
\end{lemma}
\begin{proof}
The fact that $S := \interp{ \boldsymbol{f}(\thetab)}\in C^2(\Theta)$ as well as the properties \ref{spline:prop:deriv:error} and \ref{spline:prop:spline:stable} are well known results in approximation theory \cite{hall1976optimal,de1978practical, quarteroni2010numerical}.
To prove property \ref{spline:prop:deriv:stable}, we first note that
\begin{equation*}
\norm{S}_{\linf(\Theta)} = \max_{1\le j \le n-1}\norm{ P_j}_{\linf\left([\theta_j,\theta_{j+1}]\right)}=\max_{1\le j \le n-1}\norm{ P_j\circ g_j}_{\linf\left([-1,1]\right)}\;, 
\end{equation*}
where $g_j(t) = \frac{\theta_j+\theta_{j+1}}{2} + \frac{\delta}{2} t$ with $\delta := \theta_{j+1}-\theta_j = \frac{|\Theta|}{n-1}$. 
Here, $P_j$ denotes the cubic spline polynomial on the interval $[\theta_j, \theta_{j+1}]$. 
It then follows from the Markov type inequality result \cite{Govil1999} that
\begin{equation*}
\norm{S}_{\linf(\Theta)} \ge \frac{1}{9}\max_{1\le j \le n-1} \norm{\frac{d}{dt} P_j(g_j)}_{\linf\left([-1,1]\right)} = \frac{\delta}{18}\max_{1\le j \le n-1} \norm{ P_j^{(1)}}_{\linf\left([\theta_j,\theta_{j+1}]\right)}\;, 
\end{equation*}
which shows that
\begin{equation*}
\norm{ S^{(1)} }_{\linf(\Theta)}\le \frac{18}{\delta}\norm{S}_{\linf(\Theta)} = \frac{18(n-1)}{|\Theta|}\norm{S}_{\linf(\Theta)}\;.
\end{equation*}
An analogous analysis for the second derivative yields $\norm{ S^{(2)}}_{\linf(\Theta)}\le \frac{48}{\delta^2}\norm{ S}_{\linf(\Theta)}$, which completes the proof.
We recall that the constants $C_1(0)$ and $C_1(1)$ in \ref{spline:prop:deriv:error} above are known to be optimal \cite{hall1976optimal}.
\end{proof}
\end{appendices}

\bibliography{References}

\begin{bibdiv}
\begin{biblist}

\bib{ExaQUte_XMC}{misc}{
      author={Amela, Ramon},
      author={Ayoul-Guilmard, Quentin},
      author={Badia, Rosa~M},
      author={Ganesh, Sundar},
      author={Nobile, Fabio},
      author={Rossi, Riccardo},
      author={Tosi, Riccardo},
       title={Exa{QU}te {XMC}},
        type={software},
organization={Exa{QU}te consortium},
        date={2020},
}

\bib{MUMPS:2}{article}{
      author={Amestoy, P.R.},
      author={Buttari, A.},
      author={L'Excellent, J.-Y.},
      author={Mary, T.},
       title={{Performance and Scalability of the Block Low-Rank Multifrontal
  Factorization on Multicore Architectures}},
        date={2019},
     journal={ACM Transactions on Mathematical Software},
      volume={45},
       pages={2:1\ndash 2:26},
}

\bib{MUMPS:1}{article}{
      author={Amestoy, P.R.},
      author={Duff, I.~S.},
      author={Koster, J.},
      author={L'Excellent, J.-Y.},
       title={A fully asynchronous multifrontal solver using distributed
  dynamic scheduling},
        date={2001},
     journal={SIAM Journal on Matrix Analysis and Applications},
      volume={23},
      number={1},
       pages={15\ndash 41},
}

\bib{avikainen2009irregular}{article}{
      author={Avikainen, Rainer},
       title={On irregular functionals of sdes and the euler scheme},
        date={2009},
     journal={Finance and Stochastics},
      volume={13},
      number={3},
       pages={381\ndash 401},
}

\bib{babuvska2007stochastic}{article}{
      author={Babu{\v{s}}ka, Ivo},
      author={Nobile, Fabio},
      author={Tempone, Ra{\'u}l},
       title={A stochastic collocation method for elliptic partial differential
  equations with random input data},
        date={2007},
     journal={SIAM Journal on Numerical Analysis},
      volume={45},
      number={3},
       pages={1005\ndash 1034},
}

\bib{bayer2020multilevel}{article}{
      author={Bayer, Christian},
      author={Hammouda, Chiheb~Ben},
      author={Tempone, Ra{\'u}l},
       title={Multilevel monte carlo combined with numerical smoothing for
  robust and efficient option pricing and density estimation},
        date={2020},
     journal={arXiv preprint arXiv:2003.05708},
}

\bib{bierig2016approximation}{article}{
      author={Bierig, Claudio},
      author={Chernov, Alexey},
       title={Approximation of probability density functions by the multilevel
  monte carlo maximum entropy method},
        date={2016},
     journal={Journal of Computational Physics},
      volume={314},
       pages={661\ndash 681},
}

\bib{collier2014continuation}{article}{
      author={Collier, Nathan},
      author={Haji-Ali, Abdul-Lateef},
      author={Nobile, Fabio},
      author={von Schwerin, Erik},
      author={Tempone, Ra{\'u}l},
       title={A continuation {M}ultilevel {M}onte {C}arlo algorithm},
        date={2014},
     journal={BIT Numerical Mathematics},
       pages={1\ndash 34},
}

\bib{de1978practical}{book}{
      author={De~Boor, Carl},
       title={A practical guide to splines},
   publisher={Springer-Verlag New York},
        date={1978},
      volume={27},
}

\bib{ghanem2003stochastic}{book}{
      author={Ghanem, Roger~G},
      author={Spanos, Pol~D},
       title={Stochastic finite elements: a spectral approach},
   publisher={Courier Corporation},
        date={2003},
}

\bib{Giles2015a}{article}{
      author={Giles, M.~B.},
       title={Multilevel {M}onte {C}arlo methods},
        date={2015},
     journal={Acta Numer.},
      volume={24},
       pages={259\ndash 328},
}

\bib{giles2008multilevel}{article}{
      author={Giles, Michael~B},
       title={{M}ultilevel {M}onte {C}arlo path simulation},
        date={2008},
     journal={Operations Research},
      volume={56},
      number={3},
       pages={607\ndash 617},
}

\bib{giles2015multilevel}{article}{
      author={Giles, Michael~B},
       title={Multilevel monte carlo methods},
        date={2015},
     journal={Acta Numerica},
      volume={24},
       pages={259\ndash 328},
}

\bib{giles2019multilevel}{article}{
      author={Giles, Michael~B},
      author={Haji-Ali, Abdul-Lateef},
       title={Multilevel nested simulation for efficient risk estimation},
        date={2019},
     journal={SIAM/ASA Journal on Uncertainty Quantification},
      volume={7},
      number={2},
       pages={497\ndash 525},
}

\bib{giles2009analysing}{article}{
      author={Giles, Michael~B},
      author={Higham, Desmond~J},
      author={Mao, Xuerong},
       title={Analysing multi-level monte carlo for options with non-globally
  lipschitz payoff},
        date={2009},
     journal={Finance and Stochastics},
      volume={13},
      number={3},
       pages={403\ndash 413},
}

\bib{giles2015multilevelfunc}{article}{
      author={Giles, Michael~B},
      author={Nagapetyan, Tigran},
      author={Ritter, Klaus},
       title={Multilevel monte carlo approximation of distribution functions
  and densities},
        date={2015},
     journal={SIAM/ASA Journal on Uncertainty Quantification},
      volume={3},
      number={1},
       pages={267\ndash 295},
}

\bib{giles2017adaptive}{article}{
      author={Giles, Mike~B},
      author={Nagapetyan, Tigran},
      author={Ritter, Klaus},
       title={Adaptive multilevel monte carlo approximation of distribution
  functions},
        date={2017},
     journal={arXiv preprint arXiv:1706.06869},
}

\bib{gou2016estimating}{misc}{
      author={Gou, Wenhui},
       title={Estimating value-at-risk using multilevel monte carlo maximum
  entropy method},
   publisher={Oxford University MSc thesis},
        date={2016},
}

\bib{Govil1999}{article}{
      author={Govil, N.~K.},
      author={Mohapatra, R.~N.},
       title={Markov and {B}ernstein type inequalities for polynomials},
        date={1999},
     journal={J. Inequal. Appl.},
      volume={3},
      number={4},
       pages={349\ndash 387},
}

\bib{hall1976optimal}{article}{
      author={Hall, Charles~A},
      author={Meyer, W~Weston},
       title={Optimal error bounds for cubic spline interpolation},
        date={1976},
     journal={Journal of Approximation Theory},
      volume={16},
      number={2},
       pages={105\ndash 122},
}

\bib{heinrich2001multilevel}{inproceedings}{
      author={Heinrich, Stefan},
       title={Multilevel monte carlo methods},
organization={Springer},
        date={(2001)},
   booktitle={International conference on large-scale scientific computing},
       pages={58\ndash 67},
}

\bib{krumscheid2017multilevel}{article}{
      author={Krumscheid, Sebastian},
      author={Nobile, Fabio},
       title={Multilevel monte carlo approximation of functions},
        date={2018},
     journal={SIAM/ASA Journal on Uncertainty Quantification},
      volume={6},
      number={3},
       pages={1256\ndash 1293},
}

\bib{Pisaroni_mlmcPart1_pre}{article}{
      author={Krumscheid, Sebastian},
      author={Nobile, Fabio},
      author={Pisaroni, M},
       title={Quantifying uncertain system outputs via the multilevel monte
  carlo method—part i: Central moment estimation},
        date={2020},
     journal={Journal of Computational Physics},
      volume={414},
       pages={109466},
}

\bib{le2010spectral}{book}{
      author={Le~Ma{\^\i}tre, Olivier},
      author={Knio, Omar~M},
       title={Spectral methods for uncertainty quantification: with
  applications to computational fluid dynamics},
   publisher={Springer Science \& Business Media},
        date={2010},
}

\bib{linde1975mappings}{article}{
      author={Linde, W},
      author={Pietsch, A},
       title={Mappings of gaussian cylindrical measures in banach spaces},
        date={1975},
     journal={Theory of Probability \& Its Applications},
      volume={19},
      number={3},
       pages={445\ndash 460},
}

\bib{logg2012automated}{book}{
      author={Logg, Anders},
      author={Mardal, Kent-Andre},
      author={Wells, Garth},
       title={Automated solution of differential equations by the finite
  element method: The fenics book},
   publisher={Springer Science \& Business Media},
        date={2012},
      volume={84},
}

\bib{pisaroni2016continuation}{article}{
      author={Pisaroni, Michele},
      author={Nobile, Fabio},
      author={Leyland, P{\'e}n{\'e}lope},
       title={A continuation multi level monte carlo (c-mlmc) method for
  uncertainty quantification in compressible inviscid aerodynamics},
        date={2017},
     journal={Computer Methods in Applied Mechanics and Engineering},
      volume={326},
       pages={20\ndash 50},
}

\bib{quarteroni2010numerical}{book}{
      author={Quarteroni, Alfio},
      author={Sacco, Riccardo},
      author={Saleri, Fausto},
       title={Numerical mathematics},
   publisher={Springer Science \& Business Media},
        date={(2010)},
      volume={37},
}

\bib{rockafellar2002conditional}{article}{
      author={Rockafellar, R~Tyrrell},
      author={Uryasev, Stanislav},
       title={Conditional value-at-risk for general loss distributions},
        date={2002},
     journal={Journal of banking \& finance},
      volume={26},
      number={7},
       pages={1443\ndash 1471},
}

\bib{sagiv2021spectral}{article}{
      author={Sagiv, Amir},
       title={Spectral convergence of probability densities for forward
  problems in uncertainty quantification},
        date={2021},
     journal={arXiv preprint arXiv:2101.07395},
}

\bib{scott1979optimal}{article}{
      author={Scott, David~W},
       title={On optimal and data-based histograms},
        date={1979},
     journal={Biometrika},
      volume={66},
      number={3},
       pages={605\ndash 610},
}

\bib{taverniers2020estimation}{article}{
      author={Taverniers, S{\o}ren},
      author={Tartakovsky, Daniel~M},
       title={Estimation of distributions via multilevel monte carlo with
  stratified sampling},
        date={2020},
     journal={Journal of Computational Physics},
      volume={419},
       pages={109572},
}

\bib{tibshirani1993introduction}{article}{
      author={Tibshirani, Robert~J},
      author={Efron, Bradley},
       title={An introduction to the bootstrap},
        date={1993},
     journal={Monographs on statistics and applied probability},
      volume={57},
       pages={1\ndash 436},
}

\bib{uryasev2001conditional}{incollection}{
      author={Uryasev, Stanislav},
      author={Rockafellar, R~Tyrrell},
       title={Conditional value-at-risk: optimization approach},
        date={2001},
   booktitle={Stochastic optimization: algorithms and applications},
   publisher={Springer},
       pages={411\ndash 435},
}

\bib{xiu2002wiener}{article}{
      author={Xiu, Dongbin},
      author={Karniadakis, George~Em},
       title={The wiener--askey polynomial chaos for stochastic differential
  equations},
        date={2002},
     journal={SIAM journal on scientific computing},
      volume={24},
      number={2},
       pages={619\ndash 644},
}

\end{biblist}
\end{bibdiv}

\bibliographystyle{amsplain} 

\end{document}